\documentclass[a4paper,11pt]{article}

\usepackage[margin=0.8in]{geometry}  %
\usepackage[english]{babel}
\usepackage[utf8x]{inputenc}
\usepackage[T1]{fontenc}
\usepackage{parskip}
\usepackage{setspace}
\usepackage{fancyhdr}
\usepackage[dvipsnames]{xcolor}
\usepackage{graphicx}
\usepackage[font=normalsize,labelfont=bf]{caption}
\usepackage{subcaption}
\usepackage{longtable} 
\usepackage{multirow}
\usepackage{listings}
\usepackage{makecell}
\usepackage{array}
\usepackage{float}
\usepackage{placeins}
\usepackage{dsfont}
\usepackage{rotating}
\usepackage{booktabs}
\usepackage{enumerate}
\usepackage{tikz}
\usetikzlibrary{chains}
\usetikzlibrary{arrows}
\usepackage{pgf}

\usepackage{apptools}
\AtAppendix{\counterwithin{theorem}{section}}
\AtAppendix{\counterwithin{definition}{section}}
\AtAppendix{\counterwithin{proposition}{section}}
\AtAppendix{\counterwithin{lemma}{section}}
\AtAppendix{\counterwithin{assumption}{section}}
\AtAppendix{\counterwithin{algocf}{section}}

\usepackage{amsmath}
\usepackage{amssymb}
\usepackage{amsthm}
\usepackage{bm}
\usepackage[linesnumbered,ruled]{algorithm2e}
\usepackage{algorithmicx}
\usepackage{mathtools}
\usepackage{mathrsfs}

\usepackage{natbib}
\usepackage{bibunits}
\usepackage[
  colorlinks,
  citecolor=blue,
  linkcolor=blue,
  anchorcolor=blue,
  urlcolor=blue
]{hyperref}
\mathtoolsset{showonlyrefs}
\usepackage{authblk}
\usepackage{todonotes}
\defaultbibliographystyle{apalike}
\defaultbibliography{ref}
\theoremstyle{plain}

\newtheorem{definition}{Definition}
\newtheorem{theorem}{Theorem}
\newtheorem{lemma}{Lemma}

\newtheorem{corollary}{Corollary}
\newtheorem{remark}{Remark}

\newtheorem{proposition}{Proposition}

\newcommand{\RR}{\mathbb{R}}
\newcommand{\ZZ}{\mathbb{Z}}
\newcommand{\NN}{\mathbb{N}}
\newcommand{\EE}{\mathbb{E}}
\newcommand{\PP}{\mathbb{P}}
\newcommand{\ind}{\mathds{1}}

\newcommand{\cov}{\mathrm{Cov}}
\newcommand{\indep}{\perp\!\!\!\perp}

\newcommand{\Bern}{\text{Bern}}

\newcommand{\fdr}{\textnormal{FDR}}
\newcommand{\fdp}{\textnormal{FDP}}
\newcommand{\indc}[1]{\mathds{1}\left\{#1\right\}}

\newcommand{\eqd}{\stackrel{\textnormal{d}}{=}}

\usepackage{xspace}

\newcommand{\stepa}[1]{\overset{\rm (a)}{#1}}
\newcommand{\stepb}[1]{\overset{\rm (b)}{#1}}
\newcommand{\stepc}[1]{\overset{\rm (c)}{#1}}
\newcommand{\stepd}[1]{\overset{\rm (d)}{#1}}

\providecommand{\given}{}
\renewcommand{\given}{{\,|\,}}
\newcommand{\biggiven}{\,\big{|}\,}
\newcommand{\Biggiven}{\,\Big{|}\,}
\newcommand{\bigggiven}{\,\bigg{|}\,}
\newcommand{\Bigggiven}{\,\Bigg{|}\,}
\def\@#1\@{\begin{align}#1\end{align}}
\def\$#1\${\begin{align*}#1\end{align*}}

\definecolor{myblue}{rgb}{.8, .8, 1}
\definecolor{mathblue}{rgb}{0.2472, 0.24, 0.6} %
\definecolor{mathred}{rgb}{0.6, 0.24, 0.442893}
\definecolor{mathyellow}{rgb}{0.6, 0.547014, 0.24}

\newcommand{\cond}{\mid}

\newcommand{\te}{{\tilde{e}}}

\newcommand{\tD}{{\tilde{D}}}

\newcommand{\tT}{{\tilde{T}}}

\newcommand{\tX}{{\tilde{X}}}
\newcommand{\tY}{{\tilde{Y}}}
\newcommand{\tZ}{{\tilde{Z}}}

\newcommand{\cD}{{\mathcal{D}}}
\newcommand{\cE}{{\mathcal{E}}}
\newcommand{\cF}{{\mathcal{F}}}

\newcommand{\cH}{{\mathcal{H}}}
\newcommand{\cI}{{\mathcal{I}}}

\newcommand{\cN}{{\mathcal{N}}}

\newcommand{\cR}{{\mathcal{R}}}
\newcommand{\cS}{{\mathcal{S}}}

\newcommand{\cV}{{\mathcal{V}}}
\newcommand{\cW}{{\mathcal{W}}}
\newcommand{\cX}{{\mathcal{X}}}
\newcommand{\cY}{{\mathcal{Y}}}
\newcommand{\cZ}{{\mathcal{Z}}}

\renewcommand{\hat}{\widehat}
\renewcommand{\tilde}{\widetilde}

\newcommand{\be}{{\bm e}}
\newcommand{\bp}{{\bm p}}
\newcommand{\bh}{\textnormal{BH}}
\newcommand{\ebh}{\textnormal{e-BH}}
\newcommand{\wcs}{\mathrm{WCS}}
\newcommand{\boost}{\mathfrak{b}} 
\newcommand{\calib}{\textnormal{calib}}
\newcommand{\test}{\textnormal{test}}
\newcommand{\dif}{\mathrm{d}}
\newcommand{\kn}{\mathrm{kn}}
\newcommand{\dbh}{\textnormal{dBH}}
\newcommand{\hcR}{\hat{\cR}}
\newcommand{\ci}{\mathrm{CI}}
\newcommand{\avcs}{\mathrm{AVCS}}

\newcommand{\train}{\textnormal{train}}
\newcommand{\cc}{\textnormal{CC}}

\graphicspath{figures-and-tables/}
\long\def\comment#1{}
\newcommand{\red}[1]{{#1}}

\title{Boosting e-BH via conditional calibration}
\author[1]{Junu Lee\thanks{Corresponding author: \href{mailto:junulee@wharton.upenn.edu}{\nolinkurl{junulee@wharton.upenn.edu}}}}
\author[1]{Zhimei Ren}
\affil[1]{{\normalsize Department of Statistics and Data Science, University of Pennsylvania, USA}} 
\date{}
\begin{document}
\begin{bibunit}

\maketitle

\begin{abstract}
    \normalsize
    The e-BH procedure is an e-value-based multiple testing procedure that provably 
    controls the false discovery rate (FDR) under any dependence structure 
    between the e-values. Despite this appealing theoretical FDR control guarantee, 
    the e-BH procedure often suffers from low power in practice. 
    In this paper, we propose a general framework that boosts the power of 
    e-BH without sacrificing its FDR control under arbitrary dependence.
    This is achieved by the technique of conditional calibration, 
    where we take as input the e-values and calibrate them to be a 
    set of ``boosted e-values'' that are guaranteed to 
    be no less---and are often more---powerful than the original ones.
    Our general framework is explicitly instantiated 
    in three classes of multiple testing problems: 
    (1) testing under parametric models, 
    (2) conditional independence testing under the model-X setting, and 
    (3) model-free conformalized selection. Extensive numerical experiments 
    show that our proposed method significantly improves the power of 
    e-BH while continuing to control the FDR.
    We also demonstrate the effectiveness of our method through an application to an 
    observational study dataset for identifying individuals 
    whose counterfactuals satisfy certain properties. 
  \end{abstract}

\section{Introduction}
\label{sec:intro}
We study the problem of testing $m$ null hypotheses simultaneously while 
controlling the false discovery rate (FDR)~\citep{benjamini1995controlling}. 
For this task, classical methods (e.g.,~\citet{benjamini1995controlling,
benjamini2001control,storey2002direct}) associate each hypothesis with a 
p-value and select a subset of hypotheses to reject based on these
p-values. Recently, the notion of e-values has been proposed for quantifying evidence
against the null hypothesis in place of 
p-values~\citep{shafer2011test,grunwald2020safe,vovk2021values,grunwald2023posterior}.
To be concrete, an {\em e-value} for a null hypothesis $H_0$ is the realization 
of an {\em e-variable} $E$, which obeys 
$E\ge 0$ almost surely and $\EE_{H_0}[E]\le 1$.
In contrast, we recall that a {\em p-value} for $H_0$ is the realization of a 
{\em p-variable} $P$, such that 
$\PP_{H_0}(P \le t) \le t$, for all $t \in (0,1).$
In what follows, we will not distinguish between e-values (resp.~p-values) 
and e-variables (resp.~p-variables) when the context is clear.
Per their definitions, both e-values and p-values are summaries of evidence against $H_0$, 
where we reject $H_0$ for small p-values or large e-values.
Compared to p-values, e-values have several properties that make them attractive 
for hypothesis testing. For example, the e-value allows the experimenter 
to adaptively decide whether to collect new evidence or to stop the experiment.
It is also convenient for combining evidence from multiple sources (see 
Section~\ref{sec:background} for more discussion). 

When it comes to 
multiple hypothesis testing,~\citet{wang2022false} propose a simple and elegant 
e-value-based procedure, 
called the e-BH procedure, that provably controls the FDR at the desired level 
under unknown and arbitrary dependence among the e-values. This is a rather surprising result, 
since for  
p-value-based procedures the FDR control is only guaranteed under 
special dependency 
structures---e.g., when the p-values are independent or positively correlated---unless 
one is willing to tolerate an inflated FDR level. 
Despite this theoretical appeal, the e-BH procedure is observed to 
be conservative in practice: it often achieves an FDR much lower than
the target level, which affects its power and greatly hinders its wide application. 
It is thus of great interest to improve the power of e-BH 
without sacrificing the FDR control guarantee. 

\subsection{A peek at our contribution}

In this paper, we propose a general framework that boosts the power of 
e-BH for a wide class of multiple testing problems when partial information 
on the dependence structure is accessible. To set the stage, 
consider $m$ null hypotheses $H_1,H_2,\ldots,H_m$, where the subset of the \red{$m_0$ true
null hypotheses} is denoted by $\cH_0$. 
For each $j \in [m] := \{1,2,\ldots,m\}$, 
$H_j$ is associated with an e-value $e_j$. Applying the e-BH procedure to 
$\be := (e_1,e_2,\ldots,e_m)$, one obtains a rejection set $\cR(\be) \subseteq 
[m]$. 
e-BH is guaranteed to control the 
FDR, which is the expected false discovery proportion (FDP) of
the algorithm:
\begin{align}
  \begin{split}
    \label{eq:fdr_defn}
    \fdr &= \EE[\fdp],\quad 
    \fdp = \frac{\sum_{j \in \cH_0} 
    \ind\{j \in \cR(\be)\}}{|\cR(\be)|\vee 1}.
  \end{split}
\end{align}
We slightly abuse notation by using $\cH_0$ to also denote the indices of the true null hypotheses. We will do the same for other collections of hypotheses when the context is obvious.

To identify the source of conservativeness in the e-BH procedure,  
we
decompose the FDR into the 
``contribution'' of each null hypothesis:
\$
\fdr= \sum_{j\in \cH_0}\fdr_j := \sum_{j\in \cH_0}
\EE\bigg[\frac{  
\ind\{j \in \cR(\be)\}}{|\cR(\be)|\vee 1}\bigg].
\$ 
Intuitively, if a multiple testing procedure is efficient---i.e.,  
it uses up all its FDR budget---then $\fdr_j$ should be close to 
$\alpha/m$ (or sum to $\alpha$).  As we shall see later, the $\fdr_j$ are often much smaller than 
$\alpha/m$ for e-BH, resulting in lack of power. 

\begin{figure}[t]
  \centering
  \includegraphics[width = \textwidth]{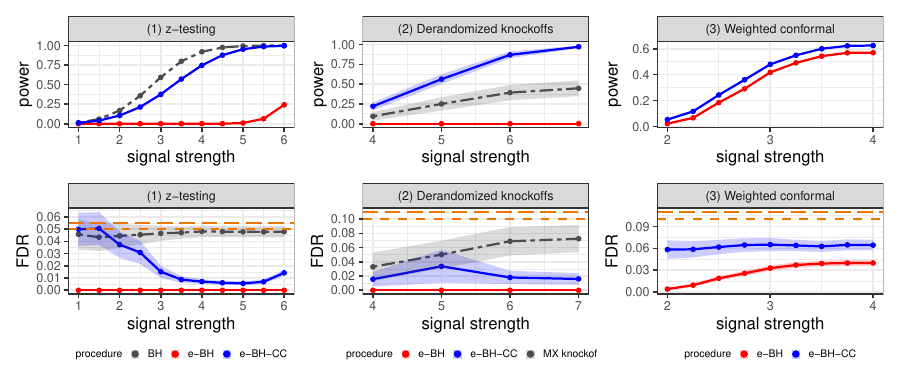} 
  \caption{{A selection of simulation results from each of the  instantiations of our proposed method, 
  e-BH-CC. The three problem instances are (1) one-sided $z$-testing using likelihood ratio e-values; 
  (2) conditional independence testing using derandomized knockoffs e-values; 
  and (3) outlier detection using weighted conformal e-values. 
  The plots above show mean power and false discovery proportion (FDP) curves over 1,000, $100$, 
  and 1,000 replications, respectively. 
  The dashed gray lines correspond to the relevant baselines (if they exist),  
  and the target FDR level is illustrated by the short-dashed orange line.
  Shading represents error bars, when they are deemed necessary.}} 
  \label{fig:signature}
  \centering
\end{figure}

At a high level, our 
proposal is to boost the e-BH procedure by filling these gaps, which requires 
identifying for each $j\in [m]$ a sufficient statistic $S_j$ for which
we can evaluate the conditional distribution of $\be\given S_j$ under the null $H_j$. 
The proposed method takes the e-values $\be$ and 
returns a set of {\em boosted e-values} $\be^\boost := (e_1^\boost,e_2^\boost,\ldots,
e_m^\boost)$ constructed with this distributional information 
that are at least as powerful as the original e-values; 
it then applies the e-BH procedure to
$\be^\boost$ to obtain a rejection set $\cR(\be^\boost)$, which 
dominates $\cR(\be)$ in power while maintaining 
FDR control.

Our second major contribution is to explicitly instantiate our general framework 
in three classes of multiple testing problems: (1) Gaussian mean testing, (2) conditional independence testing in the model-X setting, 
and (3) model-free conformal selection.
For each problem, we identify the sufficient statistics and provide a concrete way to 
boost the e-values. As a preview, Figure \ref{fig:signature}
shows the empirical power improvement of our proposed method over the e-BH procedure for a selection of experiments;
in all settings, it continues to both 
theoretically and empirically control the FDR at a preset level.
In the respective simulation studies of Sections \ref{sec:example-parametric}, \ref{sec:example-knockoffs}, and \ref{sec:example-conformal}, 
we find that power improvement occurs over all settings, not just the specific ones chosen to be presented in the figure.

\paragraph{Organization of the paper.}
We introduce the background on e-values and discuss related literature 
in Section~\ref{sec:background}. Our main boosting framework 
is presented in Section~\ref{sec:boosting}. 
Sections~\ref{sec:example-parametric},~\ref{sec:example-knockoffs} and~\ref{sec:example-conformal} each
contain an instantiation of our method in testing under parametric models, 
conditional independence testing in the model-X setting, and model-free conformalized 
selection, respectively. 
Section~\ref{sec:realdata} presents the results from a real data analysis,
and Section~\ref{sec:discussion} contains a  discussion.

\paragraph{Reproducibility.}
The code for reproducing the numerical experiments and real data analysis 
in this paper is available at \url{https://github.com/leejunu/e-bh-cc}.

\section{Background}
\label{sec:background}
\subsection{E-values}
The notion of e-values encompasses many commonly used statistics, such 
as betting scores, Bayes factors, likelihood ratios, and 
stopped martingales~\citep{wasserman2020universal,shafer2021testing,
howard2021confidence,grunwald2023posterior,waudby2024estimating}.
As mentioned earlier, we reject the null hypothesis $H_0$ when the e-value is large.
In particular, for any $\alpha \in (0,1)$, rejecting $H_0$ when $e \ge 1/\alpha$ yields a level $\alpha$ 
test as a consequence of Markov's inequality: 
$ 
\PP_{H_0}(e \ge 1/\alpha) \le \alpha \cdot \EE_{H_0}[e] \le \alpha.
$
There is also a close connection between the e-value and the p-value. 
Let $e$ be an e-value for a null hypothesis $H_0$. Then, 
$p = 1/e$ is a p-value
\footnote{
We use 
the term ``p-value'' a bit loosely since we allow the p-value to be larger than $1$. One 
can always transform $p$ into a strict p-value by taking $\min(p,1)$.
}
for $H_0$ by Markov's inequality, since
$\PP_{H_0}(p \le t) = \PP_{H_0}(e \ge 1/t) \le t \cdot \EE_{H_0}[e] \le t$, 
for any $t\in (0,1)$.
Conversely, a p-value $p$ for $H_0$ can be transformed into an e-value using a ``p-to-e calibrator''~\citep{shafer2011test}, 
defined as a decreasing function $f:[0,1] \to [0,\infty)$, such that 
$ 
\int_0^1 f(t) \, dt = 1.
$
For example, we can take $f(t) = \lambda t^{\lambda - 1}$ for some $\lambda \in (0,1)$ to be the 
calibrator~\citep{shafer2021testing,vovk2021values}.
Then, for any p-value $p$, $f(p)$ is an e-value.

Prior works have discussed 
when one would prefer e-values over p-values~\citep{shafer2021testing,grunwald2020safe,wang2022false,grunwald2023posterior,ramdas2023game,ramdas2025hypothesis}, 
and we refer readers to them for 
a comprehensive review.

\subsection{The e-BH procedure}
\label{sec:ebh}
For $m$ null hypotheses $H_1,\ldots,H_m$ and their associated e-values 
$\bm e \coloneqq (e_1,\ldots,e_m)$,
let $e_{(1)}, \ldots,e_{(m)}$ denote the ordered e-values in descending order.
The e-BH procedure at level $\alpha$ 
rejects the hypotheses corresponding to the $k^*$ largest e-values, where
$ 
k^* = \max\Big\{k \in [m]: e_{(k)} \ge \frac{m}{\alpha k} \Big\},
$
with the convention $\max \varnothing = 0$.
To understand why e-BH is conservative, it is helpful to go through 
the proof of its FDR control. Let $\cR^{\ebh}(\be)$ denote the e-BH rejection set. 
It can be checked that $j \in \cR^\ebh(\be)$ if and only if $e_j \ge \frac{m}{\alpha |\cR^\ebh(\be)|}$,
so we can write the FDR of e-BH as
\begin{equation}\label{eq:fdr_ineq}
    \fdr = %
    \sum_{j\in \cH_0}\EE\bigg[\frac{\ind\{e_j \ge \frac{m}{\alpha |\cR^\ebh(\be)|}\}}
    {|\cR^\ebh(\be)| \vee 1}\bigg] 
    \stepa{\le} \sum_{j\in \cH_0}\EE \bigg[ 
    \frac{e_j \frac{\alpha |\cR^\ebh(\be)|}{m}}{|\cR^\ebh(\be)|\vee 1} 
    \bigg] \le \frac{\alpha}{m}\sum_{j \in \cH_0}\EE[e_j]
    \stepb{\le} \alpha,
\end{equation} 
where step (a) follows from the deterministic 
inequality $\ind\{X\ge t\} \le X/t$ for any $t>0$, 
and step (b) is due to the definition of e-values. 
Note that the inequality in (a) 
is tight if and only if 
\@\label{eq:sharp-e-val}
\PP\Big(e_j \in \Big\{0,\frac{m}{\alpha |\cR^\ebh(\be)|}\Big\}\Big) = 1,~\forall j\in \cH_0.
\@
Step (a) therefore constitutes a major source of the gap when~\eqref{eq:sharp-e-val} does not hold.
This motivates us to study e-values satisfying~\eqref{eq:sharp-e-val}, 
which we formally define as {\em BH-sharp e-values}.
\begin{definition}[BH-sharp e-values]
Fix null hypotheses $H_1,\ldots,H_m$ and target FDR level $\alpha \in (0,1)$. 
We say $(e_1,\ldots,e_m)$ are BH-sharp e-values for the null hypotheses at level $\alpha$ if 
(a) for any $j\in[m]$, $e_j$ is a valid e-value for $H_j$ and (b)
$(e_1,\ldots,e_m)$ satisfy the condition in~\eqref{eq:sharp-e-val}. 
\end{definition}
BH-sharpness is a joint property of the e-value vector, rather than a marginal property 
of the individual e-values, because $\cR^\ebh(\be)$ depends on the full 
vector $(e_1,\ldots,e_m)$ and the BH-sharp condition 
induces a strong negative dependence structure among 
the e-values. Also, BH-sharp e-values are not necessarily binary random variables:
given $\cR^\ebh(\be)$, a null e-value can only take two values, but the nonzero value  
$m/(\alpha |\cR^\ebh(\be)|)$ is itself random.
When $(e_1,\ldots,e_m)$ are sharp e-values, we can express the FDR of e-BH 
exactly as the sum of null e-value expectations.
\begin{proposition}\label{prop:sharp-eval-fdr}
Fix null hypotheses $H_1,\ldots,H_m$ and target FDR level $\alpha \in (0,1)$.
Suppose $e_j$ is a valid e-value for $H_j$, $\forall j \in [m]$. The e-values 
$(e_1,\ldots,e_m)$ are BH-sharp if and only if 
\$ 
\fdr(\cR^\ebh(\be)) = \frac{\alpha}{m}\sum_{j\in \cH_0} \EE[e_j].
\$
\end{proposition}
The proof is given in Appendix~\ref{app:proof-of-sharp-eval-fdr}.
Proposition~\ref{prop:sharp-eval-fdr} shows that BH-sharpness removes the slack in step (a).
However, this alone does not guarantee that the e-BH procedure uses the available FDR budget
efficiently---the procedure can still be conservative if the null e-values have expectations 
substantially below 1. Thus, among BH-sharp e-values, it is desirable for the null expectations 
to be as close to one as possible.

In fact, if we denote $e_j' = \frac{m}{\alpha}\cdot 
\frac{\ind\{e_j \ge \frac{m}{\alpha |\cR^\ebh(\be)|}\}}{|\cR^\ebh(\be)| \vee 1} $,
then we can see $(e_1', \dots, e_m')$ are BH-sharp e-values such that
$\cR^\ebh(e_1', \dots, e_m') = \cR^\ebh(e_1, \dots, e_m)$; i.e., they
underlie the e-BH procedure on the original e-values. 
However, $\EE_{H_j}[e_j']$ can be much smaller than  $\EE_{H_j}[e_j] \le 1$, 
as suggested by~\eqref{eq:fdr_ineq},
which tells us the true FDR can be quite below the target (Proposition~\ref{prop:sharp-eval-fdr}).
The main idea behind our proposal, detailed in Section~\ref{sec:boosting}, 
is to increase the expectation of these BH-sharp e-values  
by incorporating available distributional information---specifically, 
leveraging the distribution of e-values conditional on available sufficient statistics.

Finally, we note that for the purpose of FDR control, it suffices to require  
that the sum of the expectations of null e-values is bounded by $m$ (step (b) in~\eqref{eq:fdr_ineq}).
Following \citet{ignatiadis2024compound}, we call such e-values \emph{compound e-values}, whose formal definition is given below.
\begin{definition}[Compound e-values]
    \label{def:compound_eval}
The non-negative random variables $e_1,e_2,\ldots,e_m$ are called compound e-values
if $\sum_{j \in \cH_0}\EE[e_j] \le m$.
\end{definition}
The e-BH procedure can be applied to compound e-values to achieve
level-$\alpha$ FDR control~\citep{wang2022false}; this can be shown 
by noting that step (b) in Equation~\eqref{eq:fdr_ineq} still holds for compound e-values.
We will see examples of such compound e-values in Section \ref{sec:example-knockoffs}.

Recent literature has focused on the utility of e-values for
multiple tasks, including multiple hypothesis testing.
For the sake of brevity, a detailed account of related work is deferred to Appendix \ref{appd:related_work}. Below, we discuss the most relevant attempt at improving the power of e-BH.

\paragraph{Previous attempts at addressing power concerns.}
\label{sec:marginal_boosting}
\citet*{wang2022false}, after introducing the e-BH procedure, also provide a method for 
boosting the e-values using available information 
on their {\em marginal} distributions.\footnote{
    \label{foot:PRDS}
\citet{wang2022false}
also provide a more powerful boosting scheme when the p-values are 
positively dependent on a subset (PRDS). Since the PRDS condition already 
ensures the FDR control of the BH procedure, in this paper we choose to 
focus primarily on the more general cases beyond the PRDS condition.
} 
Specifically, the method transforms an e-value $e_j$ into $e_j' = b_je_j$
for some boosting factor $b_j \ge 1$, then applies the e-BH procedure
to $\{e_j'\}_{j\in[m]}$. The boosting factor $b_j$ is chosen to be the largest
$b\ge 1$ such that $\EE\big[T(\alpha b e_j)\big] \le \alpha$, where 
$T(x)$ is the largest element in $\{0,1, \frac{m}{m-1},\ldots, \frac m2,m\}$ that is 
no greater than $x$. 
For example, when $e_j = \frac{1}{2}p_j^{-1/2}$, with $p_j\sim\mathrm{Unif([0,1])}$, one can take $b_j = (2/\alpha)^{1/2}$.
Since this boosting scheme mainly leverages the marginal distribution of e-values, 
we refer to it as the \emph{marginal boosting} scheme. 
Generally, it is not sufficient for closing the gap caused by step (a) in 
Equation~\eqref{eq:fdr_ineq}. Appendix \ref{appd:marg_boost_not_enough} covers 
numerical experiments which detail a comparison between marginal boosting and our 
boosting proposal (to be introduced in Section \ref{sec:boosting}); those results suggest that
our method generally dominates marginal boosting in power, and can be used on top of marginal boosting
to attain even higher power.

\section{Boosting e-BH via conditional calibration}
\label{sec:boosting}

Given a collection of e-values $\bm e = (e_1, \dots, e_m)$ corresponding to null hypotheses $H_1, \dots, H_m$, 
our method returns a new collection of e-values $\bm e^\boost = (e^\boost_1, \dots, e^\boost_m)$ through a 
technique called \emph{conditional calibration} (CC)~\citep{fithian2020conditional}. 
Originally proposed as a method to produce separately-calibrated thresholds for \textit{p-values} 
to return an FDR-controlling rejection set, 
\red{
we use conditional calibration to deterministically (given the data)}
boost the e-values to more powerful versions of themselves while retaining e-value validity.
\red{
Applying e-BH  to these boosted e-values then yields a rejection set that controls the FDR at
the nominal level and contains the rejection set obtained from the original e-values.
}

\red{To motivate our construction, recall the inequalities (a) and (b) in the FDR control of $\cR\coloneqq \cR^\ebh$, provided in~\eqref{eq:fdr_ineq}.  
From our analysis in Section~\ref{sec:ebh}, for any fixed level $\alpha \in (0,1)$, $\cR(\bm e) = \cR(\bm e')$, where $\bm e' = (e_1', \dots, e_m')$ is the BH-sharp representation at that level.
For the original e-values, (a) is usually very loose, while (b) is often sharp (up to the $m_0/m$ factor); for the BH-sharp e-values, the opposite holds true, and the conservativeness of e-BH is hence exhibited in step (b) when working with the BH-sharp representation.
This suggests the approach of improving the power by uniformly
increasing the expectation of $\bm e'$ toward that of  $\bm e$ by expanding the events on which the $e_j'$ is nonzero, as discussed in Section~\ref{sec:ebh}.

With this goal in mind, we propose a framework   to leverage conditional distribution information to close the expectation gap between each $e_j'$ and $e_j$ via conditional calibration.}
Suppose we can identify for each $j\in [m]$ a sufficient statistic $S_j$ \red{
  such that the distribution of $\bm e \cond  S_j$ under $H_j$ is known. 
}
Fixing $\alpha \in (0,1)$, for $j \in [m]$ define  $\hat \cR _j (\bm e) \coloneqq  \cR(\bm e) \cup \{j\}$
and  
\begin{equation}
  \begin{aligned}
    \label{eq:phi_func}
    \phi_j(c; S_j) \coloneqq \EE \left[\frac{m}{\alpha} \cdot \frac{\indc{c \te_j \ge \frac{m}{\alpha |\hat \cR_j(\bm \te)|}}}{|\hat \cR_j(\bm \te)|}  - \te_j \Bigggiven S_j \right],
  \end{aligned}
\end{equation}
where $\bm \te = (\te_1, \dots, \te_m)$ follows the conditional distribution of $\bm e \cond S_j$ under $H_j$. 
\red{
At $c=1$, the first term inside the expectation is precisely the BH-sharp $e_j'$ computed from the resampled $\bm \te$, and $\phi_j(1;S_j) \le 0$ as $e_j' \le e_j$. 
Increasing $c$ expands the event in the indicator, so intuitively we desire a $c$ as large as possible without the conditional expectation of the first term surpassing $\EE[e_j \given S_j]$. To that end, define the critical value 
\begin{equation}
  \label{eq:critical_value}
  \hat c_j \coloneqq \sup\{c \colon \phi_j(c; S_j) \le 0\}.
\end{equation}
Assuming $\phi_j(c;S_j)$ is left-continuous at $c = \hat c_j$,\footnote{
  \red{
  For a general treatment without the left-continuity assumption, see Appendix~\ref{appd:general_boosting_definition}.
  }
} 
$\phi_j(\hat c_j ;  S_j) \le 0$. Hence,  we define our boosted e-values as follows:
\begin{align}
  \label{eq:boosted_eval}
  e^\boost_j \coloneqq 
  \frac{m}{\alpha |\hat \cR_j(\bm e)|} \cdot \ind\bigg\{ \hat c_j  e_j \ge \frac{m}{\alpha |\hat \cR_j(\bm e)|}\bigg\}.
\end{align}
For $j\in\cH_0$, $e^\boost_j$ is designed to have conditional expectation given \(S_j\) close to that of  $e_j$, thereby reducing the slack in step (b) relative to $e_j'$.
However, $(e^\boost_1, \dots, e^\boost_m)$ need not be BH-sharp themselves, so step (a) may become loose again. 
Nevertheless, as stated later in Theorem~\ref{thm:power-improvement}, e-BH applied to the boosted e-values is guaranteed to make at least as many rejections as the original procedure. 
Finally, we note that the particular form of our boosted e-values is, in some sense, chosen heuristically; other boosting constructions may also be considered, though this is outside the scope of the paper.}

We formalize in Theorem~\ref{thm:eval_valid} that 
the boosted e-values $(e_1^\boost,\dots,e^\boost_m)$ 
are valid e-values. 
\begin{theorem}
\label{thm:eval_valid}
Assume that, for each $j$, $c\mapsto \phi_j(c;S_j)$ is left-continuous at $\hat c_j$ almost surely.
When $(e_1,\ldots,e_m)$ are (resp.~compound) e-values, 
the boosted e-values $\be^\boost = (e^\boost_1,\ldots,e_m^\boost)$ 
defined in~\eqref{eq:boosted_eval} are (resp.~compound) e-values.
\end{theorem}
At a high level, the proof follows from 
$\EE_{H_j}[e_j^\boost \given S_j] = \phi_j(\hat c_j;S_j) + \EE_{H_j}[e_j \given S_j] \le \EE_{H_j}[e_j \given S_j]$, 
where the inequality is due to the choice of $\hat c_j$;
we provide the formal proof in Appendix~\ref{appd:proof_eval_valid}.
As a consequence of Theorem~\ref{thm:eval_valid}, 
the rejection set returned from running e-BH on $\bm e^\boost$ at level $\alpha$
satisfies FDR control, i.e., $\fdr[\cR(\bm e^\boost)] \le \alpha$. 
We summarize this procedure, henceforth referred to as e-BH-CC (e-BH with Conditional Calibration), 
in Algorithm~\ref{algo:ebh_generic}.
\begin{remark}
\label{remark:ebh_cc}
Since $(e_1^\boost,\dots,e_m^\boost)$ are valid (compound) e-values,  
applying e-BH on them at any level $\alpha_\ebh \in (0,1)$ yields a rejection set
with FDR controlled by $\alpha_\ebh$. In the most general version of e-BH-CC, 
we denote the level used in constructing the boosted e-values (the $\alpha$ hyperparameter in \eqref{eq:phi_func} and \eqref{eq:boosted_eval}) as $\alpha_\cc$, 
and allow $\alpha_\cc$ to differ from $\alpha_\ebh$. As we shall see shortly, 
the power improvement of e-BH-CC is theoretically guaranteed when 
$\alpha_\cc = \alpha_\ebh$. In what follows, we will assume that $\alpha_\cc = \alpha_\ebh = \alpha$
unless otherwise specified.
\end{remark}

\begin{algorithm}[ht]
\caption{e-BH-CC}
\label{algo:ebh_generic}
\DontPrintSemicolon
\KwIn{e-values $e_1,e_2,\ldots,e_m$; sufficient statistics $S_1,S_2,\ldots,S_m$; 
target FDR level $\alpha$.}
\For{$j \in [m]$}{
  1. Compute the boosting factor $\hat c_j$.\; 
  2. Construct the boosted e-values $e_j^\boost$ according to~\eqref{eq:boosted_eval}.\;
}
Apply e-BH to $(e_1^\boost, e_2^\boost,\ldots,e_m^\boost)$ at level $\alpha$ 
and obtain the rejection set $\cR^\ebh(\be^\boost)$. 

\KwOut{the rejection set $\cR^\ebh(\be^\boost)$.}
\end{algorithm}

\red{The following theorem states} that the boosted rejection set $\cR(\bm e^\boost)$ dominates  $\cR(\bm e)$
in power\red{, thereby ``boosting'' the e-BH procedure.}
\begin{theorem}\label{thm:power-improvement}
  Under the same left-continuity assumption, given e-values $\bm e$, denote $\bm e^\boost = (e^\boost_1, \dots, e^\boost_m)$ 
  to be the boosted e-values from conditional calibration defined in~\eqref{eq:boosted_eval}. 
  Then $\cR(\bm e^\boost) \supseteq \cR(\bm e)$, where each rejection set comes from running the e-BH procedure at the same level $\alpha \in (0,1)$.
\end{theorem} 
The key to the proof, provided in Appendix~\ref{appd:proof_power_improvement}, is the fact that $e_j^\boost \ge e_j'$ almost surely  for all $ j\in [m]$. 

\subsection{\texorpdfstring{Multiple rounds of conditional calibration}{Multiple rounds of conditional calibration}}

Observe that by design, 
when $j \notin \cR(\bm e)$, $e_j^\boost \in \{ 0, \frac{m}{\alpha(|\cR(\bm e)|+1)} \}$.
Hence, when it is nonzero, $\cR(\bm e^\boost) \supseteq \cR(\bm e)\cup\{j\}$.
This containment is generally strict, meaning such $e_j^\boost$ do not individually look
like BH-sharp e-values (which would have support 
$\{ 0, \frac{m}{\alpha(|\cR(\bm e^\boost)|)} \})$. 
Hence, we suspect some 
residual inefficiency even after the boosting step.

One solution is to boost using CC \emph{again}, on top of the
boosted rejection set.  
Write $\be^{(0)}=\be$, $\be^{(1)}=\be^\boost$,
$\cR^{(1)}(\be)=\cR^\ebh(\be^{(1)})$, and 
$\hcR_j^{(1)}(\be)\coloneqq \cR^{(1)}(\be)\cup\{j\}$.  
We can then define the calibration function for this second CC round:
\[
  \phi_j^{(2)}(c;S_j)
  \coloneqq
  \EE\left[
    \frac{m}{\alpha |\hcR_j^{(1)}(\tilde{\be})|}
    \indc{c\te_j^{(0)}\ge
    \frac{m}{\alpha |\hcR_j^{(1)}(\tilde{\be})|}}
    -\te_j^{(0)}
  \Bigggiven S_j\right].
\]
Here, $\tilde{\be}^{(0)} \coloneqq \tilde {\be}\sim \bm e\given S_j$,
and 
$\tilde{\be}^{(1)}$ are the boosted e-values
obtained by running the first CC round on
$\tilde{\be}^{(0)}$. Likewise, 
$\hcR_j^{(1)}(\tilde{\be})\coloneqq \cR^\ebh(\tilde{\be}^{(1)})\cup\{j\}$.
Given $\phi_j^{(2)}$, we can define the critical value $\hat c_j^{(2)}=\sup\{c:\phi_j^{(2)}(c;S_j)\le0\}$ and use it to construct
\[
  e_j^{(2)}
  =
  \frac{m}{\alpha |\hcR_j^{(1)}(\be)|}
  \indc{\hat c_j^{(2)}e_j^{(0)}\ge
  \frac{m}{\alpha |\hcR_j^{(1)}(\be)|}}.
\]
These boosted e-values are now calibrated at the e-BH threshold based on 
$\cR^{(1)}(\bm e)$, rather than the initial, 
smaller $\cR^{(0)}(\bm e)$. 
Hence $e^{(2)}_j$ has a larger chance of 
being nonzero, since that event 
requires clearing a lower BH-sharp cutoff.

More generally, at step $k\ge1$, let $\be^{(k-1)}$ be the boosted e-values obtained after
the previous round of CC and set
$
\cR^{(k-1)}(\be)\coloneqq \cR^\ebh(\be^{(k-1)}),
\hcR_j^{(k-1)}(\be)\coloneqq \cR^{(k-1)}(\be)\cup\{j\}.
$
Let $\tilde{\be}^{(k-1)}$ denote the result of applying the first $k-1$
rounds of CC to a conditional copy $\tilde{\be}^{(0)}=\tilde{\be}$ of $\be$
given $S_j$,
and define $\hcR_j^{(k-1)}(\tilde{\be})\coloneqq
\cR^\ebh(\tilde{\be}^{(k-1)})\cup\{j\}$. Then
\[
  \phi_j^{(k)}(c;S_j)
  \coloneqq
  \EE\left[
    \frac{m}{\alpha |\hcR_j^{(k-1)}(\tilde{\be})|}
    \indc{c\te_j^{(0)}\ge
    \frac{m}{\alpha |\hcR_j^{(k-1)}(\tilde{\be})|}}
    -\te_j^{(0)}
  \Bigggiven S_j\right].
\]
The corresponding critical value $\hat c_j^{(k)}=\sup\{c:\phi_j^{(k)}(c;S_j)\le0\}$ gives the $k$-round boosted e-values
\[
  e_j^{(k)}
  =
  \frac{m}{\alpha |\hcR_j^{(k-1)}(\be)|}
  \indc{\hat c_j^{(k)}e_j^{(0)}\ge
  \frac{m}{\alpha |\hcR_j^{(k-1)}(\be)|}}.
\]
Although the boosted e-values seem to be iteratively sharpened, 
multiple rounds of CC do not enjoy the same uniform power improvement
guarantee (Theorem~\ref{thm:power-improvement}) that single-round
e-BH-CC does. At a high level, this is because $\phi_j^{(k)}$
does not enjoy any useful monotonicity property across rounds. 
Fortunately, this issue can be fixed by designing $\phi_j^{(k)}$
more carefully. Appendix~\ref{appd:multiple_rounds} 
discusses this fixed multiple-round procedure in depth and
provides a numerical study in the Gaussian setting 
(Section~\ref{sec:example-parametric}) that shows the
second round of CC does in fact improve the power over e-BH-CC. 
This is evidence of further BH-sharpening, though we find 
empirically that one round of CC suffices for 
substantial power improvement.

\subsection{Implementing e-BH-CC}
\label{sec:implementing}

To calculate $\hat c_j$ in \eqref{eq:critical_value}, we assume access to $\phi_j$ through knowledge of the distribution $\be \cond S_j$ under $H_j$,  
and will obtain its value by numerically evaluating the conditional expectation 
using i.i.d.~resamples from $\be\cond S_j$.
It remains unclear how to best estimate the critical value of $\phi_j$ without suffering from unpredictable Monte-Carlo error and downstream FDR violation.

In this subsection, we outline our computationally efficient Monte-Carlo method for implementing e-BH-CC, 
which only requires resamples from $\be \cond S_j$. In our implementation, 
we forego estimating the critical value $\hat c_j$ and instead evaluate 
$\ind\{\hat c_j e_j \ge m/(\alpha| \hat\cR_j(\be)| )\}$ directly,
since this determines the value of $e_j^\boost$. We show that it suffices to evaluate $\phi_j(\cdot; S_j)$ at a specific value, which we can Monte-Carlo estimate using our resamples. By making this simplification, we crucially avoid the issue of finding $\hat c_j$. Our implementation also allows anytime-valid methods to control the estimation error in an online manner, which has methodological and computational benefits to be seen later. This approach takes inspiration from the work of~\cite{luo2022improving}, which similarly uses conditional calibration to improve the power of the knockoff filter~\citep{barber2015controlling,candes2018panning}. We draw a distinction with their method, as ours applies in a general multiple testing framework---boosting the power of the knockoff filter is a specific application of e-BH-CC, as detailed in Section \ref{sec:example-knockoffs}.

\subsubsection{Evaluating the numerator}

The only part of $e^\boost_j$ to evaluate is the indicator in the numerator 
(the denominator is a function of the original e-values). 
Observe the following: 
\begin{equation}
\begin{aligned}\label{eq:evaluate_fast_coef}
\ind\bigg\{\hat c_j e_j \ge \frac{m}{\alpha | \hat\cR_j (\be)| }\bigg\}
&=
\ind\bigg\{\hat c_j\ge \frac{m}{\alpha | \hat\cR_j (\be)| }/ e_j\bigg\}
= \ind\{ \phi_j(\tilde c_j; S_j) \le 0\},
\end{aligned}
\end{equation}
where $\tilde c_j   = \frac{m}{\alpha | \hat\cR_j (\be)| }/ e_j$. The last event
$\{\phi_j(\tilde c_j; S_j) \le 0\}$ may confuse the astute reader, who will
recall that $\phi_j(c;S_j)$ is an expression with $c$ inside the conditional expectation.
By evaluating $\phi_j(\cdot;S_j)$ at the random $\tilde c_j$, we mean to evaluate it at the \emph{value taken} by $\tilde c_j$, rather than substituting $c$ with $\tilde c_j$ in \eqref{eq:phi_func}. 
This is equivalent to evaluating
\begin{equation}\label{eq:evaluate_fast_coef_explicit}
  \EE\left[
    \frac m \alpha \cdot \frac{ \indc{ (\frac{m}{\alpha | \hat\cR_j (\be)| }/ e_j)\cdot \te_j \ge \frac{m}{\alpha |\hat\cR_j(\tilde \be )|} } }{|\hat \cR_j(\tilde \be )|} - \tilde e_j \bigggiven S_j, \be
  \right]
\end{equation}
with the expectation over the distribution $\tilde \be \cond S_j, \be$.
{In~\eqref{eq:evaluate_fast_coef_explicit}, we can write the resample 
$\bm \te = f(S_j, U_j)$
for some function $f$ and a uniform random 
variable $U_j$ that is independent of everything else.
From this representation, we can clearly see that $\tilde \be \indep \be  \cond S_j$, and 
the conditional expectation \eqref{eq:evaluate_fast_coef_explicit} is effectively over  $\tilde \be\given S_j$.}
Equations \eqref{eq:evaluate_fast_coef} and \eqref{eq:evaluate_fast_coef_explicit} imply that we can evaluate the numerator by estimating the conditional mean
(given $S_j$ and $\be$) of 
\begin{equation}
  \label{eq:eval_quant}
  E_j \coloneqq \frac{m}{\alpha }\cdot \frac{\indc{ \tilde e_j/e_j \ge |\hat \cR_j( \be )|/|\hat \cR_j(\tilde \be )|  }}{|\hat \cR_j(\tilde \be )|}  -  \tilde e_j .
\end{equation}
Assume we have $K$ i.i.d.~resamples from $\be \cond S_j$. 
From each resample $\tilde \be^{(k)}, k\in[K]$, we compute  
$E_{j,k}$ by replacing $\tilde \be$ with $\tilde \be^{(k)}$
in~\eqref{eq:eval_quant} to resample $E_j$. 
The typical Monte-Carlo estimator for $\phi_j(\tilde c_j;S_j)$ is the average of the samples $\overline E_{j,K} \coloneqq \frac 1K \sum_{k=1}^K E_{j,k}$. 
When $K$ is large, giving us arbitrarily precise Monte-Carlo estimation, 
we can replace $\indc{\phi_j(\tilde c_j ; S_j) \le 0}$ with $\indc{ \overline E_{j,K} \le 0}$ in the construction of $e_j^\boost$ with no practical repercussions. 

When $K$ is instead chosen with regard to computational budget restrictions, then we might 
expect Monte-Carlo estimation error. 
We would like to control the effect of such error on the resulting FDR of the overall procedure. 
Importantly, we want to be confident  in the \emph{sign} of $\phi_j(\tilde c_j; S_j)$ 
rather than its value. 
Therefore, we can use confidence intervals to control the error from estimating the sign, which filters through as an additive penalty to the resulting FDR.

For some fixed $K \in \ZZ^+$ and for each $j$,
produce $K$ i.i.d.~resamples $\be^{(1)}, \dots, \be^{(K)} $  conditional on $S_j$, 
and compute $E_{j,1},\dots,E_{j,K}$ according to~\eqref{eq:eval_quant}.
Define $\alpha_{\ci} = \alpha_0 |\cR^\ebh(\be)| /m$, 
for some $\alpha_0 \in (0,1)$ reflecting the Monte-Carlo 
error budget. 
Using $E_{j,1}, \dots, E_{j,K}$, 
we construct a $(1-\alpha_{\ci})$-coverage confidence interval 
$C_{j,K} \coloneqq C_{j,K}(E_{j,1},\dots, E_{j,K})$  
such that 
$\PP\big( \phi_j(\tilde c_j; S_j)  \in C_{j,K}  \cond S_j, \be\big ) \ge 1-\alpha_\ci$.
Let $U_{j,K}$ be the upper endpoint of $C_{j,K}$. 
Define the CI-approximated boosted e-values
\begin{equation}
  \label{eq:ci_boosted_eval}
  e^{\boost, \ci}_j = m\cdot \frac{ \indc{U_{j,K} \le 0} }{\alpha |\hat \cR_j(\bm e)|}.
\end{equation}
The above procedure is explicated in Appendix \ref{appd:algos} as Algorithm~\ref{alg:ebh_cc_ci}. It replaces Line 2 of Algorithm~\ref{algo:ebh_generic} with a specific (approximate) computation for $\hat c_j$.  We further formalize the validity of the CI-approximated boosted e-values 
in Proposition~\ref{prop:mc_error}.

\begin{proposition}
  \label{prop:mc_error}
  Suppose $\be = (e_1,\ldots,e_m)$ are compound e-values. 
  Running e-BH on the collection of CI-approximated 
  boosted e-values $ \be^{\boost, \ci }=  \{e^{\boost, \ci}_j \}_{j \in [m]}$ 
  defined in~\eqref{eq:ci_boosted_eval} with the target FDR level $\alpha$ and 
  the Monte-Carlo error budget $\alpha_0$, we have 
  $
    \fdr[\cR^\ebh(\be^{\boost, \ci})] \le \alpha+\alpha_0.
  $
\end{proposition}
If we replace the $(1-\alpha_{\ci})$ confidence interval with 
an asymptotic $(1-\alpha_\ci)$ confidence interval, we obtain 
instead asymptotic FDR control, as formalized by the 
following corollary.
\begin{corollary}\label{cor:asymptotic_mc_error}
  Replacing the $(1-\alpha_\ci)$ confidence interval 
  in Proposition~\ref{prop:mc_error} with an {asymptotic} $(1-\alpha_{\ci})$ confidence interval leads to the conclusion
  $
    \lim_{K\to\infty}\fdr[\cR^\ebh(\be^{\boost, \ci})] 
    \le \alpha+\alpha_0,
  $
  where $K$ is the number of resamples from $\be \cond S_j$.
\end{corollary}
The proofs of Proposition~\ref{prop:mc_error} and Corollary~\ref{cor:asymptotic_mc_error}
can be found in Appendix~\ref{appd:proof_mc_error}. 
There are also anytime-valid \citep{ramdas2023game} versions of these CI-approximated e-values which allow for continuous resampling with validity guarantees; details are delegated to Appendix~\ref{appd:avcs_approx}. In addition, one can improve the power of these Monte-Carlo strategies by filtering for high-potential rejections first; see  Appendix \ref{appd:filtering_out} for further details. Both ideas are implemented in our numerical studies.

\section{Example: Gaussian mean testing}
\label{sec:example-parametric}

First, we study the relatively simple task of identifying which components of a Gaussian random vector have {positive} mean.
Given $Z \sim \cN_m(\mu, \Sigma)$ for a positive definite matrix $\Sigma$, we define the null hypotheses $H_j\colon \mu _j = 0$ and their one-sided alternatives $H_j^{\mathrm{alt}}\colon \mu_j > 0$.~\citet{fithian2020conditional} consider this problem for both known and unknown covariance matrix $\Sigma$ as an example application of dBH, their conditionally calibrated correction to the BH procedure. \red{Here, we assume $\Sigma$ is known}; 
\red{see Appendix \ref{sec:tstat} for the setting where $\Sigma$ is known only up to a scalar multiplicative  factor.}
After constructing the relevant e-values for testing $H_1, \dots, H_m$, we closely follow Section 3.1 of their work as their conditioning statistic plays the role of our sufficient statistic in the e-value resampling step.

Consider $Z \sim \cN_m(\mu, \Sigma)$ with $\Sigma$ known. Without loss of generality, we assume $\Sigma_{jj} = 1$ for all $j$.~\citet{vovk2021values} observe that for $a_j\neq 0$,
\begin{equation}\label{eq:zstat-evalue}
  e_j = \exp(a_j Z_j - a_j^2/2) 
\end{equation} 
is a valid e-value with respect to $H_j$. The e-value $e_j$ is the likelihood ratio test (LRT) statistic for $H_j$ versus the point alternative 
$H_j^{(a_j)}\colon \mu_j =a_j$. The choice of the hyperparameter $a_j$ can be either fixed \emph{a priori} 
or estimated independently: either leads to a valid e-value.
However, different values of $a_j$ will lead to varying levels of power. 
\red{For example, the Neyman-Pearson lemma suggests  choosing $a_j$ to match  the true data-generating distribution under the alternative hypothesis.}
\red{Moreover, when the alternative for $H_j$ is composite, 
it is impossible to be certain whether a choice of $a_j$ will produce a powerful $e_j$.} 
In these cases, the practitioner may attempt to learn the best $a_j$ independently from a separate, independent dataset. We run simulations, detailed in Section \ref{sec:sim-parametric}, where our method is applied on e-values constructed using correctly-specified ``exact'' LRT e-values as well as a range of misspecified LRT e-values. 

In order to use conditional calibration, we define the sufficient statistic $S_j = Z_{-j} - \Sigma_{-j, j} Z_j$ for each $j \in [m]$, 
as in \cite{fithian2020conditional}. To resample $(\te_1, \dots, \te_m) \cond S_j$, 
it suffices to resample from the conditional joint distribution $Z\cond S_j$ (under $H_j$) and compute \eqref{eq:zstat-evalue} for each $j$. 
We claim that we can sample $Y  \sim \cN(0, \Sigma_{jj})$ and construct the $z$-statistic $\tZ^{(j)} = (\tZ^{(j)}_1 ,\dots, \tZ^{(j)}_m)$ such that
\begin{equation}
  \label{eq:resample_zstat}
  \tZ^{(j)} _j = Y,\quad \tZ^{(j)} _{-j} =    S_j + \Sigma_{-j,j} Y
\end{equation}
in order to obtain a resample from $Z\cond S_j$. The following proposition (with proof in Appendix \ref{appd:proof_zstat_eval_resample}) formalizes the claim. %
\begin{proposition}\label{prop:zstat}
  For each $j\in[m]$, choose sufficient statistic $S_j = Z_{-j} - \Sigma_{-j, j} Z_j$ and resample the $m$-dimensional $z$-statistic $\tZ^{(j)}$ from $Z \cond S_j$ as written in \eqref{eq:resample_zstat}. Define the resampled e-values $\tilde e_k^{(j)} = \exp(a_k \tZ_k^{(j)}- a_k^2/2)$ for each $k\in[m]$. Then, under $H_j$,
  $$
  \left(\te_1^{(j)}, \ldots, \te_m^{(j)} \right)\cond S_j \sim (e_1, \ldots, e_m) \cond S_j.
  $$ 
  Furthermore, independent samples of $Y\sim \cN(0, \Sigma_{jj})$ give rise to independent resamples from $(e_1, \ldots, e_m) \cond S_j$.
\end{proposition}

\subsection{Simulation studies}\label{sec:sim-parametric}

\paragraph{Settings.} We take $m=100$ and
$
\mu = (\underbrace{\Delta, \dots, \Delta}_{10}, 0, \dots, 0) \in \RR^{100}
$
(so $\cH_1  =  \{1, \dots, 10\}$), 
where $\Delta$ is a constant determining the signal strength.
The covariance matrix has the form $\Sigma_{ij} = \rho^{|i-j|}$ for any $i,j \in [m]$. 
The $z$-statistics thus follow $\cN_m(\mu, \Sigma)$. 
We set $\rho = 0.5$ and vary $\Delta \in \{1, 1.5, 2, \dots, 6\}$. 

\paragraph{Methods.} Four different constructions of LRT e-values \eqref{eq:zstat-evalue}
are considered. The first three constructions use $a_j = a,~ \forall j\in[m]$ for $a \in \{1, 2, 3\}$, 
while the fourth uses a correctly specified $a_j$ (setting it equal to $\Delta$). 
For each construction, we compare e-BH and e-BH-CC, both at level $\alpha = 0.05$. As reference, we also run BH \red{and its conditionally calibrated version, dBH~\citep{fithian2020conditional},} at level $\alpha$ on the one-sided p-values derived from the $z$-statistics.
We choose $\alpha_0 = 0.1\cdot \alpha = 0.005$ as our Monte-Carlo budget; the full description of the Monte-Carlo implementation is in Appendix \ref{appd:deferred_details_sims_parametric}.

\suppressfloats[t]
\begin{figure}[t]
  \centering
  \includegraphics[scale=1.0]{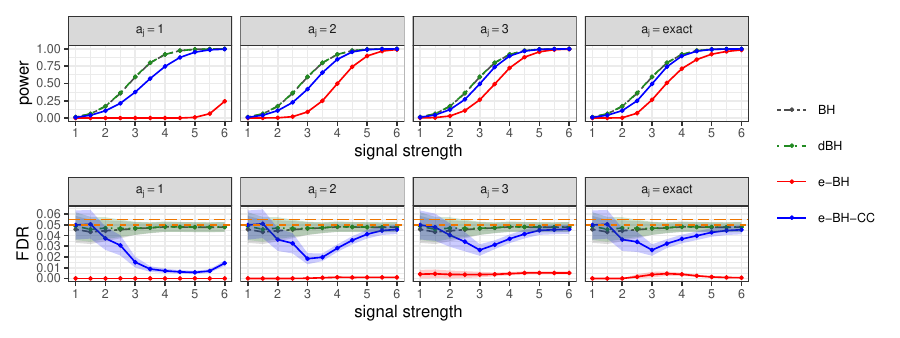} 
  \caption{Realized power and FDP of the simulated experiments for $z$-testing. 
  Each plot contains the averaged metrics over 1,000 replications. 
  The short and long orange dashed lines in the FDP plots represent 
  the target FDR (0.05) and the FDR bound with Monte-Carlo error (0.055), respectively. Shading represents error bars of two standard errors above and below.}
  \label{fig:zstat}
  \centering
\end{figure}

\paragraph{Experiments and results.}
In Figure \ref{fig:zstat}, we exhibit the mean power and FDP curves 
taken over 1,000 replications for each choice of e-value construction (i.e., $a_j$) and $\Delta$.

For each choice of $a_j$, we see a major improvement in power comparing e-BH-CC to base e-BH. The power improvement is uniform over all possible signal strengths, but is especially large when $a_j$ is incorrectly specified. For example, when $a_j =1$, we see that e-BH has minimal power even in the large signal experiments. In this situation, e-BH-CC immensely outperforms e-BH.

It is worth noting that in these experiments, BH still controls the FDR as the p-values are PRDS~\citep{benjamini2001control}. 
\red{Furthermore, they fulfill the conditional PRDS condition~\citep{fithian2020conditional}, so we may use the deterministic dBH variant, which provably contains the BH rejection set.}
Thus, we can interpret the power gap between BH and e-BH as the power penalty of using e-values in place of p-values. As such, the e-BH-CC power curve demonstrates how much of the power loss is reclaimed. In all cases, we see that e-BH-CC improves substantially over e-BH to approach BH. The power reclamation is even more evident in the misspecified case $a_j = 1$, as mentioned previously. One takeaway is that when $a_j$ is chosen poorly through fixed means or estimation, e-BH-CC can still perform powerfully, giving an added margin of safety for parameter misspecification.

Lastly, it is clear that the FDR is controlled empirically at $\alpha + \alpha_0$ 
(and even at $\alpha$) for all settings. 
The simulated FDR of e-BH-CC is generally much higher than that of base e-BH 
(which is close to zero), affirming our ability to use more of the FDR budget by boosting e-BH.

\section{Example: knockoffs}
\label{sec:example-knockoffs}

\subsection{Conditional independence testing}

\noindent We turn to the problem of \emph{feature selection}, where the goal is to discover which of the covariates 
$X = (X_1, X_2, \dots, X_m)$ are significant to the outcome variable $Y$ conditional on the other covariates. 
We can encode $X_j$'s significance to the outcome with the following \emph{conditional independence} hypothesis $
  H_j\colon Y \indep X_j \cond X_{-j}$.
We call $X_j$ a \emph{null} variable when $H_j$ is true and \emph{nonnull} when $H_j$ is false. 
The null hypothesis set $\cH_0$ corresponds to the indices of the null covariates,
and analogously $\cH_1$ corresponds to the nonnull covariates. Our goal is to reject 
a subset of $[m]$ in a way that controls the FDR, i.e., there is little intersection with $\cH_0$.

One popular approach for testing conditional independence is to assume the \emph{model-X framework}. In this framework, the covariate-outcome tuple $(X_1, \dots, X_m, Y)$ is interpreted as a draw from some joint distribution $\PP_{XY} = \PP_X \times \PP_{Y\cond X}$, where $\PP_{X}$ (the covariate distribution) is assumed to be known (or reasonably well-approximated) while the model $\PP_{Y\cond X}$ is totally unknown. Under these assumptions, one can use a multiple testing procedure known as the \emph{model-X knockoff filter}~\citep{candes2018panning} to conduct feature selection with a provable exact-sample FDR control guarantee. We give a brief exposition of the procedure in the following section with the goal of showing its relevance to e-values and e-BH-CC.

\subsection{Model-X knockoffs}

Given the design matrix $\bm X = (X_{i1}, \dots, X_{im})_{i \in [n]} \in \RR^{n\times m}$ 
and outcomes $\bm Y = (Y_1, \dots, Y_n)^\top$ in the dataset, the model-X knockoff filter uses the practitioner's knowledge of the covariate distribution $\PP_X$ to construct ``knockoff variables'' $\bm \tX =  (\tX_{i1}, \dots, \tX_{im})_{i \in [n]}$. These knockoff variables must satisfy an independence condition $\bm Y\indep \bm \tX \cond \bm X$ as well as the ``swap'' condition:
\begin{equation}
  (\bm X_{j}, \bm \tX_{j}, \bm X_{-j}, \bm \tX_{-j}) \sim (\bm \tX_{j}, \bm X_{j}, \bm X_{-j}, \bm \tX_{-j})\;\; \forall j\in[m],
\end{equation}
that is, the joint distribution of $(\bm X, \bm \tX)$ is preserved after swapping the positions of the columns $\bm X_j$ and $\bm \tX_j$. Hence, $\bm X$ and $\bm \tX$ are quite similar in terms of their dependency structure, but the latter has no significance to the outcome that is not already expressed through $\bm X$. Hence, conditional on $\bm X$, $\bm \tX$ are knockoff covariates which carry no value (in predicting $\bm Y$).

At a high level, access to these knockoff variables provides a way to appropriately calibrate any feature importance measure of $X_j$. 
The procedure uses the augmented design matrix $[\bm X, \bm \tX]$ and the outcome $\bm Y$ to construct a feature importance vector $W = \cW([\bm X, \bm \tX], \bm Y) \in \RR^{m}$ 
(where $W_j$ corresponds to $X_j$ and its knockoff) using some algorithm $\cW(\cdot)$ with the property that switching $\bm X_j$ and $\bm \tX_j$ in the augmented design matrix flips the sign of the resulting $W_j$. In general, such statistics are computed by finding individual feature importances for $X_j$ and $\tX_j$ and taking their difference. As noted in~\cite{candes2018panning}, 
the signs of $\{W_j \colon j\in \cH_0\}$ conditioned on their magnitudes $\{|W_j| \colon j\in \cH_0\}$ are i.i.d.~coin flips. In contrast, $W_j$ for nonnull $X_j$ would tend to be positive (signifying more importance of $X_j$ than $\tX_j$). This aligns with the intuition of $\tX_j$ as ``negative controls'': when $X_j$ and $\tX_j$ are both insignificant, the feature importance statistic is symmetric around 0, but when $X_j$ is significant $W_j$ will be positively skewed. Using the feature importances $W$, we select the set of features:
\begin{equation}\label{eq:mxkn_threshold}
  \cR^{\mathrm{kn}} = \{j\colon W_j \ge T  \},\quad\textnormal{ where } T \coloneqq \inf\left\{ t > 0 \colon \frac{1 + \sum_{j\in[m]} \indc{W_j \le -t}}{\sum_{k \in [m]} \indc{W_k \ge t}} \le \alpha \right\}.
\end{equation}
The threshold $T$ is defined to be the smallest $t >0$ such that an estimate of the FDR (which is constructed using our intuition on the null i.i.d.~coin flips above) is controlled by $\alpha$. For a rigorous proof of FDR control, see~\cite{candes2018panning,barber2015controlling}.

Although the knockoff filter is an elegant and effective approach to feature selection, it still suffers from a few weaknesses. One such weakness, known as the \emph{threshold phenomenon}, occurs when the number of features with non-negligible significance is less than $ 1/ \alpha$. In this sparse setting, the knockoffs procedure experiences a drastic loss in power. An explanation for this phenomenon is in the definition of $T$, where we compare an estimator of the FDP against $\alpha$. The denominator measures the size of the rejection set when using $t$, while the numerator is at least 1; when the denominator is less than $ 1 / \alpha$, the rejection set $\{j\colon W_j \ge t\}$ is forced to be empty. In practice, one still sees the procedure suffer from powerlessness when the proportion of significant features is generally sparse, even when $|\cH_1|$ is technically above $ 1 / \alpha$.

\subsubsection{Derandomizing knockoffs with e-values}

In addition to the threshold phenomenon, another weakness of the knockoff filter is its high selection variability. Sampling the knockoff matrix $\bm \tX$ introduces extraneous randomness into the procedure, so two different runs of the knockoff filter (with two distinct knockoff matrices) may result in wildly different rejection sets. For practical purposes, this is a major detraction---for example, two scientists studying the same feature selection problem with the same dataset may conclude totally different sets of significant covariates. Hence, for scientific reproducibility, a way to derandomize the knockoff filter was highly sought-after.

\citet{ren2024derandomised} solve the variability issue by proposing a derandomized version of the knockoff filter, 
which begins by connecting the knockoffs with e-BH through defining
\begin{equation}\label{eq:kn_evalues}
  e_j \coloneqq m \cdot \frac{\indc{W_j \ge T}}{1 + \sum_{k\in [m]} \indc{W_k \le -T} },
\end{equation}
where $W_j$ is the feature importance for $X_j$ and $T$ is the rejection threshold in \eqref{eq:mxkn_threshold}. It can be shown that $\sum_{j\in\cH_{0}}\EE[e_j] = m$, implying that $e_1, \ldots, e_m$ are compound e-values (Definition \ref{def:compound_eval}). Running e-BH on these e-values will return a rejection set with FDR control, and the authors show that this rejection set is identical to the output of the knockoff filter using the same knockoff matrix.  
Moreover, it is shown that $\alpha$ can be replaced by a tuning parameter $\alpha_\kn$ in the threshold $T$ and the 
resulting $e_j$'s are still valid (compound) e-values. The procedure is then derandomized 
by averaging the knockoff e-values over $d$ distinct runs of knockoffs. Specifically, we define    
$\bar e_j = \frac{1}{d}\sum^d_{k=1}e_j^{(k)}$, where $e_j^{(k)}$ denotes the knockoff e-value 
for $X_j$ in the $k$-th run. The e-BH procedure is then applied to $(\bar e_j)_{j\in[m]}$.
Additional details of the derandomized knockoffs procedure are 
given in Appendix~\ref{appd:derand_knockoff}.

Despite the improved stability, the power of the derandomized knockoffs procedure tends to suffer relative to that of the original method. 
This is due to the fact that the knockoff e-values in \eqref{eq:kn_evalues} (over which e-BH is equivalent to the original method) 
are nearly BH-sharp, while the averaged e-values no longer have this property, leading to a loss of power from e-BH. 
This power drop is most apparent in regimes with low-to-moderate signal strength. %
To add to this, the gap becomes larger with more derandomization.

Therefore, we can identify two settings  where the knockoff filter, as an e-value procedure, suffers from power loss: sparse signals and derandomization. However, we can reclaim the power loss and fill in the gap by using conditional calibration to boost the derandomized knockoffs procedure---even in settings where the threshold phenomenon may happen. In the following subsection, we will describe how to implement e-BH-CC for derandomized model-X knockoffs.

\subsection{Conditional calibration for derandomized knockoffs}
\label{sec:cc-kn}

We choose $S_j = (\bm X_{-j}, \bm Y)$ as the sufficient statistic.
To construct a resample $\bm X'$ of the design matrix $\bm X$ conditional on $S_j$, 
we invoke our knowledge of $\PP_{X_j \cond X_{-j}}$ to let $\bm X'_j$ be samples from the conditional distribution 
of $\bm X_j \cond \bm X_{-j}$ and assign to $\bm X'_{-j}$ the corresponding column values of $\bm  X_{-j}$. We can use this resample as our starting point for the derandomized knockoffs procedure (with $d$ derandomization runs), 
which outputs the resampled averaged e-values $\tilde e_1, \dots, \tilde e_m$.

We provide a proposition stating the correctness of this
resampling procedure in 
producing i.i.d. resamples from $(\bar e_1, \dots, \bar e_m)\cond S_j$ (and its proof) in 
Appendix \ref{appd:derand_knockoff}.

\subsection{Simulation studies}\label{sec:sim-knockoffs}

\begin{figure}[t] 
  \centering
  \includegraphics[scale=1.0]{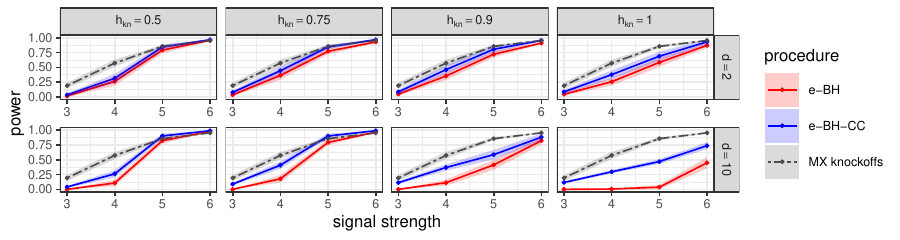}
  \caption{Realized power for the knockoffs simulations when $\beta$ is dense ($|\cH_1|=25$). 
  Each facet contains the averaged metric over 100 replications; shading indicates error bars.}
  \label{fig:mxkn-tp25}
  \centering
\end{figure}

\begin{figure}[b]
  \centering
  \includegraphics[scale=1.0]{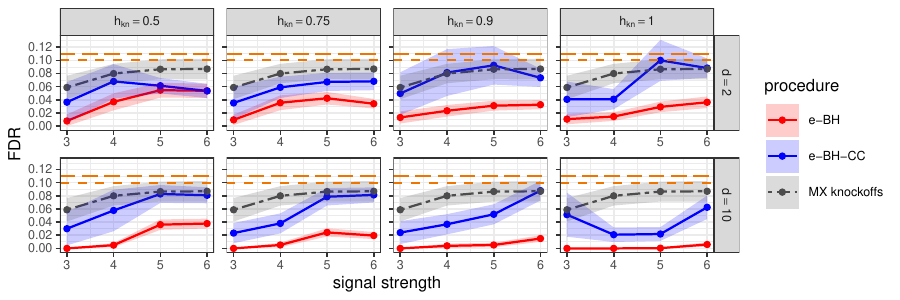}
  \caption{Realized FDP for the knockoffs simulations when $\beta$ is dense ($|\cH_1|=25$), corresponding to Figure~\ref{fig:mxkn-tp25}. Each facet contains the averaged metric over 100 replications; shading indicates error bars. The short and long orange dashed lines represent the target FDR (0.1) and the FDR bound with Monte-Carlo error (0.11), respectively.}
  \label{fig:mxkn-tp25-fdr}
\end{figure}

\begin{figure}[t]
  \centering
  \includegraphics[scale=1.0]{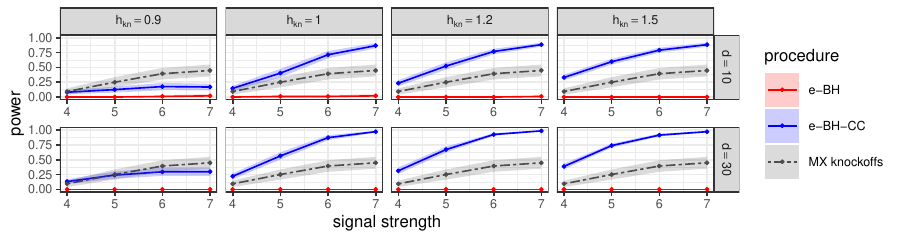}
  \caption{Realized power for the knockoffs simulations when $\beta$ is sparse (specifically, $|\cH_1|=9$ and $\alpha=0.1$, so this is \emph{below} the threshold). 
  The details are otherwise the same as in Figure \ref{fig:mxkn-tp25}.}
  \label{fig:mxkn-tp9}
  \centering
\end{figure}

\begin{figure}[b]
  \centering
  \includegraphics[scale=1.0]{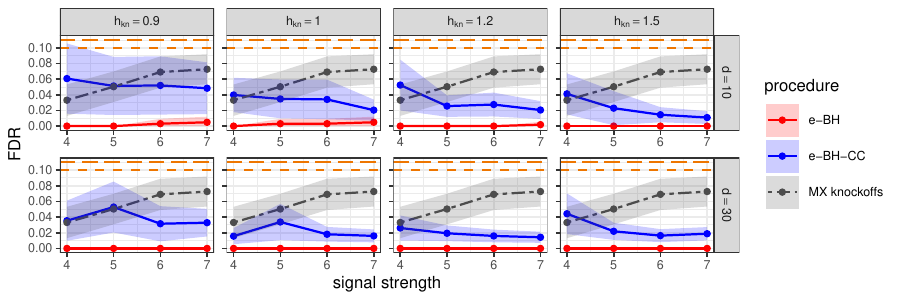}
  \caption{Realized FDP for the knockoffs simulations when $\beta$ is sparse (specifically, $|\cH_1|=9$ and $\alpha=0.1$, so this is \emph{below} the threshold), corresponding to Figure~\ref{fig:mxkn-tp9}. The details are otherwise the same as in Figure~\ref{fig:mxkn-tp25-fdr}.}
  \label{fig:mxkn-tp9-fdr}
\end{figure}

\paragraph{Settings.} 
We generate the design matrix and outcome vector under the Gaussian linear model. 
Let $n = 500$ be the number of observations and $m=200$ be the number of covariates 
(i.e., number of hypotheses). For each observation, the row of covariates is jointly drawn from $X\sim  \cN_m(\mu, \Sigma)$. In our experiments, we set $\mu =  0$ and $\Sigma_{ij} = 0.5^{|i-j|}$, for any $i,j\in[m]$.

We then conditionally draw the response by the Gaussian linear model
$Y\cond X \sim \cN(X^\top\beta, 1)$, giving one independent draw of the covariate-outcome pair $(X, Y)$. We repeat this $n$ times to obtain our dataset $(\bm X, \bm Y)$.
The coefficient vector $\beta$ is formatted as

$$
\beta = \Big(\underbrace{0, \ldots, 0}_z, {\Delta}/{\sqrt{n}}, \underbrace{0, \ldots, 0}_z,  -{\Delta}/{\sqrt{n}}, \ldots \Big) \in \RR^m;
$$

that is, every $z$ zeroes followed by one nonzero amplitude with alternating sign. Note that this directly translates to $|\cH_1| = \lfloor \frac{m}{z+1}\rfloor$. 

We will consider two separate regimes: the dense regime, with $z = 7$ and $|\cH_1| = 25$ nonnull variables
and the sparse regime, with $z \in \{19, 20\}$ and $|\cH_1| \in \{10, 9\}$ 
(the result for the $|\cH_1| = 10$ case is delegated to Appendix~\ref{appd:add-sim-kn}). By running simulations for both regimes, we can demonstrate the power improvement from e-BH-CC even in the presence of the threshold phenomenon, as the target level will be set at $\alpha = 0.1$. 
In the dense regime, we will simulate for each $\Delta \in \{3, 4, 5, 6\}$; 
in the sparse regime, we will simulate for each $\Delta \in \{4, 5, 6, 7\}$. 
By doing so, we can capture the range of power over a variety of signal magnitudes.

\paragraph{Methods.} The objective is to compare the derandomized knockoffs procedure (e-BH) before and after boosting (e-BH-CC) at the target FDR level $\alpha = 0.1$, so we keep the hyperparameters and implementations the same between the two procedures. 
For boosting,  we set the Monte-Carlo budget $\alpha_0 = 0.01$.
Appendix \ref{appd:deferred_details_sims_knockoffs} contains specifics on the knockoffs
as well as computational details of e-BH-CC. The Monte-Carlo details only pertain to e-BH-CC.

For the dense $\beta$ experiments, we use $d\in \{2, 10\}$ runs of the knockoffs procedure to derandomize. 
We run experiments for multiple values of the hyperparameter $\alpha_\kn$ by choosing a factor $h_\kn \in \{0.5, 0.75, 0.9, 1.0\}$ and 
defining $\alpha_\kn = h_\kn  \alpha =  0.1 h_\kn$. In contrast, for the sparse $\beta$ experiments, we use $d\in \{10, 30\}$ copies of the knockoff design matrix to derandomize and instead consider $h_\kn \in \{0.9, 1.0, 1.2, 1.5\}$. We consider a range of $h_\kn$ to show that e-BH-CC empirically uniformly improves upon e-BH. Further, to combat the threshold phenomenon, it intuitively helps to have $\alpha_\kn > \alpha$ so that the threshold $T$ is lower and the number of nonzero e-values   is consequently higher (to demonstrate this, we let $h_\kn$ range from below to above $1$). We choose $\alpha_\kn \le \alpha$ in the dense $\beta$ case as per the original derandomized knockoffs paper~\citep{ren2024derandomised}.

\paragraph{Experiments and results.}
We run 100 replications for each choice of $\Delta, z, d, h_\kn$. 
Each replication consists of constructing the basic derandomized knockoffs e-values and boosting these e-values using conditional calibration. 
In addition, for each choice of $\Delta$ and $z$, we run the regular model-X knockoff filter
to serve as a baseline and to visualize the power gap (shown in Figures~\ref{fig:mxkn-tp25} and~\ref{fig:mxkn-tp9}). For each procedure, we record its power and FDP.
The FDP diagnostics corresponding to the power plots in Figures~\ref{fig:mxkn-tp25} and~\ref{fig:mxkn-tp9} are shown in Figures~\ref{fig:mxkn-tp25-fdr} and~\ref{fig:mxkn-tp9-fdr}, respectively.
 
We find that e-BH-CC closes the power gap and at times surpasses the regular knockoffs filter while staying virtually derandomized.

The results of the dense $\beta$ experiment are in Figure \ref{fig:mxkn-tp25}. Even when $d=2$, the power drop from derandomization is notable. 
Our proposed e-BH-CC procedure improves power \emph{for all} choices of $d$ and $\alpha_\kn$. The power drop, and subsequent power gain from boosting, are even more apparent when $d=10$.

The results of the sparse $\beta$ experiment ($|\cH_1|=9$)
are shown in Figure~\ref{fig:mxkn-tp9}. 
The setting of Figure \ref{fig:mxkn-tp9}, where the threshold phenomenon is in full swing (as $|\cH_1| < 1/\alpha = 10$), shows an even larger contrast between the three procedures. 
When $h_\kn \ge1$ and for either $d\in\{10,30\}$, derandomized knockoffs has near-zero power; 
e-BH-CC improves to near-full power for high signal amplitudes. Furthermore, e-BH-CC has much better performance than the regular knockoffs filter.

\FloatBarrier

\section{Example: conformalized outlier detection}
\label{sec:example-conformal}

Given a set of data, we consider the task of identifying the units 
whose distributions differ from that of a 
reference dataset. 
This problem is known as outlier detection, a fundamental task in many fields.
For example, in finance, an important task is to detect fraudulent user activities in transaction data~\citep{ahmed2016survey}.
More generally, in some applications, the inliers (e.g., the non-fraudulent transactions) may follow 
different distributions in the reference and test datasets---this could happen if 
the reference dataset is collected with some preferences based on observed covariates. 
For instance, one might include more individuals from minority groups in the reference dataset to 
ensure representation. In these cases, the goal is then to distinguish the outliers from the inliers 
given an identifiable distribution shift between the inliers and the reference dataset.

Formally, let $Z = (X,Y)$ denote a unit,
where $X\in\cX$ are the covariates and $Y \in \cY$ is the response. 
Assume that we are given a calibration dataset 
$\cD_{\calib} = \{Z_i\}_{i=1}^n$ that is assumed to contain 
i.i.d.~units drawn from some distribution $P$. For a test dataset of independent 
units $\cD_\test = \{Z_{n+j}\}_{j=1}^m$, 
the objective is to decide whether each unit $Z_{n+j}$ follows distribution $Q$, where 
$    \frac{\dif Q}{\dif P}(x,y) \equiv w(x)$
for some known weight function $w\colon \cX\to\RR^+$ 
(where $ \frac{\dif Q}{\dif P}$ denotes the Radon-Nikodym derivative). 
Here, we are restricting our attention to the distribution shift in the inliers 
that is entirely driven by the observable covariates.
$w(x) \equiv 1$ (such that $Q = P$) is a special case that reduces back to the (vanilla) outlier detection 
problem where we select test units that are different from the 
calibration units in distribution.

We frame this as a multiple testing problem, defining the null hypotheses  $
    H_j \colon Z_{n+j} \sim Q$ for all $j\in[m]$;
i.e., $Z_{n+j}$ is an inlier. By rejecting $H_j$, we are expressing the belief 
that $Z_{n+j}$ is an outlier in the dataset. The goal is to reject a subset of these hypotheses  while controlling the FDR.

This multiple testing problem has been studied in~\citet{bates2023testing} 
when $P=Q$, then subsequently generalized by
\citet{jin2023weighted} to allow for covariate shift in the inliers.
The methods proposed therein are based on the so-called (weighted) conformal p-values~\citep{vovk2005algorithmic}.
In what follows, we introduce the methods of conformal outlier detection, 
establish their equivalent e-BH interpretations, 
and subsequently describe  how to use conditional calibration to boost their power.

Throughout, we assume that we have access to a fixed nonconformity score function $V\colon \cZ\to\RR$ 
that assigns a score to each unit in $\{Z_{i}\colon i\in[n+m]\}$ such that a larger score translates to more 
evidence of being an outlier (though, it need not be accurate to guarantee FDR control).
For example, $V$ can be the residual of some model fitted on independent hold-out data.

\subsection{Warm-up: conformal p-values and e-values}

At its very simplest, conformalized outlier detection (with no covariate shift) uses $\cD_{\calib}$ and $V$ 
to construct a p-value for each unit in $\cD_{\test}$.  
Using the shorthand $V_i \coloneqq V(Z_i)$ for $i\in[n+m]$,
the conformal p-value for $H_j$ can be constructed as follows:
\begin{align}\label{eq:conf_pval}
    p_j = \frac{1 + \sum_{i=1}^n \indc{V_i \ge V_{n+j}}  }{n+1},\quad\forall j\in[m].
\end{align}
It can be verified that under $H_j$, 
$p_{j}$ is super-uniform, making it a valid p-value. 
In addition, under the alternative hypothesis, $p_j$ is expected to be smaller---$V_{n+j}$ should tend to be larger than the scores of the ``conforming'' calibration units. 
In the case of ties, we can add randomized tiebreakers to the numerator in order 
to attain exact uniformity. 
Finally,~\cite{bates2023testing} show that although the conformal p-values 
are not independent, they are PRDS, so applying BH to the conformal p-values 
$(p_1,\ldots,p_m)$ will still control the FDR.

We now proceed to construct a collection of conformal e-values. First, define 
\begin{align}\label{eq:conf_threshold}
    T = \inf\left\{ t\in \{V_i\}_{i=1}^{n+m} \colon \frac{m}{n+1} \cdot \frac{1 + \sum_{i=1}^n \indc{V_i \ge t} }{( \sum_{j=1}^m \indc{V_{n+j} \ge t} ) \vee 1} \le \alpha  \right\}.
\end{align}
Using this threshold (which finds the smallest rejection threshold that controls an estimate of the FDR, similar to the model-X knockoff filter; see Section \ref{sec:example-knockoffs}), we can construct a conformal e-value per hypothesis:
\begin{align}\label{eq:conf_eval}
    e_j = (n+1) \cdot \frac{ \indc{V_{n+j} \ge T} }{1 + \sum_{i=1}^n \indc{ V_i \ge T }},\quad\forall j\in[m].
\end{align}
The e-value $e_j$ defined above satisfies $\EE[e_j] \le 1$ for all 
$j\in\cH_0$. Furthermore, applying e-BH to them 
will yield the same rejection set as applying BH to the conformal p-values defined 
in~\eqref{eq:conf_pval}.

\begin{proposition}
    \label{prop:conf_bh_ebh_equiv}
    Suppose $P=Q$.
    Given calibration data $\cD_{\calib}$ and test data $\cD_{\test}$, 
    construct the conformal p-values $p_1, \dots, p_m$ using~\eqref{eq:conf_pval} 
    and conformal e-values $e_1, \dots, e_m$ using~\eqref{eq:conf_eval}. 
    Let $\cR^{\bh}\coloneqq \cR^{\bh}(p_1, \dots, p_m)$ and  $\cR^{\ebh}\coloneqq \cR^{\ebh}(e_1, \dots, e_m)$ be the rejection sets obtained by BH and e-BH, respectively, at some FDR control level $\alpha \in (0,1)$. 
    The following statements hold:
    \begin{itemize}
    \item [(1)] $\EE[e_j] \le 1$ for any $j\in \cH_0$;
    \item [(2)] $\cR^{\bh}=\cR^{\ebh}$.
    \end{itemize}
\end{proposition} 
The proof of Proposition~\ref{prop:conf_bh_ebh_equiv} and additional remarks 
are provided in Appendix~\ref{appd:conf_bh_ebh_equiv}.

\subsection{Extension to covariate shift}

We now return to the more general setting, where the inliers in $\cD_\calib$ and $\cD_\test$ 
are respectively drawn from distributions $P$ and $Q$ that can be different due to a covariate shift. 
Recall that we assume that the Radon-Nikodym derivative between $P$ and $Q$ is known 
to be $w(x)$---a function of only the covariates.

For such a setting,~\cite{tibshirani2019conformal,hu2023two,jin2023weighted} 
have defined weighted analogues of the p-values in \eqref{eq:conf_pval}:
\begin{align}\label{eq:weighted_pval}
    p_j = \frac{w(X_{n+j}) + \sum_{i=1}^n w(X_i)\indc{V_i \ge V_{n+j}}}{w(X_{n+j}) 
    + \sum_{\ell=1}^n w(X_\ell)},\quad\forall j\in[m].
\end{align}
It can be shown that $p_j$ is still super-uniform and is thus valid. 
To achieve exact uniformity, we can again extend the construction 
by introducing random tie-breakers.

\cite{jin2023weighted} show that the weighted conformal p-values no longer satisfy the PRDS condition, which means BH no longer guarantees FDR control 
when applied to them. The authors propose a new procedure, 
called Weighted Conformal Selection (WCS), which does provably control the FDR. 
We instead present a more straightforward e-BH alternative that is 
almost equivalent to (if not more powerful than) WCS.
We will define weighted conformal e-values, which are weighted analogues of \eqref{eq:conf_eval}, and run e-BH. Similar to their unweighted versions, we first define the following hypothesis-specific thresholds:
for each $j\in [m]$,
\begin{equation}\label{eq:w_conf_threshold}
    T_j = \inf\bigg\{t \in \{V_i\}_{i=1}^{n+m}\colon 
    \frac{ m }{w(X_{n+j})+\sum_{i=1}^n w(X_i)} \cdot 
    \frac{ w(X_{n+j})+\sum_{i=1}^n w(X_i)  \indc{V_i \ge t }  }
    {( \sum_{k=1}^m \indc{V_{n+k} \ge t} ) \vee 1} \le \alpha \bigg\}.
\end{equation}
Using $T_j$, we then construct the e-value $e_j$:
\begin{equation}\label{eq:w_conf_eval}
    e_j = \bigg(w(X_{n+j}) +\sum_{i=1}^n w(X_i)\bigg) \cdot 
    \frac{\indc{V_{n+j} \ge T_j}}{w(X_{n+j}) + \sum_{i=1}^n w(X_i) \indc{V_i \ge T_j}}.
\end{equation}
One can notice that in the presence of no covariate shift, where $w(x)
\equiv 1$, the weighted conformal e-values coincide with their unweighted versions.
The next proposition shows that the weighted conformal e-values
defined above are still valid e-values; its proof is deferred to Appendix~\ref{appd:conf_eval_weighted}.
\begin{proposition}\label{prop:conf_eval_valid}
For each $j\in[m]$, construct $e_j$ by~\eqref{eq:w_conf_eval} 
with $T_j$ defined in~\eqref{eq:w_conf_threshold}. 
Then $\forall j\in\cH_0, \EE[e_j] = 1$.
\end{proposition}
By Proposition~\ref{prop:conf_eval_valid}, applying e-BH to the weighted conformal e-values 
controls the FDR, but unlike before, this is no longer equivalent to  
applying BH to the weighted conformal p-values.

The rejection set of e-BH, however, is (almost) identical to that of WCS with deterministic 
pruning. By ``almost identical'', we mean the following:~\citet{jin2023weighted} also propose 
an e-BH interpretation of WCS, and the corresponding e-values therein are provably no greater than 
the e-values constructed here  (but the gap is usually very small). 
A detailed discussion on the connection will be delegated to Appendix~\ref{appd:ebh_wcs}. 
Effectively, we can expect the two methods to deliver similar empirical performance.

In practice, the power of WCS can be improved at the cost of randomization: instead of deterministically pruning to get the resulting rejection set, a randomized pruning rule can be used instead~\citep{jin2023weighted}. 
The power gap between randomized and regular WCS can be quite significant in certain problem settings. This gap and the e-BH equivalence with WCS suggest that e-BH-CC has the potential to exhibit much higher power without invoking randomization. The next subsection demonstrates how we can fit the framework of conditional calibration to conformalized outlier detection.

\subsection{Conditional calibration for weighted conformal selection}

As in the other settings, we first identify the sufficient statistic for conditioning then prove that the constructed resampled e-values 
follow the desired conditional distribution. 
The proof of Proposition \ref{prop:conf_eval_valid} suggests that the 
statistic $S_j = (\cE_j,  \{Z_{n+k}\}_{k \in [m]\setminus\{j\}})$ 
can serve as a conditioning statistic, where $\cE_j$ is the unordered 
set of $\{Z_1,\ldots,Z_n, Z_{n+j}\}$ with repetitions allowed.

Formally, fix $j\in[m]$ and choose the sufficient statistic $S_j = (\cE_j,  \{Z_{n+k}\}_{k \in [m]\setminus\{j\}})$ of the combined data $\cD_\calib, \cD_\test$. Conditional on $S_j$, resample
\begin{equation}
    \label{eq:w_z_resampled}
    \tZ_{n+j} \sim Z_{n+j}\cond \cE_j,  \{Z_{n+k}\}_{k \in [m]\setminus\{j\}}\sim \sum_{Z \in \cE_j} \frac{w(X)}{\sum_{Z'\in \cE_j} w(X')} 
    \cdot \delta_Z,
\end{equation}
where $Z = (X,Y)$, $Z' = (X',Y')$, and $\delta_a$ denotes the point mass at $a$.
For the remaining elements $\{Z \in \cE_j\colon Z \neq \tZ_{n+j}\}$, 
arbitrarily assign them to $\tZ_1, \dots, \tZ_n$ without replacement by random permutation.
  Using the resampled calibration dataset $\tilde\cD_\calib^{(j)} = \{\tZ_k\}_{k \in [n]}$ 
  and resampled test dataset $\tilde\cD_\test^{(j)} = \{\tZ_{n+k}\}_{k \in [m]}$, where 
  $\tZ_{n+k} = Z_{n+k}$ for $k \neq j$,
  construct the thresholds (for each $k\in[m]$) defined in \eqref{eq:w_conf_threshold}, denoted $\tT_k^{(j)}$ to highlight its resampled nature. We then define the resampled e-values
\begin{equation}
    \label{eq:w_conf_eval_resampled}
    \te_k^{(j)} = \bigg(w(\tX_{n+k}) +\sum_{i=1}^n w(\tX_i)\bigg) 
    \cdot \frac{\ind\{V(\tZ_{n+k}) \ge \tT_k^{(j)}\}}{w(\tX_{n+k}) + 
    \sum_{i=1}^n w(\tX_i) \ind\{V(\tZ_i) \ge \tT_k^{(j)}\}},\quad\forall k\in[m],
\end{equation}
where again we denote $\tZ_i = (\tX_i,\tY_i)$ for $i \in [n+m]$.

We now state the correctness of our resampling scheme (with proof provided in Appendix~\ref{appd:conf_eval_resample}).
\begin{proposition}
    \label{prop:conf_eval_weighted}
    For each $j\in[m]$, choose sufficient statistic $S_j = (\cE_j,  \{Z_{n+k}\}_{k \in [m]\setminus\{j\}})$ and construct the thresholds $\tT_k^{(j)}$ and their corresponding e-values $\te_k^{(j)}$ as defined in \eqref{eq:w_conf_eval_resampled}. Then 
    $$
    \left(\te_1^{(j)}, \ldots, \te_m^{(j)} \right)\cond S_j \sim (e_1, \ldots, e_m) \cond S_j
    $$
    under $H_j \colon Z_{n+j} \sim Q$. 
    Furthermore, by choosing independent assignments of $(\tilde Z_{n+j}, \tZ_1, \dots, \tZ_n)$ by \eqref{eq:w_z_resampled} for each resample, the corresponding resampled e-values are i.i.d. conditional on $S_j$. 
\end{proposition}

\subsection{Simulation studies}
\label{sec:sim-conformal}

\begin{figure}[tb]
    \centering
    \includegraphics[scale=1.0]{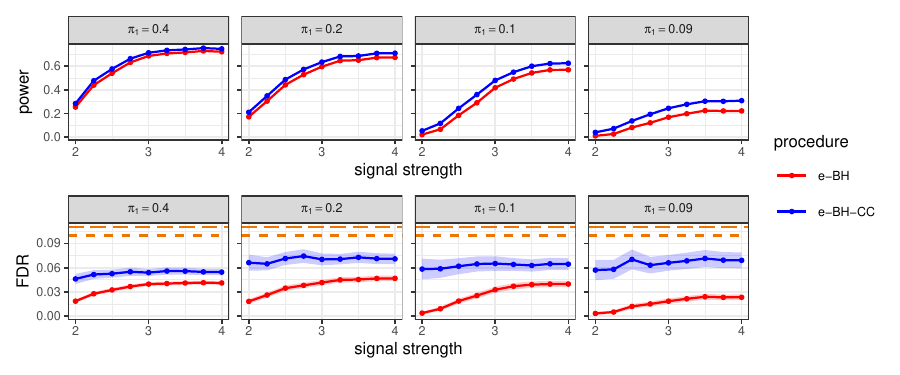} 
    \caption{Realized power and FDP of the simulated experiments for unsupervised outlier detection. Each plot contains the averaged metrics over 1,000 replications. The orange dashed lines in the FDP plots represent the target FDR (0.1) and the FDR bound with Monte-Carlo error (0.11). Shading represents error bars of two standard errors above and below.}
    \label{fig:outlier}
    \centering
  \end{figure}

\paragraph{Settings.}
We closely follow the outlier detection setup in \cite{jin2023weighted}. The authors state that as $n$ and $m$ grow larger, the selection set of WCS approaches that of BH. To show the benefit of our method, we take $n = m =200$ to avoid that limiting regime. In addition, to get a sense of the power drop and subsequent power gain for both sparse and dense outlier regimes, we vary the proportion of outliers $\pi _1\in \{ 0.09, 0.1, 0.2, 0.4 \}$.

At the beginning of each experiment, we sample 50 i.i.d.~draws from $\mathrm{Unif} ([-3,3]^{50})$ to get an initial subset of points $\cW  \subseteq \RR^{{50}}$. For each of the $ \pi_1 n$ outliers in the test set, we generate $X_{n+j} = \sqrt{a} L_{n+j} + W_{n+j}$, where $L_{n+j} \sim \cN_{50}(0, I)$,  $W_{n+j} \sim \mathrm{Unif}(\cW)$ and the signal strength $a$ varies through $\{ 2, 2.25, 2.5, \dots, 4\}$. The inliers of the test set are instead generated i.i.d. as $X_{n+j} = L_{n+j} + W_{n+j}$, whose distribution we will denote as $Q_X$. The calibration dataset, all of which are inliers, is generated i.i.d.~from $P_X$, where $\frac{\dif Q_X}{\dif P_X} (x) = w(x) \propto \sigma(x^\top \theta)$. Here, $\sigma(\cdot)$ is the sigmoid function and $\theta \in \RR^{50}$ is defined as $\theta = (0.3, 0.3, 0.2, 0.2, 0.1, 0.1, 0, \dots, 0)$. As \cite{jin2023weighted} state, this choice reflects a situation where the calibration dataset is sampled weighted by preference given by a logistic function of $X$.

\paragraph{Methods.}
We construct weighted conformal e-values \eqref{eq:w_conf_eval} using $w(\cdot)$ as described above and $V(\cdot)$ given by a one-class SVM trained on an independent hold-out set of 500 calibration inliers, multiplicatively weighted by $w(\cdot)$. We use the SVM implementation of the \texttt{scikit-learn} Python package.%
We then run e-BH and e-BH-CC at $\alpha = 0.1$ on the e-values. The Monte-Carlo budget $\alpha_0$ is set at $0.01$. Computational details are deferred to Appendix \ref{appd:deferred_details_sims_conformal}.

\paragraph{Experiments and results.}
For each choice of $\pi_1$ and $a$, we run 1,000 replications of both e-BH and e-BH-CC. Figure \ref{fig:outlier} reports their mean power and FDP over varying signal strengths for each choice of $\pi_1$. The e-BH-CC procedure is able to improve upon e-BH while continuing to control the FDR at $\alpha$, regardless of the density of outliers. It is worth noting that there is a larger power gain when the outliers are sparse in the test dataset, reflecting the notion that the difficulty of the testing problem and the looseness of e-BH are interconnected. %

\section{Real data analysis}
\label{sec:realdata}

We illustrate the empirical performance of e-BH-CC 
(the conformal selection instantiation\footnote{
We also include a separate real-data application
with the knockoffs instantiation of e-BH-CC in 
Appendix \ref{appd:hiv_knockoffs_realdata}. We show that
 e-BH-CC boosts the power of derandomized knockoffs on 
an HIV mutation dataset.}; Section~\ref{sec:example-conformal}) using the  national study of learning mindset (NSLM) observational study  
dataset~\citep{carvalho2019assessing}, with the goal of identifying individuals 
whose counterfactual outcomes satisfy certain conditions. In this study, 
the intervention is to instill students with a learning mindset---the belief that intelligence can be developed, 
as opposed to being fixed---and the outcome of interest is their academic performance. 
The NSLM dataset contains observations on $n=10,\!391$ students, among which 
$3,\!384$ students received the intervention and $7,\!007$ students did not.
Each student has  $11$ covariates, including four student-level 
covariates and seven school-level covariates (more details about the 
covariates can be found in Table 1 of~\citet{carvalho2019assessing}). 

\paragraph{Problem setup.}
For each individual, denote by $X$ the covariates, $T \in \{1,0\}$
the intervention status ($T=1$ corresponds to receiving the intervention and 
$T=0$ otherwise), and $Y \in \RR$ the academic performance. We adopt the potential outcome 
framework, defining $Y(1)$ and $Y(0)$ as the 
academic performance of the student had this student received or not received the intervention, respectively. 
Under the standard \emph{stable unit treatment value assumption (SUTVA)}~\citep{imbens2015causal}, 
the observed outcome $Y = T \cdot Y(1) + (1-T) \cdot Y(0)$. 
Throughout, we also adopt the \emph{strong ignorability assumption},
which states that
$ 
(Y(1), Y(0)) \indep T \given X.
$
Given a set of i.i.d.~data $\{(X_i, T_i, Y_i)\}_{i \in [n]}$, we seek to 
detect the individuals with large $Y(1)$ in the control group and 
those with small $Y(0)$ in the treatment group, while controlling the FDR.

To proceed, we randomly split the dataset into three folds: the training fold $\cD_\train$, 
the calibration fold $\cD_\calib$, and the test fold $\cD_\test$. In what follows, we shall 
slightly abuse the notation and let $\cD_\train$, $\cD_\calib$, and $\cD_\test$ also 
refer to the corresponding index sets. Define  $\cD_\calib(t) = \{i \in \cD_\calib: T_i = t\}$
and $\cD_\test(t) = \{i \in \cD_\test: T_i = t\}$ for $t = 0,1$. 
We consider the following two tasks with FDR controlled by $\alpha$:
\begin{enumerate}
\item [(1)] {\em Inference on the control}: find a subset of $\cD_\test(0)$ with $Y(1) > 0.3$. 
\item [(2)] {\em Inference on the treated}:  find a subset of $\cD_\test(1)$ with $Y(0)<-0.3$.
\end{enumerate}

For ease of presentation, we focus on 
our implementation for task (1), with the details for 
task (2) provided in Appendix~\ref{appd:real_data}. 
Following~\citet{jin2023weighted}, we consider this problem under the multiple testing framework, 
where the null hypothesis for any $j \in \cD_\test(0)$ is 
$H_{j}: Y_j(1) \leq 0.3$.
Here, rejecting $H_j$ corresponds to identifying an individual $j$ who 
could have an outcome greater than $0.3$ had they received the intervention. 
In slight contrast to the standard multiple testing setting, the null hypothesis 
$H_j$ is random (see~\citet{jin2023weighted} for a detailed discussion) and we focus 
on the {\em hypothesis-conditional FDR}  defined as 
\$ 
\fdr[\cR] = \EE\left[\frac{\sum_{j \in \cH_0} 
\ind\{j\in \cR\}}{|\cR|\vee 1} \Biggiven \cH_0\right].
\$

\paragraph{The conformal e-values.}
For task (1), we adopt a subset of $\cD_\calib$ as the calibration data: 
$\tD_\calib\coloneq \{i \in \cD_\calib: T_i = 1, Y_i\le 0.3\}$.
By the selection rule, for any $i \in \tD_\calib$, 
$(X_i,Y_i(1))$ follows $P_{X,Y(1)\given T=1, Y(1)\le 0.3}$.
For any $j\in \cH_0$---the inliers---we have
that $(X_j, Y_j(1))$ follows $P_{X,Y(1)\given T=0, Y(1)\le 0.3}$.
We can compute the likelihood ratio between the test and calibration inliers
to be proportional to $w(x) \,:=\, (1-e(x))/e(x)$, where $e(x) = \PP(T=1 \given X = x)$ 
is the propensity score function. 
Since the likelihood ratio is a function of only the covariates, the task can 
fit into the framework introduced in Section~\ref{sec:example-conformal}.
We provide the detailed computation of the likelihood ratio and the 
instantiation of e-BH-CC in Appendix~\ref{appd:real_data}.

\paragraph{Results.}
For each sample split, both WCS (its e-BH equivalent) with 
deterministic pruning and 
e-BH-CC are applied at FDR levels $\alpha \in\{ 0.1, 0.2, 0.3, 0.4, 0.5\}$
for both tasks (1) and (2).
For reference, we also implement WCS with heterogeneous random pruning---this is not meant to be 
compared with our method (since it is a randomized procedure), but rather to show the 
potential improvement over base WCS that e-BH-CC could achieve.
Figure~\ref{fig:nslm} reports the 
average number of discoveries over $100$ sample splits in the treated and control groups, respectively. 
We can see that e-BH-CC improves substantially over WCS with deterministic pruning, 
and is comparable to WCS with heterogeneous random pruning, underlining the message that using conditional calibration 
reclaims the power gap between the deterministic pruning and the randomized pruning.

\begin{figure}[tb]%
\centering
\includegraphics[width = 0.75\textwidth]{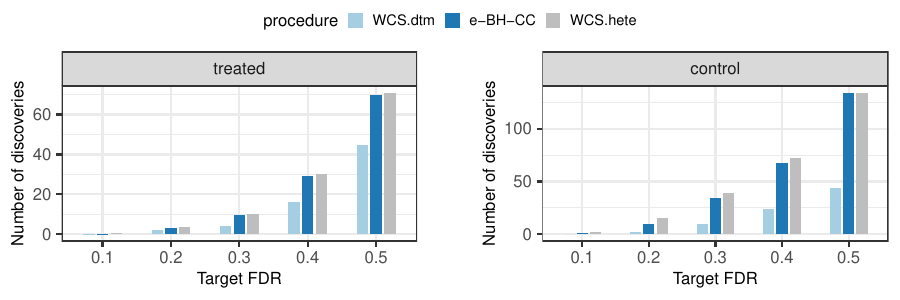}
\caption{Number of identified students with low $Y(0)$ in 
the treatment group (left) and with high $Y(1)$ in 
the control group (right) from the NSLM dataset. WCS.dtm (hrp) refers to 
WCS with deterministic (heterogeneous random) pruning 
and e-BH-CC refers to our method.
The results are averaged over $100$ sample splits.}
\label{fig:nslm}
\end{figure}

\section{Discussion}
\label{sec:discussion}
In this paper, we introduce a framework to improve the power of e-BH by conditionally 
calibrating the e-values.
Through three distinct classes of multiple testing problems, we demonstrate how the proposed 
method e-BH-CC can achieve substantial power gains (in a nonrandom fashion) while continuing to control the FDR, especially when underlying methods lack power. 
For usability, we describe a fast implementation of the conditional calibration mechanism. 

One limitation of our method is that it requires identifying a sufficient statistic $S_j$ 
for each null hypothesis $j$ such that the conditional distribution of the e-value $e_j$ 
given $S_j$ is known or can be computed.
In many settings, finding such a sufficient statistic $S_j$ can be challenging---for instance, 
when the hypotheses are highly complex and composite, or when the e-values 
are obtained from some optional sampling strategies in online settings. 
Extending conditional calibration ideas to such broader classes of e-values remains an important open problem 
that we leave for future work.

Another interesting future direction is to explore alternative boosting statistics.
Our current framework boosts the original e-value  
through a multiplicative boosting factor $\hat c_j$.
However, we can alternatively fill 
in this gap by considering an auxiliary statistic. To be specific, suppose 
we have another statistic $W_j \in \RR$ for testing $H_j$ %
that is designed so that a large value of $W_j$ suggests evidence 
against the null. Consider the ``$W_j$-assisted'' 
$\phi_j$ function:
\$
 \phi_j(c; S_j) = \EE \left[\frac{m}{\alpha} \cdot \frac{\indc{e_j \ge \frac{m}{\alpha |\hat \cR_j(\bm e)|} \text{ or } W_j \ge c}}
 {|\hat \cR_j(\bm e)|}  - e_j \bigggiven S_j \right].
\$ 
Assuming we can (numerically) evaluate $\phi_j(c;S_j)$, we proceed to 
find the critical value $\hat{c}_j$:
\$
\hat c_j \coloneq \sup\big\{c \in \RR: \phi_j(c;S_j) \le 0\big\}.
\$
This leads to the boosted e-values
\$
e_j^{W,\boost} = \frac{m}{\alpha} \cdot \frac{\indc{e_j \ge \frac{m}{\alpha |\hat \cR_j(\bm e)|} \text{ or } W_j \ge \hat c_j}}
 {|\hat \cR_j(\bm e)|}. 
\$
Since $\phi_j(+\infty;S_j) \le 0$, $\EE[e_j^{W,\boost}]\le 1$
and $\cR^\ebh(\be^{W,\boost}) \supseteq \cR(\be)$. 
The proposed e-BH-CC is a specific instance of this formulation, where
$W_j^{-1} = \frac{m}{\alpha |\hcR_j(\bm e)|} / e_j$ and the threshold would be ${1}/{c}$.
Investigating powerful choices of auxiliary statistics in different scenarios---and 
determining whether they may exhibit higher power than simply using e-BH---remains
an interesting open question.

\begingroup
\putbib
\endgroup
\end{bibunit}

\newpage
\appendix
\normalsize
\begin{bibunit}
\section{General definition of boosted e-values}
\label{appd:general_boosting_definition}

In the main text, we present e-BH-CC under the simplifying assumption that
$c\mapsto \phi_j(c;S_j)$ satisfies a boundary left-continuity condition. This section is for writing the boosting mechanism in full generality, which is useful when $\phi_j$ has discontinuities.

Retain the definition of $\hat c_j$ in~\eqref{eq:critical_value}. Since
$\phi_j(c;S_j)$ is non-decreasing in $c$, discontinuities may lead to
$\phi_j(\hat c_j;S_j)>0$. To handle this boundary case, define
\begin{align}
  \label{eq:boosted_eval_general}
  e^\boost_j =
  \begin{cases}
  \frac{m}{\alpha |\hat \cR_j(\bm e)|} \cdot \ind\big\{ \hat c_j  e_j \ge \frac{m}{\alpha |\hat \cR_j(\bm e)|}\big\}
  & \text{if } \phi_j(\hat c_j;S_j) \le 0,\\
  \frac{m }{\alpha |\hat \cR_j(\bm e)|} \cdot \ind\big\{ \hat c_j  e_j > \frac{m}{\alpha |\hat \cR_j(\bm e)|}\big\}
  & \text{if } \phi_j(\hat c_j;S_j) > 0.
  \end{cases}
\end{align}
When $\phi_j$ is left-continuous at $\hat c_j$, the first case applies and
\eqref{eq:boosted_eval_general} reduces to the main-text definition
\eqref{eq:boosted_eval}.

The same boundary issue affects the numerator evaluation in
\eqref{eq:evaluate_fast_coef}. In the general case, let
\[
  \tilde c_j = \frac{m}{\alpha | \hat\cR_j(\be)|}/e_j.
\]
Then when $\{\phi_j(\hat c_j;S_j)\le 0\}$ is true, 
\begin{equation}
\begin{aligned}\label{eq:evaluate_fast_coef_general}
\ind\bigg\{\hat c_j e_j \ge \frac{m}{\alpha | \hat\cR_j (\be)| }\bigg\}
=
\ind\bigg\{\hat c_j\ge \frac{m}{\alpha | \hat\cR_j (\be)| }/ e_j\bigg\}
=
\ind\{ \phi_j(\tilde c_j; S_j) \le 0\};
\end{aligned}
\end{equation}
when $\{\phi_j(\hat c_j;S_j)> 0\}$ is true, 
\begin{equation}
\begin{aligned}
\ind\bigg\{\hat c_j e_j > \frac{m}{\alpha | \hat\cR_j (\be)| }\bigg\}
=
\ind \bigg\{\hat c_j> \frac{m}{\alpha | \hat\cR_j (\be)| }/ e_j\bigg\}
=
\ind \{ \phi_j(\tilde c_j; S_j) \le 0\}.
\end{aligned}
\end{equation}
The final equivalence in both lines follows because the sublevel set
$\{c\colon \phi_j(c;S_j)\le0\}$ is an interval with right endpoint
$\hat c_j$, including the endpoint exactly when
$\phi_j(\hat c_j;S_j)\le0$. Thus the general construction can still
evaluate the numerator through the sign of $\phi_j(\tilde c_j;S_j)$; the
strict inequality in \eqref{eq:boosted_eval_general} only handles whether the
boundary point itself is included.

\section{Technical proofs}

\subsection{Proof of Proposition~\ref{prop:sharp-eval-fdr}}
\label{app:proof-of-sharp-eval-fdr}
Throughout the proof, we use $\cR$ as a shorthand for $\cR^\ebh(\be)$.
Suppose $e_1,\ldots,e_m$ are BH-sharp. For any $j \in \cH_0$, 
if $e_j = \frac{m}{\alpha |\cR|}$, then 
\@\label{eq:fdr-equiv}
\frac{\ind\{e_j \ge m/(\alpha |\cR|\})}{|\cR| \vee 1} = \frac{\alpha e_j}{m}.
\@
The above also holds when $e_j = 0$. Since $\PP(e_j \in \{0, m/{\alpha |\cR|}\})=1$,~\eqref{eq:fdr-equiv} 
holds almost surely and  
\$ 
\fdr(\cR^\ebh(\be)) = \sum_{j \in \cH_0}\EE\Big[\frac{\ind\{e_j \ge m/(\alpha |\cR|)\}}{|\cR|\vee 1}\Big] 
& =\frac{\alpha}{m} \sum_{j\in \cH_0}\EE[e_j].
\$
We now prove the other direction. Suppose that 
\$ 
\fdr(\cR^\ebh(\be)) = \frac{\alpha}{m} \sum_{j \in \cH_0} \EE[e_j].
\$
Assume for the sake of contradiction that there exists $k\in \cH_0$ such that 
the event of $B := \{e_k \in \{0, m/(\alpha |\cR|)\}\}$ happens with probability $<1$.
On $B^c$, 
\$ 
\frac{\ind\{e_k \ge m/(\alpha |\cR|)\}}{|\cR|\vee 1} < \frac{\alpha e_k}{m}. 
\$
As a result, 
\$ 
\EE\Big[\frac{\ind\{e_k \ge m/(\alpha |\cR|)\}}{|\cR|\vee 1}\Big]
& = \EE\Big[\frac{\ind\{e_k \ge m/(\alpha |\cR|)\}}{|\cR|\vee 1}\ind\{B\}\Big]
+ \EE\Big[\frac{\ind\{e_k \ge m/(\alpha |\cR|)\}}{|\cR|\vee 1}\ind\{B^c\}\Big]\\
& < \EE\Big[\frac{\alpha e_k}{m}\ind\{B\}\Big]
+ \EE\Big[\frac{\alpha e_k}{m}\ind\{B^c\}\Big] = \frac{\EE[e_k]}{m}.
\$
The above further implies that $\fdr < \sum_{j \in \cH_0} \alpha \EE[e_j] / m$, which is a contradiction. 
So the assumption can not be true, and we complete the proof.

\subsection{Proof of Theorem~\ref{thm:eval_valid}}
\label{appd:proof_eval_valid}
The following argument proves validity for the general definition in
\eqref{eq:boosted_eval_general}; the main-text theorem follows as the
left-continuous special case.
Note that $\hat c_j$ and $\phi_j(\hat c_j;S_j)$ are both deterministic functions of $S_j$.
For any $j \in \cH_0$, we can write the expectation of $e_j^\boost$ as 
\$ 
\EE[e_j^\boost] = \EE\Big[\ind\{\phi_j(\hat c_j;S_j) \le 0\} \cdot \EE[e_j^\boost \given S_j] 
+ \ind\{\phi_j(\hat c_j;S_j) > 0\}\cdot \EE[e_j^\boost \given S_j] \Big].
\$
When $\phi_j(\hat c_j;S_j)\le 0$, by construction of $e_j^\boost$ we have
\$ 
\EE[e_j^\boost \given S_j] =
\EE\Bigg[\frac{m}{\alpha |\hcR_j(\be)|} \ind\bigg\{\hat c_j e_j \ge \frac{m}{\alpha |\hcR_j(\be)|}\bigg\}
\bigggiven S_j\Bigg] = \phi_j(\hat c_j;S_j) + \EE[e_j \given S_j] \le \EE[e_j \given S_j]. 
\$
When $\phi_j(\hat c_j;S_j) > 0$, we let $\hat c_{j,t} = \hat c_j - 1/t$. 
Then we have deterministically that 
\$ 
\ind\big\{\hat c_j \cdot e_j > \tfrac{m}{\alpha|\hcR_j(\be)|}\big\}
= \lim_{t \to \infty}
\ind\big\{\hat c_{j,t} \cdot e_j \ge \tfrac{m}{\alpha|\hcR_j(\be)|}\big\}.
\$
As a result,
\$
\EE[e_j^\boost \given S_j] &= 
\EE\Bigg[\frac{m}{\alpha |\hcR_j(\be)|} \ind\bigg\{\hat c_j e_j > \frac{m}{\alpha |\hcR_j(\be)|}\bigg\}\bigggiven S_j\Bigg] \\
& = 
\EE\Bigg[\lim_{t \to \infty}\frac{m}{\alpha |\hcR_j(\be)|} \ind\bigg\{\hat c_{j,t} e_j \ge \frac{m}{\alpha |\hcR_j(\be)|}\bigg\}\bigggiven S_j\Bigg] \\
& \stepa{=} \lim_{t \to \infty}
\EE\Bigg[\frac{m}{\alpha |\hcR_j(\be)|} \ind\bigg\{\hat c_{j,t} e_j \ge \frac{m}{\alpha |\hcR_j(\be)|}\bigg\}\bigggiven S_j\Bigg] \\
& = \lim_{t \to \infty} \big( \phi(\hat c_{j,t} ; S_j) + \EE[e_j \given S_j] \big)\\ 
& \stepb{\le} \EE[e_j \given S_j].
\$
Above, step (a) follows from the dominated convergence theorem and step (b)
follows from the definition of the critical value $\hat c_j$.

Combining the two cases, we have $\EE[e_j^\boost] \le \EE[e_j]$ for all $j \in \cH_0$.
Since $(e_1,\ldots,e_m)$ are valid (resp. compound) e-values, 
the boosted e-values are valid (resp. compound) e-values. 

\subsection{Proof of Theorem~\ref{thm:power-improvement}}
\label{appd:proof_power_improvement}
  We prove the stronger statement for the general definition in
  \eqref{eq:boosted_eval_general}; the left-continuous main-text definition is
  the first case of that construction.
  The claim mainly follows from the fact that $\phi_j(1;S_j) \le 0$ for each $j$. To see why, we invoke the inequality 
  that $\ind\{X \ge t\} \le X/t$ for $t >0$ on the indicator inside the function $\phi_j$:
  \begin{equation}
    \indc{  \te_j \ge \frac{m}{\alpha |\hat \cR_j(\bm  \te)|}} \le 
    \frac{\alpha | \hat \cR_j(\bm \te)|}m \cdot  \te_j.
  \end{equation}
  Therefore, 
  \begin{equation} 
    \phi_j(1;S_j) \le \EE\left[ \frac{m}{\alpha |\hat \cR_j(\bm \te)|} \cdot 
    \frac{\alpha | \hat \cR_j(\bm \te) |}m \cdot  \te_j - \te_j\bigggiven S_j\right] = 0.
  \end{equation}
  As a result, when $\phi_j(\hat c_j;S_j) \le 0$, we have $\hat c_j \ge 1$; when 
  $\phi_j(\hat c_j;S_j) > 0$, we have $\hat c_j > 1$. 
  
  We will now show that $j\in \cR(\bm e) \implies j \in \cR(\bm e^\boost)$. The case when $ \cR(\bm e)$ is empty is trivial, so we assume otherwise. For each $j \in \cR(\bm e)$, the containment $\hat\cR_j(\bm e) \supseteq \cR(\bm e)$ actually attains set equality, so
  \begin{equation}
  e^\boost_j = 
  \begin{cases}
  \frac{m}{\alpha |\cR(\bm e)|} \cdot \ind\big\{ \hat c_j  e_j \ge \frac{m}{\alpha |\cR(\bm e)|}\big\} 
  & \text{if } \phi_j(\hat c_j;S_j) \le 0,\\
  \frac{m }{\alpha |\cR(\bm e)|} \cdot \ind\big\{ \hat c_j  e_j > \frac{m}{\alpha |\cR(\bm e)|}\big\}
  & \text{if } \phi_j(\hat c_j;S_j) > 0.
  \end{cases}
  \end{equation}
  Recall that if $\phi_j(\hat c_j;S_j) \le 0$, then $\hat c_j \ge 1$, and 
  $\hat c_j > 1$ otherwise.
  In either case, the indicator in the above crystallizes to 1 
  since $e_j \ge \frac{m}{\alpha |\cR(\bm e)|}$ by virtue of $j$ having 
  been selected by e-BH. Therefore, $e_j^\boost = \frac{m}{\alpha|\cR(\bm e)|}$ 
  for each $j\in \cR(\bm e)$. Since there are $|\cR(\bm e)|$ many such indices, the e-BH procedure will select the set $\{ e_j^\boost \colon j\in\cR(\bm e) \}$ at the very least, proving the claim.

\subsection{Proof of Proposition~\ref{prop:mc_error}}
\label{appd:proof_mc_error}
For each $j\in \cH_0$, we have 
\$
\EE\big[e_j^{\boost,\ci}\big]
& =  \EE\bigg[\frac{m\ind\{U_{j,K}\le 0, \phi_j(\tilde{c}_j;S_j) > 0\} }{\alpha |\hcR_j(\be)|}\bigg]
+ \EE\bigg[\frac{m\ind\{U_{j,K}\le 0, \phi_j(\tilde{c}_j;S_j) \le 0\} }{\alpha |\hcR_j(\be)|}\bigg].
\$
Above, the first term can be bounded as follows:
\$
\EE\bigg[\frac{m \ind\{\phi_j(\tilde c_j; S_j) > U_{j,K}\}}{\alpha |\hcR_j(\be)|}\bigg] 
=  
\EE\Bigg[\frac{m}{\alpha |\hcR_j(\be)|} 
\PP\Big(\phi_j(\tilde c_j; S_j) \notin C_{j,K} \given S_j, \be\Big)\Bigg] 
\le  \EE\bigg[\frac{m \cdot \alpha_\ci}{\alpha |\hcR_j(\be)|}\bigg] \le \alpha_0/\alpha.
\$
As for the second term, we have
\$
\EE\bigg[\frac{m\ind\{U_{j,K}\le 0, \phi_j(\tilde{c}_j;S_j) \le 0\} }{\alpha |\hcR_j(\be)|}\bigg]
& \le \EE\bigg[\frac{m\ind\{\phi_j(\tilde{c}_j;S_j) \le 0\} }{\alpha |\hcR_j(\be)|}\bigg]\\
& =
\underbrace{\EE\bigg[\frac{m \ind\{\tilde{c}_j \le \hat c_j, \phi_j(\hat c_j;s_j)\le 0\}}{\alpha |\hcR_j(\be)|}\bigg]}_{\text{(i)}} \\
&\quad + \underbrace{\EE\bigg[\frac{m \ind\{\tilde{c}_j < \hat c_j, \phi_j(\hat c_j;s_j)> 0\}}{\alpha |\hcR_j(\be)|}\bigg]}_{\text{(ii)}}. 
\$ 
Recalling that $\tilde c_j = \frac{m}{\alpha |\hcR_j(\be)|}/e_j$, we then have
\$ 
\text{(i)} & =  \EE\bigg[\frac{m \ind\{\hat c_j e_j \ge \frac{m}{\alpha |\hcR_j(\be)|}, \phi_j(\hat c_j;s_j)\le 0\}}{\alpha |\hcR_j(\be)|}\bigg] \\
& = \EE\Big[\ind\big\{\phi_j(\hat c_j;S_j) \le 0\big\} \cdot 
\big(\phi_j(\hat c_j; S_j) + \EE[e_j\given S_j]\big)\Big]\\
& \le \EE\big[\ind\{\phi_j(\hat c_j;S_j)\le 0\} e_j\big].
\$
Next, we again take $\hat c_{j,t} = \hat c_j - 1/t$, and then we have
\$ 
\text{(ii)} & = \EE\bigg[\frac{m \ind\big\{\hat c_j e_j > \frac{m}{\alpha |\hcR(\be)|}, \phi_j(\hat c_j;s_j)> 0\big\}}{\alpha |\hcR_j(\be)|}\bigg]\\
& =\EE\Bigg[\ind\{\phi_j(\hat c_j; S_j) > 0\}\cdot 
\EE\bigg[\lim_{t \to \infty} \frac{m\ind\big\{\hat c_{j,t}e_j \ge \frac{m}{\alpha|\hcR_j(\be)|}\big\}}{|\hcR_j(\be)|}\bigggiven S_j\bigg]\Bigg]\\
& \stepa{=} \lim_{t \to \infty} \EE\Bigg[\ind\{\phi_j(\hat c_j; S_j) > 0\}\cdot 
\EE\bigg[\frac{m\ind\big\{\hat c_{j,t}e_j \ge \frac{m}{\alpha|\hcR_j(\be)|}\big\}}{|\hcR_j(\be)|}\bigggiven S_j\bigg]\Bigg]\\
& = \lim_{t \to \infty} \EE\Big[\ind\{\phi_j(\hat c_j; S_j) > 0\}\cdot 
\EE\big[\phi_j(\hat c_{j,t}) + e_j\biggiven S_j\big]\Big]\\
& \le \EE\big[\ind\{\phi_j(\hat c_j; S_j) > 0\} e_j \big],
\$
where step (a) applies the dominated convergence theorem. Combining (i) and (ii), we have 
bounded the second term by $\EE[e_j]$, and therefore $\EE[e_j^{\boost,\ci}] \le \EE[e_j]+\alpha_0/\alpha$ for any $j\in \cH_0$, 
invoking the proof of e-BH completes the proof.

\paragraph{Proof of Corollary~\ref{cor:asymptotic_mc_error}.}
If $C_{j,K}$ is a $(1-\alpha_\ci)$ asymptotic confidence interval for $\phi_j(\hat c_j;S_j)$,  
it suffices to modify the upper bound of the first term in the proof of Proposition~\ref{prop:mc_error}:
\$
\lim_{K\to \infty}
\EE\bigg[\frac{m \ind\{\phi_j(\tilde c_j; S_j) > U_{j,K}\}}{\alpha |\hcR_j(\be)|}\bigg] 
& =  \lim_{K\to \infty}\EE\Bigg[\frac{m}{\alpha |\hcR_j(\be)|} 
\PP\Big(\phi_j(\tilde c_j; S_j) \notin C_{j,K} \given S_j, \be\Big)\Bigg]\\ 
& = \EE\Bigg[\frac{m}{\alpha |\hcR_j(\be)|} 
\lim_{K\to \infty} \PP\Big(\phi_j(\tilde c_j; S_j) \notin C_{j,K} \given S_j, \be\Big)\Bigg]\\ 
& \le  \EE\bigg[\frac{m \cdot \alpha_\ci}{\alpha |\hcR_j(\be)|}\bigg] \le \alpha_0/\alpha.
\$
Again, the second step follows from the dominated convergence theorem.

\subsection{Proof of Proposition~\ref{prop:zstat}}
\label{appd:proof_zstat_eval_resample}
It is sufficient to show that the resample $\tZ$ follows the conditional joint distribution $Z\cond S_j$ under $H_j \colon \mu_j  = 0$. Since $\cov(Z_j, S_j) = 0$, the two entities are independent. Thus, the independently sampled $\tZ_j = Y \sim \cN(0, \Sigma_{jj})$, which was drawn without any knowledge of $S_j$, follows $Z_j\cond S_j$ automatically.

Define the deterministic function $g\colon \RR\times\RR^{m-1} \to \RR^{m-1} $ such that $g(z, s) = s + \Sigma_{-j,j} z$. Then $Z_{-j} = g(Z_j, S_j)$ and $\tZ_{-j} = g(Y, S_j)$ (by \eqref{eq:resample_zstat}). Therefore, $$Z_j \cond S_j \sim Y \cond S_j \implies (Z_{-j}, Z_j)  \cond S_j\sim (\tZ_{-j}, \tZ_j) \cond S_j$$
and we conclude as desired.

\subsection{Proof of Proposition~\ref{prop:conf_bh_ebh_equiv}}
    \label{appd:conf_bh_ebh_equiv}
\subsubsection{Proof of (1)}
Define $\cV$ to be the unordered collection of nonconformity scores 
$\big\{V_i: i\in[n+m]\big\}$. 
Let $\pi$ be a permutation of $[n+m]$ such that
$V_{\pi(1)} \le \cdots \le V_{\pi(n+m)}$, and let 
    $N(k) := \sum_{i\in[n]} \ind\{V_i \ge V_{\pi(k)}\}$, 
    $R_0(k) = \sum_{j \in \cH_0} \ind\{V_{n+j} \ge V_{\pi(k)}\}$, and 
    $R_1(k) = \sum_{j \in \cH_1} \ind\{V_{n+j} \ge V_{\pi(k)}\}$.
    Consider the discrete time filtration, 
    \@\label{eq:eval_filtr}
    \cF_k = \sigma\big(\cV,\{N(\ell)\}_{\ell\le k},
    \{R_0(\ell)\}_{\ell\le k},\{R_1(\ell)\}_{\ell\le k}\big), 
    \text{ for }k\ge 1.
    \@
    Instead of directly proving (a), we show the following stronger 
    result. Suppose the conformal e-value takes the 
    following form: 
    \@\label{eq:conf_eval_general_form}
    e_j = (n+1)\frac{\ind\{V_{n+j} \ge T\}}{1+\sum_{i\in [n]}
     \ind\{V_i \ge T\}},
    \text{ for }j \in [m].
    \@
    Then $e_j$ is an e-value as long as 
    $T = V_{\pi(\tau)}$, 
    where is $\tau$ a stopping time adapted to the filtration 
    $\{\cF_k\}_{k \ge 1}$.
    Apparently, the threshold $T$ defined in~\eqref{eq:conf_threshold} 
    satisfies the condition of Proposition~\ref{prop:conf_bh_ebh_equiv_general}, 
    and therefore the e-value is a strict e-value.  

    The stronger result is stated in the following proposition, followed 
    by its proof. %
    \begin{proposition}
    \label{prop:conf_bh_ebh_equiv_general}
    Under the setting of Proposition~\ref{prop:conf_bh_ebh_equiv},
    let $\tau$ be a stopping time adapted to the filtration $\{\cF_t\}_{k \ge 1}$
    defined in \eqref{eq:eval_filtr}.
    Suppose the e-value takes the form defined in~\eqref{eq:conf_eval_general_form}.
    Then we have $\EE[e_j] \le 1$, for any $j\in \cH_0$.
    \end{proposition}
    \begin{proof}
    To start, we define for $k \in \mathbb{N}_+$ that 
    \$
    M(k) = \frac{\sum_{j\in \cH_0} \ind\{V_{n+j} \ge V_{\pi(k)}\}}{1+\sum_{i \in [n]}
    \ind\{V_i \ge V_{\pi(k)}\}} = \frac{R_0(k)}{1+N(k)},
    \$
    where the last step is by the definition of $R_0(k)$ and $N(k)$.
    We claim that $M(k)$ is a supermartingale with respect to the 
    filtration $\{\cF_k\}_{k \ge 1}$. To see this, consider the following 
    two scenarios. 
    \begin{enumerate}
    \item [(1)] When $V_{\pi(k)} = V_{\pi(k-1)}$, we have 
    $N(k) = N(k-1)$ and $R_0(k) = R_0(k-1)$, and thus 
    $\EE[M(k)\given \cF_{k-1}] = M(k-1)$.
    \item [(2)] When $V_{\pi(k)} > V_{\pi(k-1)}$,
    we let $\Delta(k) = \sum_{i\in[n+m]} \ind\{V_i = V_{\pi(k-1)}\}$, 
    By the exchangeability of the inliers, conditional on $\cF_{k-1}$, 
    $N(k)$ follows 
    a hypergeometric distribution with parameters 
    $N(k-1) + R_0(k-1)$ (population size), $N(k-1)$ (success states in the population), 
    and $N(k-1)+R_0(k-1) -\Delta(k)$ (number of draws).
    As a result, by direct computation of the 
    conditional expectation with respect to the 
    hypergeometric distribution (see e.g.,~\citet[Lemma 3.2]{weinstein2017power}), 
    we have $\EE[M_k \given \cF_{k-1}] \le M_{k-1}$.
    \end{enumerate}
    Combing the two cases, we have shown that $M(k)$ is a supermartingale. 
    Since $\tau$ is a stopping time with respect to $\{\cF_k\}_{k\ge 1}$, 
    applying the optional stopping theorem leads to 
    \@ \label{eq:eval_martingale_1}
    \EE[M_\tau] \le M_1 = \frac{|\cH_0|}{1+n}.
    \@
    Note also that 
    \@\label{eq:eval_martingale_2}   
    \EE\bigg[\frac{\sum_{j\in \cH_0} \ind\{V_{n+j}\ge T\}}
    {\sum_{i\in [n]} \ind\{V_i \ge T\}}\bigg]  
    = \frac{1}{1+n}\sum_{j\in \cH_0}\EE[e_j] 
    \@ 
    Combining~\eqref{eq:eval_martingale_1} and~\eqref{eq:eval_martingale_2} 
    yields $\sum_{j\in \cH_0}\EE[e_j] \le |\cH_0|$. Note also that $\tau$ is 
    invariant to the permutation on $\{Z_{n+j}\}_{j \in \cH_0}$, and thus
    by the exchangeability of $\{Z_{n+j}\}_{j \in \cH_0}$, we have 
    $\EE[e_j] \le 1$ for any $j\in \cH_0$.

    \end{proof}

    \subsubsection{Proof of (2)}
    The statement is immediate from the following lemmas about $T$, the threshold defined in \eqref{eq:conf_threshold}.
    \begin{lemma}\label{lem:bh_to_ebh-1}
        $\hcR^{\bh} = \{j \in [m] \colon {V_{n+j} \ge T}\}$.
    \end{lemma}
    \begin{proof}
        Let $V_{n+(1)}\ge V_{n+(2)} \ge \cdots \ge V_{n+(m)}$ denote the 
        ordered nonconformity test scores in the descending order (with 
        ties broken arbitrarily). By the definition of the conformal p-values, 
        we also have $p_{(1)} \le p_{(2)} \le \cdots \le p_{(m)}$.
        
        When  $\hcR^\bh$ is empty, then for any $t \le V_{n+(1)}$, let 
        $m(t) = \sum_{j \in [m]} \ind\{V_{n+j} \ge t\}$. By the definition of 
        $m(t)$, one can check that $V_{n+(m(t))} \ge t$  and 
        \$
        \frac{m}{n+1}\frac{1+\sum_{i\in [n]}\ind\{V_i \ge t\}}
        {\sum_{j\in[m]} \ind\{V_{n+j}\ge t\}}
        =  
        \frac{m}{n+1}\frac{1+\sum_{i\in [n]}\ind\{V_i \ge t\}}
        {m(t)} 
         = \frac{m p_{m(t)}}{m(t)} > \alpha. 
        \$
        The last step is because $\hcR^\bh = \varnothing$. 
        By the definition of $T$, we have $T > V_{n+{(1)}}$, 
        and thus $\hcR^\bh = \{j\in[m]: V_{n+j} \ge T\}$. 

        When $\hcR^{\bh}$ is nonempty, 
        BH rejects the hypotheses corresponding to the 
        $R \coloneqq |\hcR^{\bh}|$ largest nonconformity test scores, where
        $$
            R = \max\left\{ r\colon  \frac{1 + \sum_{i=1}^n \indc{V_i \ge V_{n+(r)}}  }{n+1} \le \frac{\alpha r}m \right\}.
        $$ 
        Observe that the rejection rule is equivalent to rejecting all nonconformity test scores that are at least 
        $$
            \hat{v} = \max\left\{ v\colon  \frac{1 + \sum_{i=1}^n \indc{V_i \ge v}  }{n+1} \le \frac{\alpha \cdot |\{j\in[m]\colon V_{n+j} \ge v\}| }m \right\}.
        $$
        We can rearrange the condition inside the maximum for the equivalent statement
        $$
            \hat{v} = \max\left\{ v\colon  \frac{m}{n+1} \cdot \frac{1 + \sum_{i=1}^n \indc{V_i \ge v}  }{|\{j\in[m]\colon V_{n+j} \ge v\}|} \le \alpha  \right\}.
        $$
        When the rejection set is non-empty, $\hat{v}$ and $T$ coincide. 
        Thus, $\hcR^{\bh} = \{j\colon {V_{n+j} \ge T}\}$. 
    \end{proof}
    \begin{lemma}\label{lem:bh_to_ebh-2}
        $\hcR^{\ebh} = \{j \in [m]\colon {V_{n+j} \ge T}\}$.
    \end{lemma}
    \begin{proof}
        First, note that when $V_{n+j} < T$, the corresponding e-value is 0 and will never be rejected. 
        Thus, we have $\hcR^{\ebh} \subseteq \{j \in [m]\colon {V_{n+j} \ge T}\}$. 
        To prove the reverse inclusion, we can assume without loss of generality that 
        $\{j \in [m] \colon {V_{n+j} \ge T}\}$ is nonempty (since otherwise the inclusion is trivial).
         Observe that when $V_{n+j}\ge T$, the e-value is positive and takes the value
        $$
            e_j = \frac{n+1}{1 + \sum_{i=1}^n \indc{V_i\ge T}}
        $$
        which does not depend on the index $j$ (i.e., the nonzero e-values take the same value). Taking the definition of $T$ directly, we see that 
        $$  
            \frac{n+1}{ 1 +   \sum_{i=1}^n \indc{V_i\ge T}} \ge \frac {m}{\alpha\cdot |\{j\in[n]\colon V_{n+j} \ge T\}|}.
        $$
        However, $R\coloneqq  |\{j\in[n]\colon V_{n+j} \ge T\}|$ is exactly the number of nonzero e-values in our collection. Each of the $R$ nonzero e-values take the same value, which is bounded below by $\frac{m}{aR}$. By the e-BH procedure, each of these e-values will be selected, proving the reverse inclusion $\hcR^{\ebh} \supseteq \{j\colon {V_{n+j} \ge T}\}$.
    \end{proof} 
    \noindent Since the two lemmas show equality of 
    $\hcR^\bh$ and $\hcR^\ebh$ to the same rejection set, we conclude the desired proposition.
\subsubsection{Additional remarks}
\begin{remark}
    A nearly-identical e-value construction
    was proposed in~\cite{bashari2024derandomized}, whose authors use a martingale 
    argument to prove their compound e-value validity (Definition~\ref{def:compound_eval}).
    By slightly modifying their proof strategy 
    in the form of~\eqref{eq:conf_eval} are strict e-values, for a general class of 
    thresholds including~\eqref{eq:conf_threshold} and the one in~\citet{bashari2024derandomized}.
    The threshold in~\eqref{eq:conf_threshold} is a special case that leads to the 
    equivalence between BH and e-BH. This argument is put forth in the proof of Proposition \ref{prop:conf_bh_ebh_equiv_general}, found in the current subsection.
\end{remark}
    
\begin{remark}
  Recall that e-BH is efficient when the e-values are BH-sharp (Section \ref{sec:ebh}).
    With the threshold in~\eqref{eq:conf_threshold},
    $e_j$ is either $0$ or $(n+1)/(1+\sum_{i\in[n]}\ind\{V_i \ge T\})$;
    the nonzero value is very close to the e-BH threshold $m/(\alpha 
    \sum_{j}\ind\{V_{n+j}\ge T\})$ by the definition of $T$.
    This, however, is no longer the case when we extend this result to the weighted conformal e-values, justifying e-BH-CC as a more powerful option.
\end{remark}
\begin{remark}
As a side note, the selection set $\hcR^\bh$ is also 
proposed by~\cite{gao2023simultaneous} in a slightly different context.
One might also notice that a third way to define the selection set is via 
$\hcR^{\text{thres}} \coloneq \{j\in[m]: V_{n+j} \ge T\}$, which is independently 
proposed in \citet{weinstein2017power} and~\citet{mary2022semi} under 
different settings. 
We find that this selection set $\hcR^{\text{thres}}$ is also equivalent to 
$\hcR^{\bh}$ and $\hcR^{\ebh}$, which is implied by 
our proof of Proposition~\ref{prop:conf_bh_ebh_equiv}.
\end{remark}

\subsection{Proof of Proposition~\ref{prop:conf_eval_valid}}
    \label{appd:conf_eval_weighted}

    For each $j\in[m]$, let us define an alternative threshold to \eqref{eq:w_conf_threshold}:
    \begin{equation}
        \label{eq:w_conf_threshold_alt}
        \begin{aligned}
        \hat T_j = \inf\Bigg\{ &t \in \{V_i\}_{i=1}^{n+m} \colon \\
            &\frac{ m }{w(X_{n+j})+\sum_{i=1}^n w(X_i)} \cdot 
            \frac{w(X_{n+j})\indc{V_{n+j} \ge t}+\sum_{i=1}^n w(X_i)  \indc{V_i \ge t }  }
            {1 + \sum_{k \in [m]\setminus\{j\}} \indc{V_{n+k} \ge t} }\le \alpha
        \Bigg\}.
        \end{aligned}
    \end{equation}
    The significance of the above threshold is that on the event $\{V_{n+j} \ge T_j\}$, $T_j = \hat T_j$. 
    To see this, we first note that by construction, $\hat{T}_j \le T_j$.  
    On the event $\{V_{n+j} \ge T_j\}$, there is also $V_{n+j} \ge \hat T_j$.
    To prove the inverse inequality, consider the following quantity 
    \$ 
    & \frac{m}{w(X_{n+j})+\sum_{i=1}^n w(X_i)} \cdot 
    \frac{w(X_{n+j})+\sum_{i=1}^n w(X_i)  \ind\{V_i \ge \hat T_j \}  }
    {\ind\{V_{n+j} \ge \hat T_j\} + \sum_{k \in [m]\setminus\{j\}} \ind\{V_{n+k} \ge \hat T_j\} }\\
    =~& 
    \frac{m}{w(X_{n+j})+\sum_{i=1}^n w(X_i)} \cdot 
    \frac{w(X_{n+j})\ind\{V_{n+j} \ge \hat T_j\} + \sum_{i=1}^n w(X_i)  \ind\{V_i \ge \hat T_j \}  }
    {1 + \sum_{k \in [m]\setminus\{j\}} \ind\{V_{n+k} \ge \hat T_j\} } \le \alpha,
    \$
    where the first step is because $V_{n+j} \ge T_j \ge \hat T_j$ and the second step 
    is by the definition of $\hat T_j$. We can then conclude that $T_j \le \hat T_j$, 
    and therefore $T_j = \hat T_j$.

    In addition, define the set $\cE_j$ to be the unordered collection of $\{Z_1, Z_2, \dots, Z_n, Z_{n+j}\}$ 
    (with repetitions allowed). 
    Then $\hat T_j$, formally constructed using each of $Z_1, \dots, Z_{n+m}$, actually only depends on the data through $\cE_j$ and $\{Z_{n+k}\}_{k \in [m]\setminus\{j\}}$. That is, $\hat T_j$ is agnostic to the ordering of the elements in $\cE_j$.

    Using these facts, we can analyze the expectation of $e_j$ under the null $H_j$:
    \begin{align}\label{eq:weighted_conf_eval_proof}
        \begin{split}
            \EE[e_j] &= \EE\left[  \frac{\big(w(X_{n+j}) +\sum_{i=1}^n w(X_i)\big) \indc{V_{n+j} \ge T_j}}
            {w(X_{n+j}) + \sum_{i=1}^n w(X_i) \indc{V_i \ge T_j}} \right]\\
            &=   \EE\left[  \frac{\big(w(X_{n+j}) +\sum_{i=1}^n w(X_i)\big) \indc{V_{n+j} \ge T_j}}
            {w(X_{n+j})\indc{V_{n+j}\ge T_j} + \sum_{i=1}^n w(X_i) \indc{V_i \ge T_j}} \right]\\
            &=  \EE\left[  \frac{\big(w(X_{n+j}) +\sum_{i=1}^n w(X_i)\big) \ind\{V_{n+j} \ge  \hat T_j\}}
            {w(X_{n+j})\ind\{V_{n+j}\ge \hat T_j\} + \sum_{i=1}^n w(X_i) \ind\{V_i \ge \hat T_j\}} \right]\\
            &= \EE \Biggl[
                \frac{w(X_{n+j}) +\sum_{i=1}^n w(X_i) }{w(X_{n+j})\ind\{V_{n+j}\ge \hat T_j\} 
                + \sum_{i=1}^n w(Z_i) \ind\{V_i \ge \hat T_j\}}\\ &\qquad
                \qquad \times 
                \EE\bigg[\indc{V_{n+j} \ge \hat T_j} \bigggiven \cE_j, \{Z_{n+k}\}_{k \in [m]\setminus\{j\}}\bigg]
            \Biggr].
        \end{split}
    \end{align}
    The second equality technically adopts the notation $0/0 = 0$. The third equality follows from the equivalence of $T_j$ and $\hat T_j$ on the event $\indc{V_{n+j} \ge T_j}=1$ as discussed above.
    The last step uses the tower property of conditional expectation.

    Under the null $H_j$,
    we have following characterization of the conditional distribution of $Z_{n+j}$ given $\cE_{j}$~\citep{jin2023weighted}:
    \begin{equation}\label{eq:weighted_exch_cond}
        Z_{n+j}\cond \{\cE_j = z\}\sim \sum_{k \in [n] \cup \{n+j\}}  \frac{w(z_k )}{\sum_{i=1}^n w(z_i) + w(z_{n+j})} 
        \cdot \delta_{z_k},
    \end{equation} 
    where $\delta_a$ denotes a point mass at $a$.
    In the above, $z = \{z_1, z_2, \dots, z_n, z_{n+j}\}$ is a realization of $\cE_j$, with $z_i = (x_i,y_i)$.
     From \eqref{eq:weighted_exch_cond}, 
    it is immediate that for a constant $t$ (conditional on $\cE_j$), $\PP(V_{n+j}\ge t \cond \cE_j)$ is equal to the weighted sum of indicator random variables:
    \begin{equation}\label{eq:weighted_exch_cond_cdf}
        \PP(V_{n+j}\ge t \cond \cE_j) = \sum_{k \in [n] \cup \{n+j\}}\frac{w(x_k)\indc{V(z_k) \ge t}}{w(x_{n+j})+\sum_{i=1}^n w(x_i) }.
    \end{equation}
    Since $$\EE\bigg[\indc{V_{n+j} \ge \hat T_j} \cond \cE_j, \{Z_{n+k}\}_{k \in [m]\setminus\{j\}}\bigg] = \PP\big({V_{n+j} \ge \hat T_j} \cond \cE_j, \{Z_{n+k}\}_{k \in [m]\setminus\{j\}}\big),$$
    $\hat T_j$ is constant conditioned on $\cE_j$ and $ \{Z_{n+k}\}_{k \in [m]\setminus\{j\}}$, and $V_{n+j}$ is independent of $\{Z_{n+k}\}_{k \in [m]\setminus\{j\}}$ by 
    assumption, we can use \eqref{eq:weighted_exch_cond_cdf} to directly conclude that 
    $$
    \EE\bigg[\indc{V_{n+j} \ge \hat T_j} \cond \cE_j, \{Z_{n+k}\}_{k \in [m]\setminus\{j\}}\bigg] 
    = \frac{w(X_{n+j})\ind\{V_{n+j}\ge \hat T_j\} 
    + \sum_{i=1}^n w(X_i) \ind\{V_i \ge \hat T_j\}}{w(X_{n+j}) +\sum_{i=1}^n w(X_i) }.
    $$
    The above in conjunction with \eqref{eq:weighted_conf_eval_proof} implies that $\EE[e_j] \le 1$.

\subsection{Proof of Proposition~\ref{prop:conf_eval_weighted}}
\label{appd:conf_eval_resample}
The proposition is immediate if we can conclude that $\tZ_1, \tZ_2, \ldots, \tZ_{n+m}$ 
and $Z_1, Z_2, \ldots, Z_{n+m}$ are jointly equal in distribution conditioned 
on $(\cE_j,  \{Z_{n+k}\}_{k \in [m]\setminus\{j\}})$. 
Since $\tZ_{n+k} = Z_{n + k}$ for $k \neq j$ by construction, 
we only need to consider the joint distribution corresponding to 
the indices $[n]\cup \{n+j\}$ (conditional on $\cE_j$, as the other test units were independently drawn).

Assuming the null $H_j$, %
we can write joint probability density function of $Z_1,\ldots, Z_n,Z_{n+j}$ 
in terms of the weight function $w$ and the density function $p$ of $P$~\citep{tibshirani2019conformal}:
\begin{equation}
    f(z_1, z_2, \dots, z_n, z_{n+j}) = w(x_{n+j})\prod_{i\in[n]\cup \{n+j\}} p(z_i).\footnote{As in~\citet{tibshirani2019conformal},
    we use the term ``density function'' loosely to refer to the Radon-Nikodym derivative with respect to an arbitrary base measure.}
\end{equation}
We can use this to calculate the joint conditional probabilities. Treat $\cE_j$ as fixed and denote its elements as $\{z_1, z_2, \dots, z_n, z_{n+1}\}$, without any particular order.  Let $\cS_{n+1}$ denote all permutations of $[n+1]$. Then for any permutation $\sigma \in \cS_{n+1}$, we have 
\begin{align}
    \begin{split}
        \PP(Z_{n+j} = z_{\sigma(n+1)}, Z_1 = z_{\sigma(1)},\ldots, Z_n = z_{\sigma(n)} \cond \cE_j) 
        &= \frac{w(x_{\sigma(n+1)}) \prod_{i=1}^{n+1}p(z_i)  }{\sum_{\sigma' \in \cS_{n+1}} 
            w(x_{\sigma'(n+1)}) \prod_{i=1}^{n+1}p(z_i) }\\
        &= \frac{w(x_{\sigma(n+1)})}{ \sum_{i=1}^{n+1} w(x_i) \cdot n! }.
    \end{split}
\end{align}
Meanwhile, the event $\{\tZ_{n+j} = z_{\sigma(n+1)}, \tZ_1 = z_{\sigma(1)},\ldots, \tZ_n = z_{\sigma(n)}\} \cond \cE_j$ occurs when we first resample $\tZ_{n+j}$ to be $z_{\sigma(n+1)}$ and subsequently assign $\{z_1, z_2, \dots, z_{n+1}\}\setminus\{z_{\sigma(n+1)}\}$  to $\tZ_1, \ldots, \tZ_n$ uniformly at random (without replacement). Thus,
\begin{equation}
    \PP(\tZ_{n+j} = z_{\sigma(n+1)}, \tZ_1 = z_{\sigma(1)},\ldots, \tZ_n = z_{\sigma(n)} \cond \cE_j)  
    = \frac{w(x_{\sigma(n+1)})}{ \sum_{i=1}^{n+1} w(x_i) \ } \cdot \frac{1}{n!}.
\end{equation}
We conclude that the conditional joint distributions match, as desired. The i.i.d. property of the resamples follows immediately by observing that the only randomness of the e-values, conditional on $S_j$, comes from the random (weighted) assignment to $(\tilde Z_{n+j}, \tZ_1, \dots, \tZ_n)$.

\subsection{Proof of Proposition~\ref{prop:filtering}}\label{appd:filtering_proof}
We can write the intermediate e-values
$
e^{\boost, \cS}_j  \coloneqq  e^\boost_j \indc{j \in \cS},
$ for $j\in [m]$.
Decomposing the FDR as in the proof of e-BH in~\eqref{eq:fdr_ineq}, we see:
\begin{align}
\begin{split}
    \fdr[\cS] &= \sum_{j\in \cH_0} \EE\bigg[\frac{\indc{j \in \cS}}{|\cS|\vee 1}\bigg] = \sum_{j\in \cH_0} \EE\bigg[\frac{\indc{j \in \cS}}{|\cS\cup\{j\}|}\bigg]\\
    &\stepa{\le} \sum_{j\in\cH_0}\EE\left[ \frac{\indc{j\in\cS} \cdot \indc{e_j^\boost \ge \frac{m}{\alpha |\hcR_j(\be)| } }}{|\cR(\be)\cup\{j\}|} \right]\\
    &\stepb{\le}  \sum_{j\in\cH_0}\EE\left[ \indc{j\in\cS} \cdot\frac{\frac{\alpha |\hcR_j(\be)| }   {m} e_j^\boost  }{|\cR(\be)\cup\{j\}|} \right]\\
    & \stepc{=} \sum_{j \in \cH_0}  \frac{\alpha}{m} \EE\bigg[ \indc{j\in\cS}  e_j^\boost \bigg]\\
    & \stepd{\le} \alpha.
\end{split}
\end{align}
Step (a) is due to $\cR(\be) \subseteq \cS$ (the denominator), and 
$\cS \subseteq \cR(\be^\boost)$: if $j \in \cS$, we also have 
$j \in \cR(\be^\boost)$ and therefore $e_j^\boost = {m}/{(\alpha |\hcR_j(\be)|)}$ by construction.
Step (b) follows from the deterministic inequality $\ind\{X\ge t\}\le X/t$ for $t >0$ 
in the e-BH proof. Step (c) is from the definition $\hcR_j(\be)  = \cR(\be)\cup\{j\}$. 
Step (d) follows since $e_j^\boost \indc{j \in \cS}$ is a valid e-value, as mentioned previously. Note that the entire inequality chain will also hold when $\be$ are compound e-values.

\section{Related literature and connections to specific works}
\label{appd:related_work}
\subsection{Related work}
There has been a growing literature on the topic of e-values,
including the interpretation and properties of 
e-values~\citep{shafer2011test,shafer2021testing,vovk2021values,vovk2023confidence,grunwald2023posterior},
the existence and construction of powerful 
e-values~\citep{wasserman2020universal,grunwald2020safe,zhang2023existence,larsson2024numeraire},
and the use of e-values in various statistical 
problems~\citep{howard2020time,howard2021confidence,waudby2024estimating,
vovk2024nonparametric,wang2022backtesting,waudby2022anytime}.
When it comes to multiple testing, as mentioned earlier,~\citet{wang2022false} propose the e-BH procedure for 
testing multiple hypotheses while controlling the FDR under arbitrary dependence.~\citet{ignatiadis2023evalues}
study the problem when both p-values and e-values are 
available.~\citet{xu2023online} consider using the e-values for multiple testing 
in the online setting, while~\citet{xu2022post} concern themselves with building post-selection 
confidence intervals by inverting the e-values. 

More broadly, there is a substantial line of works on multiple testing with 
FDR control. As mentioned in the main text, this problem was first studied by~\citet{benjamini1995controlling}, 
in which they also proposed the Benjamini-Hochberg (BH) procedure that operates on p-values. 
The BH procedure is only known to control the FDR with independent or 
positively correlated p-values. Otherwise, 
a severe correction is needed to ensure the FDR control~\citep{benjamini2001control}.
Subsequent works have investigated the asymptotic FDR 
control of BH~\citep{genovese2004stochastic,storey2004strong,ferreira2006benjamini,farcomeni2007some}.
The recent work of~\citet{chi2022multiple} discusses the FDR control of the BH 
procedure under negative dependence, providing better correction factors
than in the arbitrary dependence case.~\citet{sarkar2023controlling} develops 
a variant of the BH procedure that controls the FDR when testing multivariate normal 
means against two-sided alternatives.  
Our work draws inspiration 
from~\citet{fithian2020conditional}.
The authors introduce the dBH procedure that uses conditional calibration to modify the p-value threshold in the BH procedure, 
thereby achieving finite-sample FDR control for a larger class of p-value dependence structures; we instead use 
conditional calibration to boost the power of the e-BH procedure. 
We provide a detailed discussion on the connection between their method and ours in later subsections.

Beyond the line of work for multiple testing 
with p-values,~\citet{barber2015controlling,candes2018panning}
propose the knockoff procedure that controls the FDR 
by adding ``knockoff'' variables to the regression.
It has been shown that the knockoff-based methods also have 
e-value interpretations~\citep{ren2024derandomised}; based on this observation, 
we apply our framework to improving the power of the knockoff method 
in Section~\ref{sec:example-knockoffs}. The follow-up work by~\cite{luo2022improving}, which boosts the power of the knockoffs filter using conditional calibration, is also relevant and discussed in Appendix \ref{sec:cknockoff}.

\subsubsection{Attempts at addressing power concerns with randomization}

As mentioned in Section \ref{sec:background}, the original authors of e-BH address its power issues using the \textit{marginal boosting} scheme (detailed in Section \ref{sec:marginal_boosting}). Here, we will summarize an alternative approach to improving the power of e-BH through using external randomness.

Two recent papers~\citep{ramdas2023randomized,xu2023more}
discuss the use of external randomness for boosting the power of e-values, where 
the former focuses on testing a single hypothesis and the latter considers multiple testing.
For example, one could replace $e_j$ with $\tilde e_j \coloneq e_j/U_j$, for $U_j \sim \text{Unif}([0,1])$ that is 
independent of everything else, and then apply the e-BH procedure to the $\tilde e_j$'s,
still guaranteeing FDR control. Since $U_j<1$ with probability 1, such an approach indeed improves the power. But 
the external randomness can potentially hinder the reproducibility of the results 
(the procedure can be quite sensitive to the realization of the $U_j$'s and 
different runs of the algorithm can yield different selections), and could   
encourage ``hacking'' the data to obtain a desired result (practitioners can repeatedly 
sample the $U_j$'s until getting significant results).
Our proposed method, in contrast, is deterministic in principle 
(assuming sufficient computational resources) and is effectively  
stable across different runs of the algorithm in practice.

\subsection{dBH}
\label{app:dbh_summary}
\citet{fithian2020conditional} propose a method named dBH, 
which uses conditional calibration to adjust the p-value threshold in BH, so as to control the FDR under dependence. 
Supposing access to a sufficient statistic $S_j$ such 
that (for each $j \in [m]$) one can evaluate the conditional distribution of $\bp\coloneq (p_1,\dots,p_m) \given S_j$ under the $H_j$,
it finds the critical value $\hat \tau_j$ (via numerical integration)\footnote{This is a simplified version of the dBH procedure. In the original formulation, 
the threshold is written as a function of the tuning parameter $c$.}:
\@\label{eq:dbh}
\hat \tau_j = \sup\Bigg\{c \in (0,1): \EE_{H_j}\bigg[\frac{\ind\{p_j \le c\}}{\hat{r}_j(\bp)} \Biggiven S_j \bigg] 
\le \frac{\alpha}{m}\Bigg\},
\@
where 
$\hat{r}_j(\bp)$ is an estimate of  $|\{j\}\cup \{k \in [m]: p_k \le \hat{\tau}_k\}|$ (not to be confused with $\hcR_j(\be)$).
The adjusted p-value threshold then yields a selection set 
$\hat{\cR}^+ = \{j\in[m]:p_j \le \hat{\tau}_j\}$. Note that the size of $\hcR^+$ 
may differ from the estimates $\hat r_j(\bp)$, so the construction of $\hat \tau_j$ 
does not necessarily imply FDR control.  The dBH procedure takes additional 
pruning steps to ensure FDR control: 
\begin{itemize}
\item [(1)] if $|\hat{\cR}^+| \ge \hat{r}_j(\bp)$ for all $j\in[m]$, then 
stop the procedure and return $\cR^\dbh \coloneq \hat{\cR}^+$;
\item [(2)] if there exists $j\in[m]$ such that $|\hat{\cR}^+| < \hat{r}_j(\bp)$, 
then generate $U_1,\ldots,U_m \stackrel{\text{i.i.d.}}{\sim} \mathrm{Unif}([0,1])$ that are 
independent of everything else, and return
\begin{equation}
  \cR^{\dbh} := \big\{j \in \hat{\cR}^+:U_j \le \tfrac{r^*}{\hat{r}_j(\bp)}\big\},
  \text{ where } 
  r^* = \max \big\{r \in [m]: \big|\{j\in \hat{\cR}^+: U_j \le \tfrac{r}{\hat{r}_j(\bp)} \}\big|\ge r \big\}.
\end{equation}
\end{itemize}

\red{
As seen above, dBH makes use of conditional calibration to 
achieve FDR control.
However, as~\citet{fithian2020conditional} discuss, 
$\cR^\dbh$ can uniformly improve upon BH when  the conditional PRDS (CPRDS) condition holds. Furthermore, 
the authors also propose dBY, their conditionally calibrated BY method, and again show uniform improvement in power. 
Our framework aligns with the latter philosophy, starting 
with an FDR-controlling selection set and applying conditional calibration to increase the power.
}

\red{
In addition, }dBH involves 
additional pruning steps that introduce external randomness, while our method does not. 
\red{
  Finally, we show in the upcoming Appendix~\ref{app:dbh_equiv} that dBH has an e-BH interpretation, which corresponds to the nonrandom, no-pruning implementation.
}

\subsection{Connection between dBH and e-BH} 
\label{app:dbh_equiv}
To see why dBH has an e-value representation, we write 
\@\label{eq:dbh_eval}
e_j^\dbh = \frac{m\ind\{p_j \le \hat{\tau}_j\}}{\alpha \hat{R}_j(\bp)}.
\@ 
For simplicity, we assume that the conditional expectation in~\eqref{eq:dbh} 
is less than or equal to $\alpha/m$ at the critical value (otherwise we can replace the 
``$\le$'' with ``$<$'' as in our e-BH-CC procedure).
As a consequence, $\EE[e_j^\dbh] \le 1$ for any $j \in \cH_0$, which 
means that $e_j^\dbh$ is a valid e-value.
With this observation, we can apply the e-BH procedure to the e-values 
and obtain a deterministic selection set with FDR control.
One can also check that applying e-BH to the ``dBH e-values'' in~\eqref{eq:dbh_eval} is 
equivalent to dBH with all the $U_j$'s replaced by $1$; moreover, 
the dBH procedure is equivalent to applying the e-BH procedure to 
$(e^\dbh_1/U_1,\ldots,e_m^\dbh/U_m)$---this 
is the pe-BH procedure in~\citet{ignatiadis2023evalues} or equivalently 
the U-eBH procedure in~\citet{xu2023more}.
Such a connection has also been noticed under a different context in~\citet{jin2023weighted}.
We formalize the equivalence in the following proposition.
\begin{proposition}
\label{prop:dbh_equiv}
The e-BH procedure applied to 
$(e^\dbh_1/U_1,\ldots,e^\dbh_m/U_m)$ is equivalent to 
the dBH procedure, where $e^\dbh_j$'s are as defined in~\eqref{eq:dbh_eval}. 
In particular, if we replace all the $U_j$'s by $1$, then 
dBH is equivalent to e-BH applied to $(e_1^\dbh,\ldots,e_m^\dbh)$.
\end{proposition}

\begin{proof}
Let $\tilde{\be} = (e_1^\dbh / U_1,\ldots,e_m/U_m)$,
and $\cR^\ebh(\tilde{\be})$ the selection set returned by the e-BH procedure 
applied to $\tilde{\be}$. Let $\cR^{\dbh}(\bp)$ denote the selection set
of the dBH procedure with $U_j$'s.
We denote by $\tilde{e}_{(1)} \ge \tilde{e}_{(2)} \ge \ldots \ge \tilde{e}_{(m)}$ the 
ordered statistics of $\tilde{\be}$ in descending order.
Let $k^* = |\cR^\ebh(\be^\dbh)|$, and recall that for dBH 
\$ 
r^* = \max \big\{r \in [m]: \big|\{j\in \hat{\cR}^+: 
U_j \le r/\hat{R}_j(\bp) \} \big|\ge r \big\}
\$ 
We can check that 
\$ 
\tilde{e}_{(k)} \ge \frac{m}{\alpha k} \iff 
\bigg|\Big\{j \in[m]: \tilde{e}_j \ge \frac{m}{\alpha k}\Big\}\bigg| \ge k.
\$
As a result, by the definition of $k^*$ and $r^*$, 
\$ 
k^* = \max\Big\{k \in [m]: \tilde{e}_{(k)} \ge \frac{m}{\alpha k}\Big\}
& = \max\bigg\{k \in [m]: \Big|\Big\{j\in[m]: \tilde{e}_{j} \ge \frac{m}{\alpha k}\Big\}\Big|\ge k\bigg\}\\
& = \max\bigg\{k \in [m]: \bigg|\Big\{j\in\hcR^+: \frac{m}{\alpha \hat{R}_j(\bp)U_j}
 \ge \frac{m}{\alpha k}\Big\}\bigg|\ge k\bigg\}\\
& = \max\bigg\{k \in [m]: \bigg|\Big\{j\in\hcR^+: U_j \le 
 \frac{k}{\hat{R}_j(\bp)}\Big\}\bigg|\ge k\bigg\}\\
& = r^*,
\$
where the second step follows from the construction of $\tilde{e}_j$'s.
This leads to 
\$ 
\cR^\ebh(\tilde{\be}) = \Big\{j\in[m]: \tilde{e}_j \ge \frac{m}{\alpha k^*}\Big\}
& = \Big\{j\in[m]: \frac{m \ind\{j \in \hcR^+\}}{\alpha \hat{R}_j(\bp)U_j} 
\ge \frac{m}{\alpha k^*}\Big\}\\
& = \Big\{j\in\hcR^+: \frac{r^*}{\hat{R}_j(\bp)} 
\ge U_j\Big\}\\
& = \cR^{\dbh}(\bp).
\$
The proof is therefore concluded.
\end{proof}

\subsection{Connection between e-BH-CC and cKnockoff}
\label{sec:cknockoff}

The construction of the e-values can be exactly applied to \emph{fixed-X knockoffs}~\citep{barber2015controlling}, the precursor to model-X knockoffs. Although their method is nonrandom in theory, most practical implementations involve randomly sampling an orthonormal matrix. Thus, one can apply the e-value formulation of the fixed-X knockoffs in order to aggregate them for the purpose of derandomization---see Appendix A.1 in \cite{ren2024derandomised} for the exact details.

\cite{luo2022improving} employ conditional calibration in order to improve the power of the fixed-X knockoffs (and model-X) procedure in the presence of the threshold phenomenon. Their proposal, named \emph{cKnockoff}, uses an additional feature importance statistic $T_j$ to return the improved rejection set $\cR^\kn\cup\{j\colon T_j \ge \hat c_j \}$, where the thresholds $\hat c_j$ are data-adaptive and conditionally calibrated to achieve their FDR control guarantee. %

Specifically, for each feature $j\in [m]$, they choose a conditioning statistic $S_j$ and budget $B_j(S_j)$ and consider
\begin{equation}
  \label{eq:cknockoff}
  \EE\left[ \frac{\indc{j \in \cR^\kn} \vee \indc{T_j\ge c_j} }{|\cR^\kn \cup \{j\}|}   \Biggiven S_j \right] \le B_j(S_j).
\end{equation}
$\hat c_j$ is chosen as the smallest value of $c_j\in \RR$
such that \eqref{eq:cknockoff} is satisfied. The budget has the restriction that $\sum_{j\in\cH_0} \EE [B_j] \le \alpha$ (unconditionally), so that the calibrated thresholds achieve FDR control.

The parallels between e-BH-CC and cKnockoff are evident by the form of \eqref{eq:cknockoff}. 
One can view this formulation as equivalent to a variant of e-BH-CC, where the decision to boost the e-values also involves some auxiliary statistic $W_j$. 
Using a modification of the $\phi$-function
\begin{equation}
  \label{eq:auxiliary_phi}
   \phi_j(c; S_j) = \EE \left[\frac{m}{\alpha} \cdot \frac{\indc{e_j \ge \frac{m}{\alpha |\hat \cR_j(\bm e)|} \text{ or } W_j \ge c}}
   {|\hat \cR_j(\bm e)|}  - e_j \bigggiven S_j \right],
\end{equation} 
we can then boost $e_j$ to $\frac{m}{\alpha |\hcR_j(\be)|}$ when the critical value of \eqref{eq:auxiliary_phi} is no greater than $W_j$. Observing that 
$$\indc{e_j \ge \frac{m}{\alpha |\hat \cR_j(\bm e)|} }
= \indc{e_j \ge \frac{m}{\alpha (| \cR_j(\bm e)| \vee 1)} }
= \indc{ j \in \cR^\kn }
$$
by the self-consistency of e-BH~\citep{wang2022false} and e-BH interpretation of knockoffs~\citep{ren2024derandomised}, 
we note that the two methods are equivalent when $T_j = W_j$ and the budget is set to the authors' default of $(\alpha/m) \EE[e_j\given S_j]$.
The distinct benefit of our formulation, which was designed for e-values in general, is that we can directly apply it when $e_j$ are now derandomized e-values from either fixed-X or model-X knockoffs.

\subsection{Connection between WCS and e-BH} 
\label{appd:ebh_wcs}
\newcommand{\hcS}{\hat \cS}
\newcommand{\hS}{\hat S}
\newcommand{\hR}{\hat R}
\newcommand{\citest}{\cI_{\text{test}}}
\newcommand{\cic}{\cI_{\text{calib}}}
\newcommand{\quant}{{\text{quant}}}
For the outlier detection problem, 
the WCS procedure~\citep{jin2023weighted} with deterministic pruning 
computes a weighted conformal p-value for 
each $j \in [m]$ as in~\eqref{eq:weighted_pval} and returns the selection set 
\$
\cR^\wcs = \Big\{j \in [m]: p_j \le \frac{\alpha}{m}\cdot|\hcR_j|,~
|\hcR_j| \le r^* \Big\}.
\$
Above, $\hcR_j$ is a ``proxy'' selection set, obtained via
applying the BH procedure
to the collection of p-values 
$(p_1^{(j)},\ldots, p_{j-1}^{(j)},0,p_{j+1}^{(j)},\ldots, p_m^{(j)})$, 
where for each $\ell \in [m]$,
\$ 
p_\ell^{(j)} = \frac{\sum_{i \in [n]} w(X_i) \cdot \ind\{V_i \ge V_{n+\ell}\} 
+ w(X_{n+j}) \cdot \ind\{V_{n+j}\geq V_{n+\ell}\}}
{\sum_{i \in [n]} w(X_i) + w(X_{n+j})},
\$
and the threshold $r^*$ is defined via
\$ 
r^* = \max\bigg\{r \in [m]: \sum_{j\in [m]} 
\ind\big\{p_j \le \alpha|\hat{R}_j|/m, 
~|\hat{R}_j| \le r \big\} \ge r  \bigg\}.
\$
As pointed out by~\citet{jin2023weighted}, $\cR^\wcs$ can equivalently 
be obtained by applying the BH procedure to the e-values: 
\$
e_j^\wcs =  \frac{\ind\{p_j \le \alpha |\hat R_j|/m\}}{\alpha |\hat R_j|/m}, 
\quad \forall j \in [m].
\$
We can compare the above e-values to the conformal e-values we constructed directly in~\eqref{eq:w_conf_eval}, which gives rise to the following proposition.
\begin{proposition}
    \label{prop:more_powerful_than_wcs}
    For each $j \in [m]$, $e_j \ge e_j^\wcs$ deterministically. As a direct result,
    \[
        \cR^\ebh(e_1, \dots, e_m) \supseteq \cR^\ebh (e_1^\wcs, \dots, e_m^\wcs) = \cR^\wcs.
    \]
\end{proposition}
\begin{proof}[Proof of Proposition \ref{prop:more_powerful_than_wcs}]
    The second statement is a direct consequence of the first, since running e-BH on deterministically bigger e-values leads to a rejection set which is deterministically at least as big. 

    All we must do is show that the weighted conformal e-value $e_j$ constructed in~\eqref{eq:w_conf_eval}
    satisfies $e_j \ge e_j^\wcs$ deterministically. The following lemma is key 
    to establish this connection.

\begin{lemma} \label{lemma:weighted_equivalence}
For any $j\in [m]$, the following holds. 
\begin{enumerate}
\item[(1)] On the event that $\{p_j \le \alpha |\hat R_j|/m\}$, 
$\hR_j = \{\ell\in[m]\colon V_{n+\ell} \ge T_j\}$.
\item[(2)] $p_j \le \alpha |\hat R_j|/m$ if and only if $V_{n+j} \ge T_j$; 
\end{enumerate}
\end{lemma}
\noindent Lemma~\ref{lemma:weighted_equivalence} then implies that 
\$ 
e_j & = \Big(w(X_{n+j}) + \sum^n_{i=1}w(X_i)\Big) \cdot 
 \frac{\ind\{V_{n+j} \ge T_j\}}{\sum_{i\in n}w(X_{i})\ind\{V_i \ge T_j\} + w(X_{n+j})}\\
 & \stackrel{\text{(a)}}{\ge} \frac{m}{\alpha} \cdot \frac{\ind\{V_{n+j} \ge T_j\}}{|\{\ell \in [m]\colon V_{n+\ell} \ge T_j\}|}
 \stackrel{\text{(b)}}{=} \frac{m}{\alpha} \cdot \frac{\ind\{p_j\le \alpha |\hat R_j|/m\}}{\hat R_j}
 = e_j^\wcs.
\$
Above, step (a) follows from the definition of $T_j$ and step (b) uses Lemma~\ref{lemma:weighted_equivalence}.
We also note that step (a) is often quite tight due to the choice of $T_j$.
\end{proof}

Now, the only missing piece is the proof of Lemma~\ref{lemma:weighted_equivalence}, which 
we provide below.
\begin{proof}[Proof of Lemma~\ref{lemma:weighted_equivalence}]
\
\begin{enumerate}
\item [(1)]Fix $j\in[m]$, 
and let $k^* = |\hat R_j|$ and $\ell^* = |\{\ell \in[m]: V_{n+\ell} \ge T_j\}|$.

Suppose $p_j \le \alpha|\hat R_j|/m$.
By the property of the BH procedure, 
\$
\hat R_j = \{j\} \cup \{\ell \neq j: p^{(j)}_\ell \le \alpha k^* / m\}
= \{\ell \in [m]: p^{(j)}_\ell \le \alpha k^* / m\},
\$
where the last step is because $p^{(j)}_j = p_j \le \alpha k^* / m$.
Let $p_{(1)}^{(j)} \le \cdots \le p_{(m)}^{(j)}$ denote the ordered p-values 
in an ascending order.
Since the rank of $p_\ell^{(j)}$ is determined by 
$V_{n+\ell}$, we have (with a slight abuse of notation) that 
$V_{n+(1)} \ge V_{n+(2)} \ge \cdots \ge V_{n+(m)}$.

Next, note that
\@\label{eq:wcs_ebh_1}
\frac{\sum_{i \in [n]} w(X_i) \ind\{V_i \ge V_{n+(k^*)}\}+w(X_{n+j})}
{\sum_{\ell \in [m]}\ind\{V_{n+\ell} \ge V_{n+(k^*)}\}}
\frac{m}{w(X_{n+j}) + \sum_{i\in[n]}w(X_i)}
= p^{(j)}_{(k^*)} \frac{m}{k^*} \le \alpha.
\@
where the first inequality follows from the definition of $p^{(j)}_{k^*}$ and 
that $\sum_{\ell \in [m]} \ind\{V_{n+\ell} \ge V_{n+(k^*)}\} = k^*$;
the last step is due to the property of the BH procedure.
Eqn.~\eqref{eq:wcs_ebh_1} implies that $V_{n+(k^*)} \ge T_j$.
Therefore for any $\ell \in \hat R_j$, $V_{n+\ell} \ge V_{n+(k^*)} \ge T_j$,
and $\hat R_j \subset \{\ell \in[m]: V_{n+\ell} \ge T_j\}$.

Conversely, we can see that 
\@\label{eq:wcs_ebh_2}
p^{(j)}_{(\ell^*)} & = \frac{\sum_{i\in[n]}w(X_i)\ind\{V_i \ge V_{n+(\ell^*)}\}+w(X_{n+j})}
{\sum_{i\in [n]}w(X_i) + w(X_{n+j})} \notag\\ 
& \stackrel{\text{(a)}}{=}  \frac{\sum_{i\in[n]}w(X_i)\ind\{V_i \ge V_{n+(\ell^*)}\}+w(X_{n+j})}
{\sum_{\ell \in m} \ind\{V_{n+\ell} \ge V_{n+(\ell^*)}\}}
\frac{m}{\sum_{i\in [n]}w(X_i) + w(X_{n+j})}\frac{\ell^*}{m}\notag\\
&\stackrel{\text{(b)}}{\le} \frac{\alpha \ell^*}{m}.
\@
Above, step (a) is because $\sum_{\ell \in [m]} \ind\{V_{n+\ell} \ge V_{n+(\ell^*)}\} = \ell^*$, 
and step (b) is because $V_{n+(\ell^*)} \ge T_j$. As a result of~\eqref{eq:wcs_ebh_2} and the 
property of the BH procedure, we have $\ell^* \le k^*$.
For any $\ell\in[m]$ such that $V_{n+\ell} \ge T_{n+j}$, there is also 
$p_\ell \le p_{(\ell^*)} \le \alpha k^* /m$, implying that $\ell \in \hat R_j$.
Collectively, we have $\hat R_j = \{\ell \in[m]: V_{n+\ell} \ge T_j\}$.

\item [(2)]
When $p_j \le \alpha |\hat R_j| / m$, we have by (1) that
$ 
p_j \le \frac{\alpha \ell^*}{m}.
$
Suppose otherwise $V_{n+j}<T_j$. We can check that
\$
&~\frac{\sum_{i \in [n]} w(X_i) \ind\{V_i \ge V_{n+j}\} + w(X_{n+j})}{\sum_{\ell \in [m]}\ind\{V_{n+\ell}\ge V_{n+j}\}}
\frac{m}{\sum_{i\in [n]} w(X_i) + w(X_{n+j})} \\
=&~ p_j \cdot \frac{m}{\sum_{\ell \in [m]}\ind\{V_{n+\ell}\ge V_{n+j}\}}\\
\le &~ p_j \cdot \frac{m}{\sum_{\ell \in [m]}\ind\{V_{n+\ell}\ge T_j\}}\\
= &~ \frac{p_j m}{\ell^*} \le \alpha.
\$
The above implies that $V_{n+j} \ge T_j$, which is a contradiction. Therefore, we conclude 
that $V_{n+j} \ge T_j$.

Conversely, when $V_{n+j} \ge T_j$, there is 
\$
\frac{\sum_{i\in[n]} w(X_i) \ind\{V_i \ge V_{n+j}\} + w(X_{n+j})}{\sum_{\ell \in [m]}\ind\{V_{n+\ell}\ge V_{n+j}\}}
\frac{m}{\sum_{i \in [n]} w(X_i) + w(X_{n+j})}\le \alpha .
\$
Rearranging the above inequality, we have
\$
 \le \frac{\alpha}{m} \sum_{\ell \in [m]} \ind\{V_{n+\ell}\ge V_{n+j}\} \le \frac{\alpha \ell^*}{m}.
\$ 
Recall that Eqn.~\eqref{eq:wcs_ebh_2} proves $p^{(j)}_{(\ell^*)} \le \alpha \ell^*/m$.
By the property of the BH procedure, 
\$
\ell^* \le |\cR^\bh(p_1^{(j)},\ldots,p_m^{(j)})| \le |\hat R_j|.
\$
As a result, we have $p_j \le \alpha |\hat R_j|/m$, completing the proof.

\end{enumerate}

\end{proof}

\section{Additional implementation details}
\subsection{AVCS-approximated boosted e-values}
\label{appd:avcs_approx}

When $K$ is fixed ahead of time, the user may experience a disappointing event: the Monte-Carlo estimate $\overline E_{j,K}$ is negative, yet the confidence interval contains zero, making the resulting boosted e-value zero as well. Given this, it would be quite tempting to continue resampling from $\be \cond S_j$, hopefully until the confidence interval is finally contained within ${\RR^{\le 0}}$.
Unfortunately, this adaptive mechanism will break the error control of the confidence interval, which has downstream implications for FDR control.

To address this issue in general,  
we use the sequential \emph{anytime-valid confidence sequence} (AVCS) 
(see, e.g.,~\citet{ramdas2023game} for a review). A $(1-\alpha)$-coverage 
AVCS for some parameter $\theta$ is a sequence of confidence intervals $\{(L_k, U_k)\}_{k\ge1}$ such that
$  \PP \big(\forall k \in \NN, \theta \in  [L_k , U_k]\big) \ge 1-\alpha$.
This is in contrast to the common confidence interval, where the ``for all $k$'' qualifier is outside of the $\PP(\cdot)$. 
Usually, the AVCS is used as an online analogue of a confidence interval: at each time step $k$, 
a new sample $X_k$ is added to the existing sequence of samples to construct the latest iterate 
of the confidence interval $[L_k, U_k]$ in a way such that the miscoverage probability of this 
process at any point is at most $\alpha$. 
Using the AVCS, we can ignore $K$ and replace the confidence interval in Proposition~\ref{prop:mc_error}
with its anytime-valid variant (with $\alpha_\avcs= \alpha_0 |\cR^\ebh(\be)| /m$), stopping the resampling process optionally. 

For each $j \in [m]$, accrue samples $E_{j,1}, E_{j,2}, \dots \cond S_j$ by resampling $\be^{(1)}, \be^{(2)}, \dots\cond S_j$.
Define $\alpha_{\avcs} = \alpha_0 |\cR^\ebh(\be)| /m$ for some $\alpha_0 \in (0,1)$ 
reflecting the Monte-Carlo error budget. 
We then construct a $(1-\alpha_\avcs)$-coverage anytime-valid confidence sequence 
$\{C_{j,k} \coloneq C_{j,k}(E_{j,1},\dots, E_{j,k})\}_{k\ge 1}$, 
and define the AVCS-approximated boosted e-values to be
  \begin{equation}
    \label{eq:avcs_boosted_eval}
    e^{\boost, \avcs}_j = m\cdot \frac{ \indc{\exists k \in \NN \colon U_{j,k} \le 0} }{\alpha |\hat \cR_j(\bm e)|},
  \end{equation}
where $U_{j,k}$ is the upper endpoint of $C_{j,k}$. Algorithm~\ref{alg:ebh_cc_acvs} summarizes 
the procedure of implementing e-BH-CC with AVCS-approximated boosted e-values, and 
Proposition~\ref{prop:avcs_mc_error} characterizes the FDR control of the resulting rejection set.

\begin{proposition}\label{prop:avcs_mc_error}
  Suppose $\be = (e_1,\ldots,e_m)$ are compound e-values. 
  Running e-BH on the collection of AVCS-approximated 
  boosted e-values  $ \be^{\boost, \avcs }=  \{e^{\boost, \avcs}_j \}_{j \in [m]}$
  with the target FDR level $\alpha$ and the target Monte-Carlo error budget $\alpha_0$, 
  we have
  $
    \fdr[\cR^\ebh(\be^{\boost, \avcs})] \le \alpha+\alpha_0.
  $
\end{proposition}

\begin{proof}
For any $j\in \cH_0$, we have 
\$ 
\EE[e_j^{\boost, \avcs}] 
& = \EE\Bigg[m\cdot \frac{\ind\{\exists k, U_{j,k}\le 0\}}{\alpha |\hcR_j(\be)|} \Bigg]\\
& = \EE\Bigg[m\cdot \frac{\ind\{\exists k, U_{j,k}\le 0, \phi_j(\tilde c_j ; S_j) > 0\}}{\alpha |\hcR_j(\be)|} \Bigg]
+ \EE\Bigg[m\cdot \frac{\ind\{\exists k, U_{j,k}\le 0, \phi_j(\tilde c_j ; S_j) \le 0\}}{\alpha |\hcR_j(\be)|} \Bigg]\\
& \le\EE\Bigg[m\cdot \frac{\ind\{\exists k, U_{j,k} < \phi_j(\tilde c_j ; S_j) \}}{\alpha |\hcR_j(\be)|} \Bigg]
+ \EE\Bigg[m\cdot \frac{\ind\{\phi_j(\tilde c_j ; S_j) \le 0\}}{\alpha |\hcR_j(\be)|} \Bigg].
\$
Above, the first term is bounded by $\alpha_0/\alpha$ by the construction of 
$\{C_{j,k}\}_{k\ge 1}$; the second term is bounded by $\EE[e_j]$ following 
exactly the same steps in the proof of Proposition~\ref{prop:mc_error}.
Now that $\EE[e_j^{\boost,\avcs}] \le \EE[e_j] + \alpha_0/\alpha$, invoking 
the proof of the e-BH procedure completes the proof.
\end{proof}

\paragraph{Asymptotic anytime-valid confidence sequences.}
The definition of an asymptotic anytime-valid confidence sequence 
below is adapted from~\citet{waudby2021time}.
\begin{definition}[Asymptotic anytime-valid confidence sequences (Asymp-AVCS)]
We say that $(\hat \theta_k - L_k, \hat \theta_k + U_k)_{k \ge 1}$ centered around the 
estimators $\{\hat \theta_k\}_{k\ge 1}$ with $L_k,U_k >0$ for any $k\ge 1$
forms a $(1-\alpha)$-asymptotic anytime-valid confidence sequence for a parameter 
$\theta$ if there exists a non-asymptotic $(1-\alpha)$-anytime-valid confidence sequence 
$(\hat \theta_k - L_k^*, \hat \theta_k + U_k^* )_{k\ge 1}$ such that 
\$ 
L_k^*/L_k \stackrel{\text{a.s.}}{\rightarrow} 1, \quad
U_k^*/U_k \stackrel{\text{a.s.}}{\rightarrow} 1.
\$
\end{definition}
In practice, we can replace the $(1-\alpha_\avcs)$-AVCS with an Asymp-AVCS 
when $k$ is sufficiently large. The time-uniform coverage of the Asymp-AVCS 
has been established in~\citet{waudby2021time} under certain conditions, and 
we refer the readers to their work for more details.

\subsection{Filtering out low-potential hypotheses}
\label{appd:filtering_out}

Since $\alpha_\avcs \propto  1 /m$, the online Monte-Carlo strategy is unsuitable for large-scale multiple testing applications, as the target coverage for each AVCS (per hypothesis) would be close to 1. Most exact AVCS constructions are already quite conservative, so this setting exacerbates those issues and would require copious amounts of resamples to terminate the AVCS.

To alleviate said issues, we can restrict our focus onto a subset of hypotheses $M$ with high rejection potential. For example, we could avoid boosting for the set $\{j\colon e_j = 0\}$, since no value of $\hat c_j$ will produce a nonzero $e_j^\boost$ from a zero-valued $e_j$. 
By filtering with $M$, we can replace $m$ with $|M|\vee 1$ in the denominator of $\alpha_\avcs$ (or $\alpha_\ci$) to make the confidence mechanism less conservative---appreciably so in sparse signal-to-noise regimes.

The key to preserving FDR control is to \textit{mask} the boosted e-values $e^{\boost, M}_j  \coloneqq  e^\boost_j \indc{j \in M}.$ By design, these are valid when $e_j$ is valid, yet for $j\notin M$ we need not evaluate $e_j^\boost$ (and when $j \in M$, the confidence mechanism is less conservative). This also allows us to forego an entire iteration of computation for each identified low-potential hypothesis, which saves time and resources. We then run e-BH on $(e^{\boost, M}_1, \dots, e^{\boost, M}_m)$.

By restricting our focus onto a subset $M$ of $\cH$ which have high potential of rejection (alternatively, filtering out hypotheses from $\cH$ with low potential of significance to get $M$), we can replace $m$ with $|M| \vee 1$ in the construction of $\alpha_\ci$ and $\alpha_\avcs$. That is, per hypothesis $H_j$, we use
\begin{equation}
    \alpha_{\ci} = \alpha_0 \frac{  |\cR^\ebh(\be)| } {|M| \vee 1}; \quad \alpha_{\avcs} = \alpha_0 \frac{ |\cR^\ebh(\be)| } {|M| \vee 1}
\end{equation}
which will lead to a lower target coverage whenever $|M| < m$. By design, we expect this to happen.

A rudimentary filter is $M = \{j \colon e_j \neq  0\}$. Sections \ref{sec:example-knockoffs} and \ref{sec:example-conformal} contain examples of problem settings where this may occur with positive probability. More generally, any e-value that is designed as an ``all-or-nothing bet''~\citep{shafer2021testing} takes value $1/p$ with probability $p > 0$ and $0$ otherwise.

Another possible strategy to filter out low-potential discoveries is to use preliminary test statistics, such as p-values or linear model coefficients. Take, for example, the $m$-dimensional one-sided $z$-testing problem setting
(covered in-depth in Section \ref{sec:example-parametric}). Each hypothesis $H_j$  has a corresponding p-value $p_j$ and e-value $e_j$. To control the FDR at $\alpha$, the BH procedure rejects $p_j$ if it lies below ${\alpha \hat k}/{m}$ for some data-driven $\hat k \in [m]$; therefore, a prerequisite to rejecting $p_j$ is that it is at most $\alpha$. The filtered hypotheses set here would be $M\coloneqq \{j\colon  p_j \le\alpha \}$. To see specific examples of filters used in simulations, we direct the reader to Appendix \ref{appd:deferred_details_sims_parametric} and Appendix \ref{appd:deferred_details_sims_knockoffs}. 

Given $M$, we construct \emph{masked} boosted e-values:
\begin{equation}\label{eq:masked_boosted_eval}
  e^{\boost, M}_j  \coloneqq  e^\boost_j \indc{j \in M}.
\end{equation}
Since $\EE[e^{\boost, M}_j  ] = \EE[e^\boost_j \indc{j \in M}] \le \EE[e^\boost_j]$, the masked e-values are also valid. When $\be$ are compound e-values, then the masked boosted e-values will also be compound e-values. We choose to mask the unchosen e-values instead of deleting them from the procedure entirely to avoid post-selection issues, noting that masking is sufficient for (a) avoid unnecessary computation for low-potential hypotheses and (b) improving efficiency (decreasing conservativeness) for the confidence mechanisms used when boosting the high-potential e-values. We state the evident implication that (b) will in turn improve the power of the Monte-Carlo versions of e-BH-CC. 

For any mask $M \supseteq \cR(\be)$, using \eqref{eq:masked_boosted_eval}
gives us a collection of zeroed-out e-values, 
over which we can run e-BH and control the FDR. At this point, the reader may express concern that filtering may cause e-BH-CC to become overly conservative. We argue this is not the case, in the sense that $M \cap \cR(\be^\boost)    =  \cR(\be^{\boost, M})$ exactly. 
In other words, $j \in M$ will be in the filtered e-BH-CC rejection set $\cR(\be^{\boost, M})$ as long as it would have been rejected by e-BH-CC with no filtering. 
This is a straightforward observation stemming from the fact that the e-BH rejection threshold over $\be^{\boost, M}$, denoted $\hat t_\ebh(\be^{\boost, M})$, has the following characterization:
\begin{equation}
  \hat t_\ebh(\be^{\boost, M})\begin{cases}
    =  \frac{m}{\alpha |\cR(\be)|}  & \textnormal{when }\cR(\be)=\cR(\be^{\boost,M}),\\ 
    \le \frac{m}{\alpha (|\cR(\be)|+1)} &\textnormal{when }\cR(\be)\subsetneq\cR(\be^{\boost,M}),
  \end{cases}
\end{equation}
where we use the fact that $\cR(\be) \subseteq \cR(\be^{\boost,M})$.
The conclusion is trivial when $\cR(\be) = \cR(\be^{\boost,M})$,  
so we suppose instead that $\cR(\be) \subsetneq \cR(\be^{\boost,M})$.
Then for any $j\in \cR(\be)\cap M$, the boosted e-value $e_j^\boost$ 
takes magnitude at least $\frac{m}{\alpha (|\cR(\be)|+1)}$,  implying that $j \in \cR(\be^{\boost,M})$.
We can thus conclude that $\cR(\be^\boost) \cap M  \subseteq  \cR(\be^{\boost, M})$. 
Since the other direction of inclusion follows from the definition, we have $\cR(\be^\boost) \cap M =  \cR(\be^{\boost, M})$.

We delve deeper into the claim that the confidence mechanisms become less conservative, improving the power of the Monte-Carlo versions of e-BH-CC. First, we note that $M$ is generally designed to contain $\cR^\ebh(\be)$ in order to preserve the uniform power improvement. In addition, when implementing e-BH-CC in practice  with Monte-Carlo estimation, we use a confidence interval or AVCS and follow Proposition \ref{prop:mc_error} or \ref{prop:avcs_mc_error}, respectively. 
In both cases, the confidence level ($\alpha_\ci$ or $\alpha_\avcs$) has a dependence of $1/m$, where $m$ is the number of hypotheses (i.e., the number of confidence intervals or sequences that can err). However, when we employ a filter $M$, the number of boosted e-values we may incorrectly evaluate 
becomes $|M|\le m$. Therefore, we can replace $m$ with $|M|$ in $\alpha_\ci$ and $\alpha_\avcs$ while preserving the FDR control with additive Monte-Carlo error in Propositions \ref{prop:mc_error} and \ref{prop:avcs_mc_error}. The effect on the confidence mechanism is that it becomes less conservative, which leads to higher power.

We close this section by formalizing the effect of filtering on the FDR control guarantee of the e-BH-CC procedure,
whose proof can be found in Appendix~\ref{appd:filtering_proof}. 
\begin{proposition}
  \label{prop:filtering}
  Let $\cR(\be)$ be the original e-BH rejection set and $\cR(\be^\boost)$ be the boosted e-BH rejection set, both at level $\alpha \in (0,1)$. Then for all sets $\cS$ such that $\cR(\be)\subseteq \cS \subseteq \cR(\be^\boost)$, $\fdr[\cS]\le \alpha$.
\end{proposition}

\section{Multiple rounds of conditional calibration}
\label{appd:multiple_rounds}
In Section~\ref{sec:boosting} 
we describe an extension of e-BH-CC where
we apply CC multiple times, one round after another. 
By doing so, we heuristically move the boosted e-values to be 
as close to BH-sharp as possible. Part of the 
discussion was that using multiple rounds to sharpen the boosted
e-values more had the consequence of no longer guaranteeing
power improvement over the initial e-BH rejection set.

In this section, we rehash the idea of multiple CC rounds, 
with a more clever design for $\phi_j(c; S_j)$ and the resulting
boosted e-values to  reestablish uniform power improvement
\emph{across rounds}. We again begin by looking at how to design
the second round of CC; after we successfully guarantee
uniform power improvement, we discuss the general $k$-round case.

We then include numerical studies, including 
the two-round version used to boost the e-values for
the one-sided Gaussian testing problem. 
We employ important computational tricks 
to enable the multiple 
round implementation, so we also cover them in detail. 

\subsection{Two rounds of e-BH-CC}
\label{appd:two_round_retained}

Let
\[
    \cR^{(0)}(\be) \coloneqq \cR^\ebh(\be)
\]
be the original e-BH rejection set at level $\alpha$. Write
\[
    \hcR_j^{(0)}(\be) \coloneqq \cR^{(0)}(\be)\cup\{j\},
    \qquad
    \tau_j^{(0)}(\be) \coloneqq \frac{m}{\alpha |\hcR_j^{(0)}(\be)|}.
\]
We allow an arbitrary data-driven mask. Let $\mathcal D$ denote the data object used to construct the e-values, 
and let $M=M(\mathcal D)$ be a random subset of $[m]$ which is measurable with respect to $\sigma(\mathcal D)$. 
The implementation requirement is that CC can conditionally resample enough of $\mathcal D$ to recompute the e-values, the mask, and the sufficient statistics needed for another round of CC. In other words, for each null $j$, a conditional resample $\widetilde{\mathcal D}^{(j)}\cond S_j$ induces $(\bm\te,\tilde M)\cond S_j$ with the joint law of $(\be,M)\cond S_j$ under $H_j$, and also induces resampled sufficient statistics $\tilde S_h=S_h(\widetilde{\mathcal D}^{(j)})$ for $h\in[m]$. Retained hypotheses are kept regardless of the mask; the mask only determines which non-retained hypotheses are eligible for additional boosting. As a concrete example, in the Gaussian simulations one can choose
\[
    M=\{j\in[m]\colon p_j\le 3\alpha\},
\]
where $p_j$ is the one-sided Gaussian p-value corresponding to 
$Z_j$, as referenced in the simulation details. Resampling $M$
can be done directly through resampling the p-values $p_j$.

\paragraph{First round.}
For each $j\in[m]$, define
\begin{equation}
  \label{eq:two_round_phi1}
  \phi_j^{(1)}(c;S_j)
  \coloneqq
  \EE\left[
    \tau_j^{(0)}(\bm \te)
    \indc{c\te_j \ge \tau_j^{(0)}(\bm \te)}
    - \te_j
  \Bigggiven S_j\right],
\end{equation}
where $(\bm \te,\tilde M)\cond S_j$ follows the joint conditional distribution of $(\be,M)\cond S_j$ under $H_j$. Rather than explicitly solving for the first-round critical value, we use the sign equivalence in \eqref{eq:evaluate_fast_coef}: for a candidate not already in $\cR^{(0)}(\be)$, set
\[
    c_j^{\mathrm{eval}}(\be)\coloneqq \frac{\tau_j^{(0)}(\be)}{e_j},
\]
and boost it if $\phi_j^{(1)}(c_j^{\mathrm{eval}}(\be);S_j)\le0$. With a mask, define
\begin{equation}
  \label{eq:first_round_retained_eval}
  e_j^{(1)}
  \coloneqq
  \tau_j^{(0)}(\be)
  \indc{
    j\in \cR^{(0)}(\be)
    \textnormal{ or }
    \big(j\in M,~\phi_j^{(1)}(c_j^{\mathrm{eval}}(\be);S_j)\le0\big)
  } .
\end{equation}
Equivalently, if
\[
    I_h^{(1)}(\be,M)
    \coloneqq
    \indc{
      h\in \cR^{(0)}(\be)
      \textnormal{ or }
      \big(h\in M,~\phi_h^{(1)}(c_h^{\mathrm{eval}}(\be);S_h)\le0\big)
    },
\]
then $e_h^{(1)}=\tau_h^{(0)}(\be)I_h^{(1)}(\be,M)$ for each $h\in[m]$.
\[
    \cR^{(1)}(\be,M)\coloneqq \cR^\ebh(\be^{(1)})
\]
be the first-round rejection set. As in Theorem~\ref{thm:power-improvement}, $\cR^{(1)}(\be,M)\supseteq \cR^{(0)}(\be)$.

\paragraph{Second round.}
The second round recalibrates using the first-round rejection set as the retained set. For $c\ge0$, define
\[
    \hcR_j^{(1)}(\be,M)\coloneqq \cR^{(1)}(\be,M)\cup\{j\},
    \qquad
    \tau_j^{(1)}(\be,M)\coloneqq \frac{m}{\alpha |\hcR_j^{(1)}(\be,M)|}.
\]
The conceptual improvement is to 
add a ``retained'' event indicator 
which forces every first-round rejection to remain nonzero in round two.
\begin{equation}
  \label{eq:two_round_h_func}
  H_j^{(2)}(c,\be,M)
  \coloneqq
  \tau_j^{(1)}(\be,M)
  \indc{
    j\in\cR^{(1)}(\be,M)
    \textnormal{ or }
    \big(j\in M,~ c e_j \ge \tau_j^{(1)}(\be,M)\big)
  }.
\end{equation}
The term $j\in\cR^{(1)}(\be,M)$ is the retained part of the construction: any first-round rejection is kept in the second-round boosted e-values without further recalibration. The second-round calibration function is
\begin{equation}
  \label{eq:two_round_phi2}
  \phi_j^{(2)}(c;S_j)
  \coloneqq
  \EE\left[
    H_j^{(2)}(c,\bm \te,\tilde M)-\te_j
  \Bigggiven S_j\right],
\end{equation}
where, inside $H_j^{(2)}(c,\bm\te,\tilde M)$, both the resampled mask $\tilde M$ and the set $\cR^{(1)}(\bm\te,\tilde M)$ are recomputed by applying the same data-driven rules to the resampled data. More explicitly, for each outer resample $(\bm\te,\tilde M,\tilde S_1,\ldots,\tilde S_m)$ and each $h\in[m]$, form the resampled first-round evaluation point
\[
    \tilde c_h^{\mathrm{eval}}
    \coloneqq
    \frac{\tau_h^{(0)}(\bm\te)}{\te_h}.
\]
This value is recomputed from the resampled e-values; it is not held fixed at the observed value $c_h^{\mathrm{eval}}(\be)$. The resampled first-round boost indicator is
\[
    \widetilde I_h^{(1)}
    \coloneqq
    \indc{
      h\in \cR^{(0)}(\bm\te)
      \textnormal{ or }
      \big(h\in \tilde M,~\phi_h^{(1)}(\tilde c_h^{\mathrm{eval}};\tilde S_h)\le0\big)
    }.
\]
Thus the resampled first-round e-values and rejection set used inside the second-round expectation are
\[
    \tilde e_h^{(1)}=\tau_h^{(0)}(\bm\te)\widetilde I_h^{(1)},
    \qquad
    \cR^{(1)}(\bm\te,\tilde M)
    =
    \cR^\ebh(\tilde{\be}^{(1)}).
\]
Let
\begin{equation}
  \label{eq:two_round_c2}
  \hat c_j^{(2)}\coloneqq \sup\{c\colon \phi_j^{(2)}(c;S_j)\le0\}.
\end{equation}
For ease of presentation, assume $\phi_j^{(2)}(\hat c_j^{(2)};S_j)\le0$; if this fails due to a jump discontinuity, replace the ``$\ge$'' in the second event in \eqref{eq:two_round_h_func} by ``$>$'', exactly as in \eqref{eq:boosted_eval}. The two-round boosted e-value is
\begin{equation}
  \label{eq:two_round_eval}
  e_j^{(2)}\coloneqq H_j^{(2)}(\hat c_j^{(2)},\be,M).
\end{equation}
The workflow used in the simulations evaluates the second-round sign with a one-sided confidence sequence for $\phi_j^{(2)}(\tau_j^{(1)}(\be,M)/e_j;S_j)$, accepting a new second-round boost only when the upper endpoint is below zero. This is the same conservative Monte-Carlo replacement used elsewhere in the paper; the oracle object in \eqref{eq:two_round_eval} is the population target of that implementation.

\begin{proposition}
  \label{prop:two_round_validity}
  Suppose $\be=(e_1,\ldots,e_m)$ are valid e-values, and let $M$ be any data-driven mask satisfying the joint resampling condition above. Then the two-round boosted e-values $\be^{(2)}=(e_1^{(2)},\ldots,e_m^{(2)})$ defined in \eqref{eq:two_round_eval} are valid e-values. The same statement holds for compound e-values. Consequently, applying e-BH to $\be^{(2)}$ controls the FDR at level $\alpha$.
\end{proposition}

\begin{proof}
Fix $j\in\cH_0$. Ignoring the data-driven mask, \eqref{eq:first_round_retained_eval} is the one-round e-BH-CC e-value written using the sign equivalence \eqref{eq:evaluate_fast_coef}. The mask only replaces some positive boosted values by zero, so
\[
    \EE[e_j^{(1)}\given S_j]\le \EE[e_j\given S_j].
\]
Moreover, the retained part of the second round is pointwise no larger than the first-round e-value:
\[
    \tau_j^{(1)}(\be,M)\indc{j\in\cR^{(1)}(\be,M)}
    \le e_j^{(1)}.
\]
Indeed, if $j\in\cR^{(1)}(\be,M)$, then $e_j^{(1)}>0$ and hence $e_j^{(1)}=\tau_j^{(0)}(\be)$. Since $\cR^{(1)}(\be,M)\supseteq\cR^{(0)}(\be)$ and $j\in\cR^{(1)}(\be,M)$, we have $\hcR_j^{(0)}(\be)\subseteq\cR^{(1)}(\be,M)$, so $\tau_j^{(0)}(\be)\ge\tau_j^{(1)}(\be,M)$.
Thus $\phi_j^{(2)}(0;S_j)\le0$, so the set in \eqref{eq:two_round_c2} is nonempty.

By the definition of $\hat c_j^{(2)}$ and the construction of $e_j^{(2)}$,
\[
  \EE[e_j^{(2)}\given S_j]
  =
  \EE\left[H_j^{(2)}(\hat c_j^{(2)},\be,M)\Biggiven S_j\right]
  =
  \phi_j^{(2)}(\hat c_j^{(2)};S_j)+\EE[e_j\given S_j]
  \le \EE[e_j\given S_j].
\]
  Taking expectations yields $\EE[e_j^{(2)}]\le\EE[e_j]\le1$ for every null $j$. Hence $\be^{(2)}$ are valid e-values. If the original e-values are compound e-values, then summing the same inequality over $j\in\cH_0$ gives $\sum_{j\in\cH_0}\EE[e_j^{(2)}]\le\sum_{j\in\cH_0}\EE[e_j]\le m$. The FDR statement follows from the usual validity of e-BH under arbitrary dependence.
\end{proof}

\begin{proposition}
  \label{prop:recursive_refinement}
  Let $\cR^{(2)}(\be,M)\coloneqq \cR^\ebh(\be^{(2)})$. For every realization of the original e-values and mask,
  \[
      \cR^{(2)}(\be,M)\supseteq \cR^{(1)}(\be,M).
  \]
\end{proposition}

\begin{proof}
If $\cR^{(1)}(\be,M)=\varnothing$, the claim is immediate. Otherwise, for every $j\in\cR^{(1)}(\be,M)$,
\[
    e_j^{(2)}=\frac{m}{\alpha|\cR^{(1)}(\be,M)|}.
\]
There are $|\cR^{(1)}(\be,M)|$ such e-values, all exactly at the e-BH threshold corresponding to $|\cR^{(1)}(\be,M)|$ rejections. Therefore e-BH applied to $\be^{(2)}$ rejects all indices in $\cR^{(1)}(\be,M)$, possibly along with additional second-round boosts.
\end{proof}

\paragraph{Numerical results.}
To illustrate the construction above, we ran numerical simulation study in the one-sided $z$-testing setting of Section~\ref{sec:example-parametric}. We set $m=100$, $|\cH_1|=10$, $\alpha=0.05$, $\Sigma_{jk}=\rho^{|j-k|}$ with $\rho=0.5$; 
we also explicitly set the signal amplitude $\Delta=3.5$ and $a_j\equiv 1$ (the hyperparameter for the likelihood ratio e-value construction). Across $100$ independent seeds, we compared BH on the one-sided Gaussian p-values, e-BH on the raw Gaussian e-values, one-round e-BH-CC, and two-round e-BH-CC. (We also use $M=\{j\colon p_j\le 3\alpha\}$.)
The first round uses the exact homotopy evaluation of \eqref{eq:two_round_phi1}, while the second round uses the ``retain round one'' rule \eqref{eq:two_round_h_func} with a hybrid confidence-sequence upper endpoint and Monte-Carlo budget $\alpha_0=0.1\alpha=0.005$. Figure~\ref{fig:masked_retained_bar} reports the resulting mean power and FDP with two-standard-error bars. 

In this regime, raw e-BH is effectively powerless, while one-round e-BH-CC recovers a substantial part of the gap to BH
(these are results already shown in Section \ref{sec:sim-parametric}). The  second round of e-BH-CC yields a further power gain over one-round e-BH-CC. Over these $100$ seeds, mean power increases from $0.590$ to $0.725$ from the first to the second round, while mean FDP increases from $0.005$ to $0.031$ and remains below both $\alpha=0.05$ and the reference level $\alpha+\alpha_0=0.055$ in this experiment. BH remains the most powerful method among the four (although it itself is not a valid FDR-controlling procedure), but the second round narrows the remaining gap substantially while continuing to empirically exhibit the FDR control guarantee.

\begin{figure}[tb]
  \centering
  \includegraphics[width=0.82\textwidth]{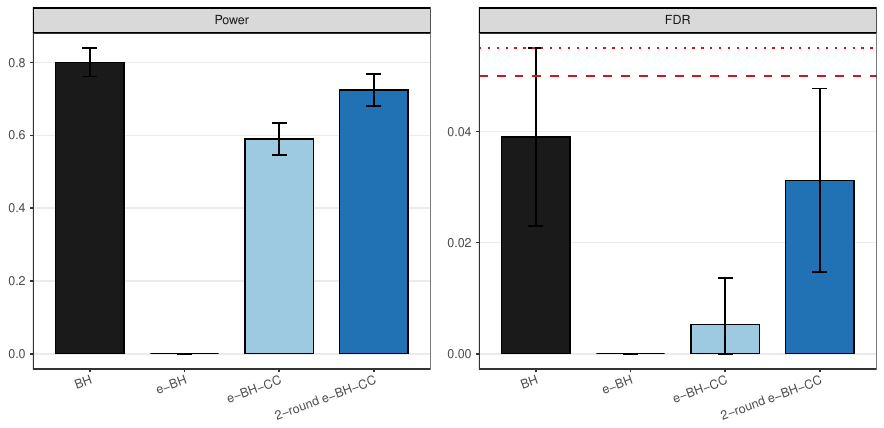}
  \caption{Mean power and FDP in the reduced Gaussian experiment described in the text. Bars are averages over $100$ seeds, and the error bars are $\pm 2$ standard errors. The dashed and dotted horizontal lines in the FDP panel mark $\alpha=0.05$ and $\alpha+\alpha_0=0.055$, respectively.}
  \label{fig:masked_retained_bar}
\end{figure}

\subsection{General multiple-round construction}
\label{appd:retained_multi_round}

The two-round construction above extends directly to any fixed finite number of rounds. 
Suppose the first-round e-values $\be^{(1)}$ and rejection set $\cR^{(1)}(\be,M)$ have been constructed as in \eqref{eq:first_round_retained_eval}. 
For notational uniformity, write $\cR^{(0)}(\be,M)\equiv\cR^{(0)}(\be)$ and $\tau_j^{(0)}(\be,M)\equiv\tau_j^{(0)}(\be)$.
For $k\ge2$, assume that $\be^{(k-1)}$ and $\cR^{(k-1)}(\be,M)\coloneqq \cR^\ebh(\be^{(k-1)})$ have already been defined. 
Let
\[
    \hcR_j^{(k-1)}(\be,M)\coloneqq \cR^{(k-1)}(\be,M)\cup\{j\},
    \qquad
    \tau_j^{(k-1)}(\be,M)\coloneqq \frac{m}{\alpha|\hcR_j^{(k-1)}(\be,M)|}.
\]
For $c\ge0$, define the round-$k$ update with the retained indicator included:
\begin{equation}
  \label{eq:retained_multi_h}
  H_j^{(k)}(c,\be,M)
  \coloneqq
  \tau_j^{(k-1)}(\be,M)
  \indc{
    j\in\cR^{(k-1)}(\be,M)
    \textnormal{ or }
    \big(j\in M,~ c e_j\ge \tau_j^{(k-1)}(\be,M)\big)
  }.
\end{equation}
The associated calibration function is
\begin{equation}
  \label{eq:retained_multi_phi}
  \phi_j^{(k)}(c;S_j)
  \coloneqq
  \EE\left[
    H_j^{(k)}(c,\bm\te,\tilde M)-\te_j
  \Bigggiven S_j\right],
\end{equation}
where the previous $k-1$ rounds, the mask, and the retained rejection set
$\cR^{(k-1)}(\bm\te,\tilde M)$ are recomputed from the conditional resample exactly as in the two-round case. 
Let
\[
    \hat c_j^{(k)}\coloneqq \sup\{c\colon \phi_j^{(k)}(c;S_j)\le0\},
    \qquad
    e_j^{(k)}\coloneqq H_j^{(k)}(\hat c_j^{(k)},\be,M),
\]
again using the same left-continuity convention as above when the supremum is attained at a jump.

\begin{proposition}
  \label{prop:retained_multi_validity_power}
  Suppose $\be=(e_1,\ldots,e_m)$ are valid e-values, and let $M$ be any data-driven mask satisfying the joint resampling condition above. For any fixed finite $K\ge1$, the $K$-round boosted e-values $\be^{(K)}$ defined above are valid e-values. The same statement holds for compound e-values. Moreover, writing $\cR^{(0)}(\be,M)\equiv\cR^{(0)}(\be)=\cR^\ebh(\be)$, for every realization of the original e-values and mask and every $k=1,\ldots,K$,
  \[
      \cR^{(k)}(\be,M)\supseteq \cR^{(k-1)}(\be,M),
  \]
\end{proposition}

\begin{proof}
As in the one-round result in our main text,
we prove validity via the stronger conditional statement
$\EE[e_j^{(k)}\given S_j]\le \EE[e_j\given S_j]$ for every null $j$. 
The case $k=1$ is the one-round e-BH-CC construction with masking, for which the conditional validity statement follows from Theorem~\ref{thm:eval_valid}; the mask only removes possible boosts among non-retained hypotheses. 
Containment at $k=1$ follows directly from the retained definition: if $j\in\cR^{(0)}(\be)$, then $e_j^{(1)}=m/(\alpha|\cR^{(0)}(\be)|)$, so e-BH applied to $\be^{(1)}$ rejects all of $\cR^{(0)}(\be)$.

Assume the result holds through round $k-1$. For a null $j$, the retained part of the next update is pointwise no larger than the previous-round e-value:
\[
  \tau_j^{(k-1)}(\be,M)\indc{j\in\cR^{(k-1)}(\be,M)}
  \le e_j^{(k-1)}.
\]
Indeed, on the event $j\in\cR^{(k-1)}(\be,M)$, the induction hypothesis gives
\[
    \cR^{(k-2)}(\be,M)\subseteq\cR^{(k-1)}(\be,M).
\]
The previous-round construction therefore implies $e_j^{(k-1)}=\tau_j^{(k-2)}(\be,M)\ge\tau_j^{(k-1)}(\be,M)$. 
The same pointwise inequality holds for every conditional resample after recomputing the previous $k-1$ rounds from $(\bm\te,\tilde M)$. Hence, using the induction hypothesis inside the conditional law given $S_j$,
\[
  \phi_j^{(k)}(0;S_j)
  \le
  \EE\left[\tilde e_j^{(k-1)}-\te_j\Biggiven S_j\right]
  \le 0.
\]
Thus the set defining $\hat c_j^{(k)}$ is nonempty. By construction,
\[
  \EE[e_j^{(k)}\given S_j]
  =
  \phi_j^{(k)}(\hat c_j^{(k)};S_j)+\EE[e_j\given S_j]
  \le \EE[e_j\given S_j].
\]
This immediately implies validity of the method.

Finally, if $j\in\cR^{(k-1)}(\be,M)$, then
\[
    e_j^{(k)}
    =
    \frac{m}{\alpha|\cR^{(k-1)}(\be,M)|}.
\]
Thus all indices in $\cR^{(k-1)}(\be,M)$ are exactly at the e-BH threshold corresponding to $|\cR^{(k-1)}(\be,M)|$ rejections, so e-BH applied to $\be^{(k)}$ rejects all of them. This proves $\cR^{(k)}(\be,M)\supseteq\cR^{(k-1)}(\be,M)$ and completes the induction.
\end{proof}

\subsection{Exact first-round computation for the Gaussian example}
\label{appd:first_round_homotopy}

We now describe the exact computation of \eqref{eq:two_round_phi1} used in the Gaussian one-sided testing workflow. The mask does not change this computation; it only decides which candidate indices require the computation. Recall the setting of Section~\ref{sec:example-parametric}: $Z\sim\cN_m(\mu,\Sigma)$ with $\Sigma_{jj}=1$, $H_j:\mu_j=0$, and
\[
    e_h=\exp(a_h Z_h-a_h^2/2), \qquad h\in[m].
\]
For each $j$, the sufficient statistic is $S_j=Z_{-j}-\Sigma_{-j,j}Z_j$. Under $H_j$, a conditional resample can be parameterized by a single scalar
$Y\sim\cN(0,1)$:
\[
    \tilde Z_j(Y)=Y,\qquad
    \tilde Z_{-j}(Y)=S_j+\Sigma_{-j,j}Y.
\]
Consequently, for each $h\in[m]$,
\begin{equation}
  \label{eq:gaussian_loge_affine}
  \log \tilde e_h(Y)
  =
  \gamma_{hj}+\lambda_{hj}Y,
  \qquad
  \gamma_{hj}=a_h\big(Z_h-\Sigma_{hj}Z_j\big)-\frac{a_h^2}{2},
  \qquad
  \lambda_{hj}=a_h\Sigma_{hj},
\end{equation}
where the formula also holds for $h=j$ because $Z_j-\Sigma_{jj}Z_j=0$ and $\lambda_{jj}=a_j$.

For a fixed $c$, the first-round quantity is
\[
  \phi_j^{(1)}(c;S_j)
  =
  \EE\left[
    \frac{m}{\alpha D_j^{(0)}(Y)}
    \indc{\log c+\gamma_{jj}+\lambda_{jj}Y
      \ge \log\frac{m}{\alpha D_j^{(0)}(Y)}}
    -\tilde e_j(Y)
  \Bigggiven S_j\right],
\]
where
\[
    D_j^{(0)}(Y)\coloneqq |\cR^\ebh(\tilde\be(Y))\cup\{j\}|.
\]
The homotopy algorithm exploits that $D_j^{(0)}(Y)$ is piecewise constant as a function of $Y$. Following the knot characterization of \citet{fithian2020conditional}, changes can occur only when some resampled e-value crosses an e-BH threshold:
\begin{equation}
  \label{eq:first_round_knots}
  \gamma_{hj}+\lambda_{hj}Y=\log\frac{m}{\alpha r},
  \qquad h,r\in[m].
\end{equation}
Thus, for $\lambda_{hj}\neq0$, the potential knots are
\[
    y_{h,r}^{(j)}
    =
    \frac{\log\{m/(\alpha r)\}-\gamma_{hj}}{\lambda_{hj}},
    \qquad h,r\in[m],
\]
with the $\lambda_{hj}=0$ cases contributing no knot unless an entire affine path lies exactly on a threshold, a measure-zero tie handled by grouping numerically equal knots.

After sorting the knots, take the induced intervals $I_k=(u_k,u_{k+1})$, allowing $u_0=-\infty$ and $u_{K+1}=\infty$. On each interval, evaluate $D_j^{(0)}(Y)$ at any interior midpoint and call the resulting constant $d_k$. The integral then has the closed form
\begin{equation}
  \label{eq:first_round_homotopy_integral}
  \phi_j^{(1)}(c;S_j)
  =
  \sum_{k=0}^{K}
  \frac{m}{\alpha d_k}
  \PP\left\{
    Y\in I_k,\;
    \log c+\gamma_{jj}+\lambda_{jj}Y
      \ge \log\frac{m}{\alpha d_k}
  \right\}
  -\EE[\tilde e_j(Y)\given S_j].
\end{equation}
Each probability in \eqref{eq:first_round_homotopy_integral} is a difference of Gaussian distribution functions over a half-interval. The final expectation is equal to one because $Y\sim\cN(0,1)$ under $H_j$ and $\tilde e_j(Y)=\exp(a_jY-a_j^2/2)$.
This computes the sign of $\phi_j^{(1)}(c_j^{\mathrm{eval}}(\be);S_j)$ exactly up to floating-point arithmetic, without Monte-Carlo error in the first round.

\section{Deferred details of the main simulation studies}
\label{appd:deferred_details_simulations}
\subsection{Relating to Section \ref{sec:sim-parametric}}
\label{appd:deferred_details_sims_parametric}

\paragraph{Computational details.} We use a hybrid AVCS formulation to control the Monte-Carlo error. This is a computational trick which addresses the conservativeness of the exact finite-sample AVCS by using the asymptotic AVCS after we have accumulated many samples. The first 3000 e-value resamples, batched in sizes of 100 and scaled to the unit interval, are used to form the hedged capital confidence sequence (HCCS) as proposed in~\cite{waudby2024estimating}. After this point, the next samples, also batched in sizes of 100, are used to construct the asymptotic AVCS described in Theorem 2.2, \cite{waudby2021time}. If after 5000 total resamples the resulting AVCS still contains zero, we stop early and fail to boost the e-value. 

In addition, we employ a filter to cut down the number of e-values which we attempt to boost. For all experiments, we filter out all $j$ such that $e_j = 0$, since they see no benefit from boosting. Furthermore, we use the filter 
$$
M = \{  j\colon p_j \le 3 \alpha \},
$$
where $p_j$ are the one-sided p-values formed from the $z$-statistics. %

We choose $\alpha_0 = 0.1 \cdot \alpha = 0.005$ and use $\alpha_\avcs$ as described in Algorithm \ref{alg:ebh_cc_acvs}, replacing $m$ with the size of the filter. This is in line with our target FDR (accounting for Monte-Carlo error) of $0.05$. Even though the asymptotic AVCS only bounds above the limit supremum of the FDR at $\alpha + \alpha_0 = 0.055$, we see that the FDR is empirically controlled at 0.05.

\begin{remark}
  The HCCS~\citep{waudby2024estimating} is an AVCS which covers the mean of {bounded random variables}. Thus, one of the criteria to use the HCCS as the AVCS in our implementation outlined in Appendix \ref{appd:avcs_approx} is that the resamples of 
  $$
  \frac m \alpha \cdot \frac{ \indc{ \tilde c_j \cdot \te_j \ge \frac{m}{\alpha |\hat\cR_j(\tilde \be )|} } }{|\hat \cR_j(\tilde \be )|} - \tilde e_j \bigggiven S_j
  $$
  are bounded random variables. Unfortunately, this is not always true, and although the settings of Sections \ref{sec:example-knockoffs} and \ref{sec:example-conformal} satisfy this requirement due to $\tilde e_j$ being dichotomized random variables, the LRT e-values based on continuous distributions will not.

  However, in our very specific case with $z$-statistics, it is true that $\tilde e_j \indep S_j$. This is a direct consequence
  of $\cov(\tilde Z_j, S_j) =0$, which is shown in Appendix \ref{appd:proof_zstat_eval_resample}, the proof of 
  Proposition \ref{prop:zstat}. Thus, the conditional expectation which is the ``mean'' we are covering with the HCCS is
  \begin{align}
    \begin{split}
      \phi_j(\tilde c_j; S_j) &= \EE\left[\frac m \alpha \cdot \frac{ \indc{ \tilde c_j\te_j \ge \frac{m}{\alpha |\hat\cR_j(\tilde \be )|} } }{|\hat \cR_j(\tilde \be )|} - \tilde e_j \bigggiven S_j \right]\\
      &= \EE\left[\frac m \alpha \cdot \frac{ \indc{ \tilde c_j\te_j \ge \frac{m}{\alpha |\hat\cR_j(\tilde \be )|} } }{|\hat \cR_j(\tilde \be )|} \bigggiven S_j \right] -1
    \end{split}
  \end{align} 
  because $\EE[\tilde e_j\given S_j] = \EE[\tilde e_j] = 1$ under $H_j$. Thus, we can use the e-value resamples to construct resamples which follow the conditional distribution 
  $$
  \frac m \alpha \cdot \frac{ \indc{ \tilde c_j\te_j \ge \frac{m}{\alpha |\hat\cR_j(\tilde \be )|} } }{|\hat \cR_j(\tilde \be )|} \bigggiven S_j, 
  $$
  which are obviously bounded, and use them to estimate $\phi_j(\tilde c_j;S_j) +1$, which is equivalent to our objective.

  The same trick can be used for $t$-statistics (Appendix \ref{sec:tstat}).
\end{remark}

\subsection{Relating to Section \ref{sec:sim-knockoffs}}
\label{appd:deferred_details_sims_knockoffs}

\paragraph{Knockoffs details.} When sampling the knockoff design matrix, we use Equation (3.2) in \citet{candes2018panning}, where the $S$-matrix is constructed via the minimum variance-based reconstructibility criterion (see Definition 3.1 in \citet{spector2022powerful}). We compute the lasso-coefficient difference (Equation (3.7), \citet{candes2018panning}) as the feature importance measure. To avoid a computational slowdown from having to repeat cross-validation steps for each call of LASSO regression, we use a separate independently-sampled dataset $(n=500)$ as a ``hold-out dataset'' to pre-compute the $\lambda$ regularization parameter. This pre-computation is repeated for every replication per experiment.

We also use a variant of the threshold $T$ in \eqref{eq:mxkn_threshold} which~\citet{ren2024derandomised} argue lead to a uniform improvement in power for derandomized knockoffs. The alternative threshold
\begin{equation}\label{eq:mxkn_threshold_alt}
  T^{(k)}_{\mathrm{alt}}\coloneqq  \inf\Bigg\{ t > 0 \colon \frac{1 + \sum_{j\in[m]} \indc{W_j^{(k)} \le -t}}
  {\sum_{\ell \in [m]} \indc{W_\ell^{(k)} \ge t}} \le \alpha_\kn \;\textnormal{or}\; 
  \sum_{j\in[m]} \indc{W_j^{(k)} \ge t} < {1}  /{\alpha_\kn}\Bigg\}
\end{equation}
is an ``early-stopping'' version---when the original threshold is infinite (and all the e-values will be zero), using the alternative stopping time can only increase the e-values. We use $T^{(k)}_{\mathrm{alt}}$ to give the base procedure some more power, while still observing that there is a large power gap to recover.

In our implementation, we frequently make use of the \texttt{knockpy} Python package, which is a Python implementation of the knockoffs procedure~\citep{spector2022powerful}. We build on their functionality to implement derandomized knockoffs.

\paragraph{Computational details.} We again use a hybrid AVCS formulation to control the Monte-Carlo error (see Appendix \ref{appd:deferred_details_sims_parametric} for details). The first 1200 samples (in batches of 100) are used to construct the exact AVCS, while the next 800 samples are used to construct the asymptotic AVCS. After 2000 total samples, if the resulting AVCS still contains zero we stop early and fail to boost.

In addition, we will filter to cut down the number of e-values which we attempt to boost. For all experiments, we filter out all $j$ such that $e_j = 0$, since they see no benefit from boosting. For $d = 30$, we will additionally filter using 
$$
M = \{j\colon \hat\beta_j^{\mathrm{LASSO}}\neq 0\} \cup \{j\colon p_j^{\mathrm{regression}}\le \alpha\},
$$ 
where $\hat\beta^{\mathrm{LASSO}}$ is the estimated coefficient vector obtained from LASSO regression on $(\bm X, \bm Y)$ and $p_j^{\mathrm{regression}}$ is the p-value corresponding to the $j$th covariate, derived from running OLS with intercept. The filter $M$ attempts to capture covariates that are significant to the outcome while avoiding being too strict. Lastly, we choose $\alpha_0 = 0.1 \cdot \alpha = 0.01$ and use $\alpha_\avcs$ as described in Algorithm \ref{alg:ebh_cc_acvs}, replacing $m$ with the size of the filter. 

\subsection{Relating to Section \ref{sec:sim-conformal}}
\label{appd:deferred_details_sims_conformal}

\paragraph{Computational details.} 
We again use a hybrid AVCS formulation to control the Monte-Carlo error (see Appendix \ref{appd:deferred_details_sims_parametric} for details). The first 1,500 samples (in batches of 100) are used to construct the exact AVCS, while the next 1,000 samples are used to construct the asymptotic AVCS. After 2,500 total samples, if the resulting AVCS still contains zero we stop early and fail to boost the particular e-value.

\section{Further details and instantiation examples of e-BH-CC}
\label{appd:further_examples}
\subsection{Gaussian mean testing with unknown covariance}
\label{sec:tstat}

When the covariance matrix $\Sigma$ is instead partially or totally unknown, 
the $t$-statistic replaces its $z$-statistic counterpart. 
Assume that $Z \sim \cN_m(\mu, \Sigma)$, where $\Sigma = \sigma^2\Psi$ with $\Psi \succ 0$ 
known but $\sigma \in \RR^+$ unknown. Furthermore, assume we have access to an auxiliary 
independent vector $W \sim \cN_{n-m}(0, \sigma^2 I_{n-m})$ for estimating $\sigma^2$.
 Then, under the same null hypotheses $H_1, \dots, H_m$, the $t$-statistic used for testing $H_j$ is 
\begin{equation}
  T_j \coloneqq \frac{Z_j}{\sqrt{\hat\sigma^2 \Psi_{jj}}} \sim t_{n-m},
\end{equation}
where $(n-m)\hat\sigma^2 = \|W\|^2 \sim \sigma^2 \chi_{n-m}^2$. Here, $t_{n-m}$ and $\chi^2_{n-m}$ refers to the student-$t$ and $\chi^2$ distributions with $n-m$ degress of freedom. 

For e-values, we can again construct a likelihood ratio test statistic per $j$. Under the point alternative $H_j^{(a_j)}\colon \mu_j =a_j$, $T_j \sim t_{n-m}(a_j)$, the noncentral $t$-distribution with noncentrality parameter $a_j$ and $n-m$ degrees of freedom; this distribution is a generalization of the regular $t$-distribution under the relation $t_{n-m}(0) \eqd t_{n-m}$. Denoting the density of $t_{n-m}(a)$ as $f_{t_{n-m}, a}$, we then define for each $H_j$ the LRT e-value 
\begin{equation}\label{eq:tstat-evalue}
  e_j = \frac{f_{t_{n-m}, a_j}(T_j)}{f_{t_{n-m}, 0}(T_j)},
\end{equation}
where again the choice of $a_j$ can be chosen \emph{a priori} or through sample splitting.

For the conditioning statistic $S_j$, \cite{fithian2020conditional} construct the pair $(U_j, V_j)$ with the components defined as follows:
\begin{align}\label{eq:tstat-suffstat-pair}
  \begin{split}
    U_j &= Z_{-j} - \Psi_{-j,j} \Psi_{jj}^{-1} Z_j,\\
    V_j &= \|W\|^2 + \frac{Z_j^2}{\Psi_{jj}}.
  \end{split}
\end{align}
By choosing this statistic, we again have $S_j \indep T_j$; in fact, $U_j, V_j$, and $T_j$ are all mutually independent. In addition, it can be shown that $T_{-j}$ is a deterministic function of $(T_j, U_j, V_j)$:
\begin{equation}\label{eq:t-from-suffstat}
  T_k = U_{jk}\sqrt{\frac{n-m + T_j^2}{\Psi_{kk} V_j}} + \frac{\Psi_{kj}}{\Psi_{jj}}T_j\quad \text{for }k\neq j.
\end{equation}
As a direct consequence of \eqref{eq:t-from-suffstat}, we can resample from $T\cond S_j$ by sampling $Y \sim t_{n-m}$ independent from $S_j$, setting $\tT_j^{(j)} = Y$, and setting $\tT_{-j}^{(j)}$ as the function of $\tT_j^{(j)}, S_j$ above. The LRT e-values can be calculated per-component. We formalize the correctness of this resampling scheme in the following proposition. %
\begin{proposition}\label{prop:tstat}
  For each $j\in[m]$, choose sufficient statistic $S_j = (U_j, V_j)$ as defined in \eqref{eq:tstat-suffstat-pair}. Using $U_j$, $V_j$, and $\tT_j^{(j)} = Y  \sim t_{n-m}$ (independent from $S_j$), construct $\tT^{(j)}_{-j}$ as per \eqref{eq:t-from-suffstat}. Then $\tT^{(j)} \cond S_j \sim T\cond S_j$ and  
  $$
  \left(\te_1^{(j)}, \ldots, \te_m^{(j)} \right)\cond S_j \sim (e_1, \ldots, e_m) \cond S_j
  $$
  under $H_j$, where $e_k$ and $\tilde e_k^{(j)}$ is constructed as the component-wise LRT by using \eqref{eq:tstat-evalue} on $T$ and $\tT^{(j)}$, respectively.
  
  Furthermore, independent samples of $Y  \sim t_{n-m}$ give rise to independent resamples from $(e_1, \ldots, e_m) \cond S_j$.
\end{proposition}
\begin{proof}[Proof of Proposition \ref{prop:tstat}]
  We can proceed similarly to the proof of Proposition \ref{prop:zstat}. As \cite{fithian2020conditional} show, $S_j = (U_j, V_j)$ is independent of $T_j$. Therefore, the independently sampled $\tT_j^{(j)} = Y \sim t_{n-m}$ is distributed as $T_j \cond S_j$. Noting that $\tT_{-j}$ depends on $\tT_j^{(j)}, S_j$ in the same exact deterministic way as $T_{-j}$ depends on $T_j, S_j$, we conclude that  
  $$  (T_{-j}, T_j)  \cond S_j\sim (\tT_{-j}^{(j)}, \tT_j^{(j)}) \cond S_j.$$
  The proposition is then immediate.
\end{proof}

\subsection{Details on derandomized knockoffs}
\label{appd:derand_knockoff}

The procedure to derandomize the knockoff filter is as follows. After choosing a hyperparameter $\alpha_\kn \in (0,1)$, sample $d$ knockoff samples $(\bm \tX^{(1)}, \dots, \bm \tX^{(d)})$. For each knockoff $\bm \tX^{(k)}$, we run the knockoff filter with the feature importances $W^{(k)} = \cW([\bm X, \bm \tX^{(k)}], \bm Y)$ at level $\alpha_{kn}$, which results in the threshold 
\begin{equation}
  T^{(k)}\coloneqq  \inf\left\{ t > 0 \colon \frac{1 + \sum_{j\in[m]} \indc{W_j^{(k)} \le -t}}{\sum_{\ell \in [m]} \indc{W_\ell^{(k)} \ge t}} \le \alpha_\kn \right\}.
\end{equation}
We then use \eqref{eq:kn_evalues} to construct the compound e-values for the $k$th run of the knockoff filter:
\begin{equation}\label{eq:kth_knockoffs_eval}
  e_j^{(k)} =m \cdot \frac{\indc{W_j^{(k)} \ge T^{(k)}}}{1 + \sum_{\ell\in [m]} \indc{W_\ell^{(k)} \le -T^{(k)}} },\quad\forall j\in[m].
\end{equation}
For each feature $j$, we can construct a derandomized e-value by averaging $e_j^{(k)}$ over all $d$ runs of the knockoff filter. The averaged e-value, formally defined as
\begin{equation}\label{eq:averaged_knockoffs_eval}
  \bar e_j \coloneqq \frac 1d\sum_{k\in[d]} e_j^{(k)},
\end{equation}
exhibits more stability as the extraneous randomness from each run of knockoffs is averaged out. Since taking the average preserves (compound) e-value validity, we can run e-BH on $(\bar e_1, \dots, \bar e_m)$ to get a rejection set while controlling the FDR at level $\alpha$. 

The hyperparameter $\alpha_\kn$ does not necessarily have to equal $\alpha$ (the desired FDR control), and~\cite{ren2024derandomised} suggest having the former depend on the latter through the choice $\alpha_\kn =\alpha /2$ in order to achieve good power. 
Regardless of the choice of $\alpha_{\kn}$, the resulting $\bar e_1, \dots, \bar e_m$ will be valid e-values. 
Figures \ref{fig:mxkn-tp25}, \ref{fig:mxkn-tp10}, and \ref{fig:mxkn-tp9} show the effect of conditional calibration on power for multiple choices of $\alpha_\kn$. The FDP diagnostics corresponding to Figures \ref{fig:mxkn-tp25} and \ref{fig:mxkn-tp9} are reported in Figures \ref{fig:mxkn-tp25-fdr} and \ref{fig:mxkn-tp9-fdr}, respectively; Figure \ref{fig:mxkn-tp10} includes the FDP for the at-threshold experiment. In addition, there are other slight improvements which can affect the power of derandomized knockoffs---we discuss our implementation for the simulations in more detail in Section \ref{sec:sim-knockoffs}.

\paragraph{Resampling.} The following proposition establishes that the sampling procedure described in Section~\ref{sec:cc-kn}
produces i.i.d.~resamples from $(\bar e_1, \dots, \bar e_m)\cond S_j$.

\begin{proposition}
  \label{prop:knockoffs_eval_resample}
  Fix $\alpha \in (0,1)$, $\alpha_\kn\in(0,1)$, and $d \in \NN$. For each $j\in[m]$, choose sufficient statistic $S_j = (\bm X_{-j}, \bm Y)$. %
  Construct a resample $\bm X'$ of $\bm X$ by letting $\bm X'_{-j} = \bm X_{-j}$ and $ X'_{ij}$ be a sample from $\PP_{X_{ij} \cond X_{i, -j}}$. Using the new, resampled dataset $(\bm X', \bm Y)$, run the derandomized knockoffs procedure (with hyperparameters $d$ and $\alpha_\kn$) at level $\alpha$ to get the averaged e-values  $\tilde e_1, \dots, \tilde e_m$, 
  where $\tilde e_j$ is constructed as \eqref{eq:kth_knockoffs_eval} and \eqref{eq:averaged_knockoffs_eval}. Then 
  $
  (\tilde e_1, \dots, \tilde e_m) \cond S_j \sim (\bar e_1, \dots, \bar e_m)\cond S_j
  $
  under $H_j$. Furthermore, independent resamples $\bm X'$ will lead to independent 
  $(\tilde e_1, \dots, \tilde e_m)$ conditional on $S_j$ (also under $H_j$).
\end{proposition}

\begin{proof}
The proposition follows once we show the distributional equality between the old and resampled columns $\bm X_j$ and $\bm X'_j$ (conditional on $\bm X_{-j}$ and $\bm Y$) under the null $H_j$. The null states that $\bm Y \indep \bm X_j \cond \bm X_{-j}$. Since the resample was thus created without any extraneous knowledge of $\bm Y$ (outside of what is contained when conditioning on $\bm X_{-j}$), the desired property simplifies to
\begin{equation}
    \bm X_j \cond \bm X_{-j} \sim \bm X'_j \cond \bm X_{-j}.
\end{equation}
However, the above is evident since we resampled $\bm X'_j$ by the law $\otimes_{i=1}^n \PP_{X_{ij} \cond X_{i, -j}}$. 
\end{proof}

\section{Additional simulation results}
\label{appd:additional_sims}
\subsection{Marginal boosting is not enough}
\label{appd:marg_boost_not_enough}

As referenced in Section \ref{sec:background},~\cite{wang2022false} detail a method to 
boost a collection of e-values using marginal distributional information. To compare 
their boosting mechanism with e-BH-CC, which leverages conditional distributional information,
we run numerical experiments in the context of one-sided $m$-dimensional $z$-testing
(Section \ref{sec:example-parametric}). 

Let $\be = (e_1, \dots, e_m)$ be a collection of e-values constructed as in 
\eqref{eq:zstat-evalue}, with $a_j$ set to some fixed $\delta > 0$ for all $j\in[m]$. 
These likelihood ratio test (LRT) statistic e-values are distributed log-normally. 
This marginal distributional information can be used to increase each e-value $e_j$
by a multiplicative factor $b$, which is the solution to
\begin{equation}
    \label{eq:marg_boost_eqn}
    b \Phi\bigg(\frac \delta2 + \frac{\log (\alpha b)}{\delta }\bigg) = 1,
\end{equation}
with $\Phi(\cdot)$ representing the standard Gaussian cumulative density 
function~\citep{wang2022false}. As $b$ turns out to be greater than 1, running e-BH
on $(be_1, \dots, be_m)$ is uniformly more powerful than doing so on the
original collection while still controlling the FDR at level $\alpha$.

Using the setup of the experiments detailed in Section \ref{sec:sim-parametric}, we 
implement marginally-boosted e-values, which will be denoted 
$b\be \coloneqq (be_1, \dots, be_m)$ (hiding the dependence of $b$ on the parameters 
used in the construction of $e_j$). Importantly, we run all experiments at the target
level of $\alpha = 0.05$ (with additional Monte-Carlo error of $\alpha_0 = 0.005$; see 
the relevant simulations for details). We find that, empirically, 
$$
    |\cR(\be^\boost )| \ge |\cR(b\be)| \ge |\cR(\be)|
$$
over all choices of $\delta$ (for both correctly-specified and misspecified LRT e-values).
Although $\cR(b\be)$ can show significant improvement over base e-BH---especially when the latter
is nearly powerless due to misspecified e-values---there continues to be a wide power gap 
with respect to the p-value BH baseline. The power of e-BH-CC, on the other hand, is higher 
than marginal boosting and is consequently more comparable to BH.
The observation that $\cR(\be^\boost)$ is more powerful than
marginal boosting continues to hold even when the signal strength 
is not high.
In a sense, e-BH-CC is able to use conditional calibration to reclaim 
significantly more of the 
power drop, while $\cR(b\be)$ only makes limited progress.

\begin{figure}[tb]
    \centering
    \includegraphics[scale=1.0]{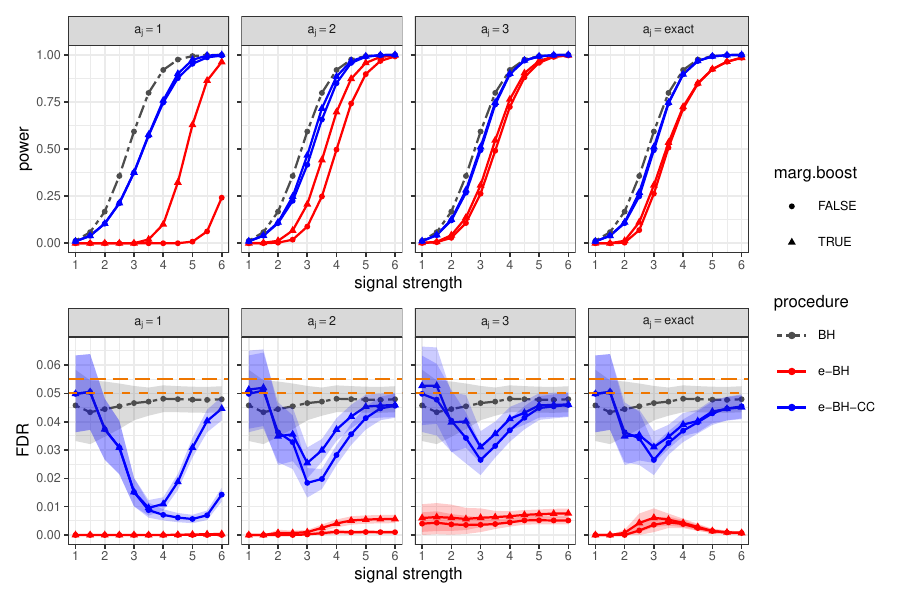} 
    \caption{Comparison of power and FDP between e-BH-CC and marginal boosting~\citep{wang2022false} 
    in the setting of one-sided $z$-testing (Section \ref{sec:sim-parametric}). In addition, the hybrid
    procedure of using e-BH-CC on marginally-boosted e-values is also included. Each plot contains the 
    averaged metric over 1000 replications. Shading represents error bars of two standard errors above and below.}
    \label{fig:zstat_margboost}
    \centering
\end{figure}

In addition, we observe that we are easily able to generate resamples of $b\be\cond S_j$ as well (since the boosting factor is constant). As the simulations studies show that marginal boosting is not enough for tightening the FDR inequality, we consider using CC to recover the rest of the power gap. Since $be_j$ is not a valid e-value, we cannot directly run e-BH-CC on the collection $b\be$ as the FDR guarantee would be $b\alpha$, which is too large. Instead, we use the original e-value as the budget:
\begin{equation}
    \phi_j(c;S_j)  = \EE\Bigg[ \frac m\alpha \cdot\frac{\indc{c\cdot be_j \ge \frac{m}{\alpha |\hcR_j(b\be)|}}}{\hcR_j(b\be)} - e_j\Biggiven S_j \Bigg].
\end{equation}
{By using the above $\phi$-function in the boosting step, 
the resulting ``doubly-boosted'' e-values, which we denote as $(b\be)^{\boost}$ 
in a serious abuse of notation, are valid e-values. 
The rejection set returned by running e-BH on these 
doubly-boosted e-values has the theoretical guarantee 
$\cR((b\be)^{\boost}) \supseteq \cR(\be)$, although it fails 
to also guarantee containment of the marginally-boosted 
improvement $\cR(b\be)$.} Empirically, however, 
we find that 
$$|\cR((b\be)^\boost)| \ge |\cR(\be^\boost )|.$$
Generally, the two methods are identical. However, the hybrid method of combining the two boosting types 
can show a small power improvement over $\cR(\be^\boost)$ in specific settings. One can interpret this result as a sign that 
The e-BH-CC procedure itself does not recapture the \emph{entire} power drop; 
with improvements in design, the gap can be closed even further.

Lastly, we call attention to the FDP comparison between marginal boosting and e-BH-CC. It is evident
that conditional calibration uses the FDR budget much more efficiently than its marginal counterpart,
as it consistently achieves much higher realized FDP. 

These results are presented in visual form in Figure \ref{fig:zstat_margboost}. We separate the e-value
construction ($\be$ versus $b\be$) with the multiple testing procedure used (e-BH versus e-BH-CC) so that 
it is straightforward to visually compare $|\cR(\be)|, |\cR(b\be)|, |\cR(\be^\boost)|,$ and $|\cR((b\be)^\boost)|$.

\subsection{One-sided $z$-testing with negative $\rho$}
\label{appd:zstat_sims_rho_neg}

\begin{figure}[tb]
    \centering
    \includegraphics[scale=1.0]{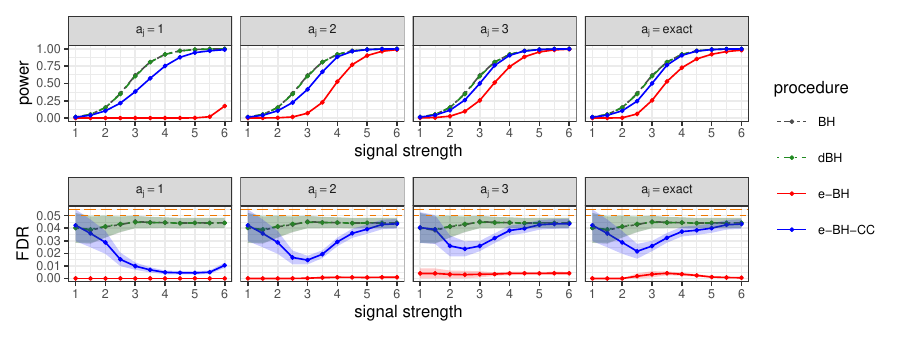} 
    \caption{Realized power and FDP of the simulated experiments for $z$-testing. 
    Each plot contains the averaged metrics over 1,000 replications. 
    Outside of setting $\rho = -0.5$, the simulation details are the exact same as those for Figure \ref{fig:zstat}.} 
    \label{fig:zstat_rho_neg}
    \centering
\end{figure}

The setting of the numerical experiments for Section \ref{sec:example-parametric} were chosen so that the BH procedure would be applied over $p$-values satisfying the PRDS condition (see footnote \ref{foot:PRDS}), implying that the resulting rejection set has theoretically-guaranteed FDR control \citep{benjamini2001control}. This situation allows us to treat BH as a true threshold for e-BH and e-BH-CC, especially as we want to contextualize the power gain of the latter. However, one benefit of e-BH is that FDR control is achieved regardless of the dependence structure between the $z$-statistics, so it would be relevant to do similar numerical experiments in a setting where the PRDS condition no longer holds. By showing that e-BH-CC still improves considerably over e-BH, we are able to present e-BH-CC as an alternative to BH in cases where there are no sufficient conditions for the latter's FDR control, rather than BH being a theoretically valid baseline. Thus, we redo the experiments in Section \ref{sec:sim-parametric} but for $\rho = -0.5$, which no longer leads to PRDS one-sided p-values (all the other details are kept exactly the same). 
\red{
In this setting, we also lack the conditional PRDS property, which means dBH~\citep{fithian2020conditional} is no longer guaranteed to uniformly improve BH.}

Figure \ref{fig:zstat_rho_neg} shows the resulting mean power and FDP curves. We draw the same qualitative conclusions as those from the $\rho = 0.5$ experiments in the main text. Notably, the BH procedure empirically controls the FDR at $\alpha = 0.05$ even though such control was not theoretically guaranteed. Regardless, e-BH-CC performs comparatively to BH, but does have such guarantees.

\subsection{One-sided $t$-testing}
\label{appd:tstat_sims}

In Appendix \ref{sec:tstat}, the problem of $t$-testing was studied from the e-value perspective. Here we provide companion simulation studies which serve as empirical evidence of the power gain achieved by e-BH-CC. Again, given $Z\sim \cN_m(\mu, \Sigma)$, the objective is to identify a set of $H_j \colon \mu_j \neq 0$ to reject in favor of their corresponding $H_j^{\text{alt}}\colon \mu_j > 0$ in a manner that controls the FDR.

The experiment setup is exactly the same as that of the one-sided positive $z$-testing simulation studies outlined in Section \ref{sec:sim-parametric}, with computational details in Appendix \ref{appd:deferred_details_sims_parametric}, but there are more hyperparameters for this setting. Recall that here, $\Sigma$ is partially unknown in that $\Sigma =\sigma^2 \Psi$, where $\sigma^2$ is an unknown scalar and $\Psi$ is a known positive definite matrix. 
We set  $\sigma^2=1.1$ and vary $\rho \in \{0.9, -0.9\}$ to construct $\Psi$ such that $\Psi_{ij} = \rho^{|i-j|}$. We also vary the degrees of freedom $n-m$ in $\{5, 50\}$; the former corresponds to the heavy-tailed setting while the latter corresponds to the light-tailed setting. In the heavy-tailed setting, 
we vary $\Delta \in \{2, 2.5, 3, \dots, 6.5\}$; 
in the light-tailed setting, 
we vary $\Delta \in \{1, 1.5, 2, \dots, 5.5\}$.

The mean power and FDP curves of e-BH, e-BH-CC, and BH as a baseline (which does not have theoretical FDR guarantees when $\rho = -0.9$, so is not a true baseline), taken over 1000 replications, are shown in Figure \ref{fig:tstat}. 
Like in the $z$-testing simulations, we see that e-BH-CC improves upon e-BH  in all settings. 
The improvement is notable when the signals are stronger, 
although when the signals are weak all three methods (including BH) suffer in power. 
Note that in these settings, BH does not necessarily control the FDR; however, we still find it pertinent to consider the gap between it and e-BH a loss of power. In terms of reclaiming the power gap between BH and e-BH, e-BH-CC does quite well and is even comparable to BH in the light-tailed experiments. To conclude, we again note that e-BH-CC controls the FDR empirically and shows a better usage of the FDR budget compared to e-BH.

\begin{figure}[tb]
    \centering
    \includegraphics[scale=1.0]{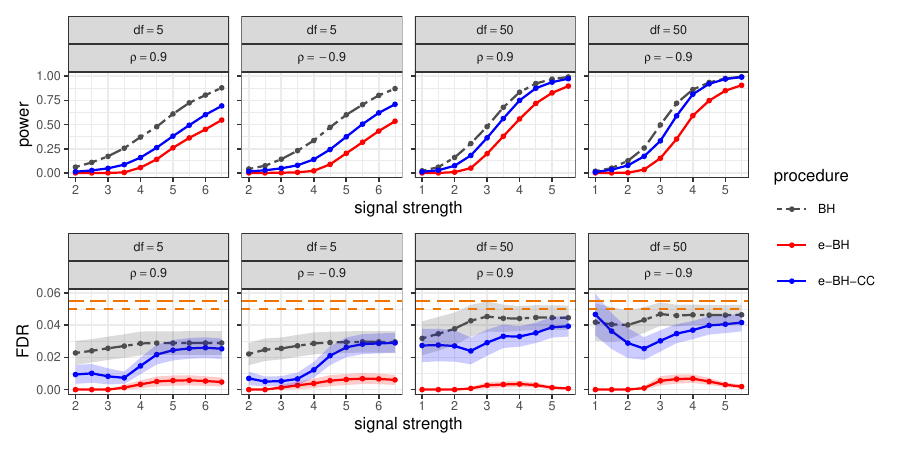} 
    \caption{Realized power and FDP of the simulated experiments for one-sided $t$-testing. 
    The details are otherwise the same as in Figure~\ref{fig:zstat}.}
    \label{fig:tstat}
    \centering
  \end{figure}

\subsection{Two-sided $t$-testing}

For the one-sided alternative $H_j^{\text{alt}}\colon \mu_j > 0$, we used the LRT statistic \emph{for a specific point alternative} as our e-value. Choosing some $a_j > 0$, we defined
$$
    e_j = \frac{f_{t_{n-m}, a_j}(T_j)}{f_{t_{n-m}, 0}(T_j)},
$$
where $f_{t_{n-m}, a}$ denotes the density function of the noncentral $t$-distribution with noncentrality parameter $a$ and degrees of freedom $n-m$. 

Consider the problem where the direction of the alternative is not known. The two-sided problem tests the null $H_j\colon \mu_j = 0$ against the alternative $H_j^{\text{alt}}\colon \mu_j \neq 0$. Using the LRT e-value with chosen parameter $a_j$ may be powerless when the true data-generating parameter has the opposite sign of $a_j$; to address this issue, one can use the average of two LRT e-values, constructed with parameters $a_j$ and $-a_j$:
\begin{equation}\label{eq:twosided-tstat-eval}
    e_j =  
\frac{(f_{t_{n-m}, a_j}(T_j) + f_{t_{n-m}, -a_j}(T_j))/2}{f_{t_{n-m}, 0}(T_j)}.
\end{equation} 
As our implementation of e-BH-CC revolves around resampling $T_j \cond S_j$, we can use the details in Appendix \ref{sec:tstat} and the construction in \eqref{eq:twosided-tstat-eval} to resample from the conditional distribution $(e_1, \dots, e_m) \given S_j$. 
A similar construction can be used for the two-sided $z$-testing problem. 

We show the results of the numerical experiments for the two-sided $t$-testing problem in Figure \ref{fig:tstat-twosided}. All the experiment details (e.g., choice of $\rho$, values for the signal strength $\Delta$, computational details) are exactly the same as those of the one-sided $t$-testing simulations covered in Appendix \ref{appd:tstat_sims} except for the true data-generating distribution---since this is a two-sided problem, the true $\mu$ vector is constructed as 
\$
\mu = (\underbrace{B_1\cdot \Delta, B_2\cdot \Delta, \dots, B_{10}\cdot \Delta}_{10}, 0, \dots, 0) \in \RR^{100},
\$
where $B_1, \dots, B_{10}$ are i.i.d. Rademacher random variables independent of everything else.

As the power and FDP curves show, we experience a gain in power by using conditional calibration. The power gain occurs for all signal strengths for each choice of $\rho$ and $n-m$, as expected. When the degrees of freedom is large, the power of e-BH-CC rivals that of BH for both $\rho = 0.9$ and $-0.9$. Checking the realized FDP of e-BH-CC, we see that e-BH-CC consistently achieves a higher FDP, showing that we address the looseness in the usage of the FDR budget by the base e-BH procedure.

\begin{figure}[tb]
    \centering
    \includegraphics[scale=1.0]{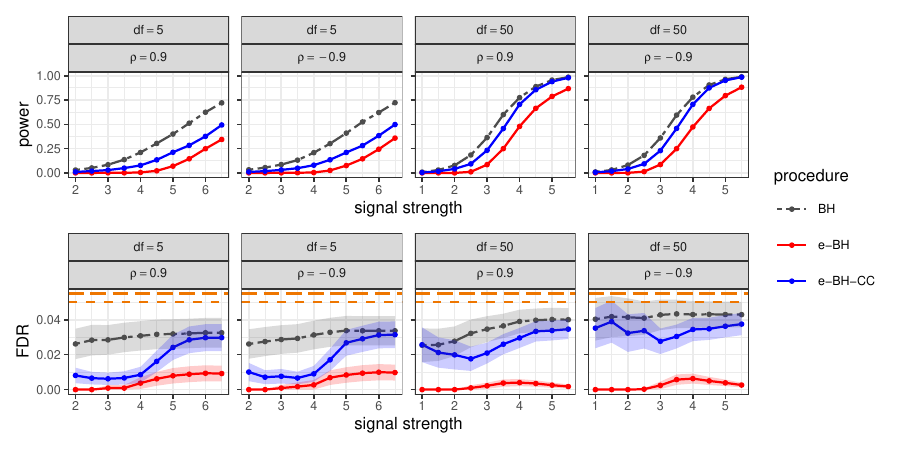} 
    \caption{Realized power and FDP of the simulated experiments for two-sided $t$-testing. 
    The details are otherwise the same as in Figure~\ref{fig:tstat}, the one-sided version of this experiment.}
    \label{fig:tstat-twosided}
    \centering
\end{figure}

\subsection{Additional results of e-BH-CC for derandomized knockoffs under the linear model}
\label{appd:add-sim-kn}
The results of the experiment in Section~\ref{sec:sim-knockoffs} with 
$|\cH_1| = 10$ are presented in Figure~\ref{fig:mxkn-tp10}. The derandomized knockoffs procedure 
experiences a major power drop over all signal amplitudes for both $d\in\{10, 30\}$. The resulting power gain from e-BH-CC is major---for $d=10$, its power is comparable to the original knockoffs filter, while for $d=30$, its power looks to be marginally better. 
These results suggest that e-BH-CC avoids a phase transition at the threshold, unlike regular knockoffs: its power curves between Figures \ref{fig:mxkn-tp10} and \ref{fig:mxkn-tp9} differ slightly, while those of the knockoffs filter differ greatly, faltering when $|\cH_1| <  1/ \alpha$.
\begin{figure}[tb]
    \centering
    \includegraphics[scale=1.0]{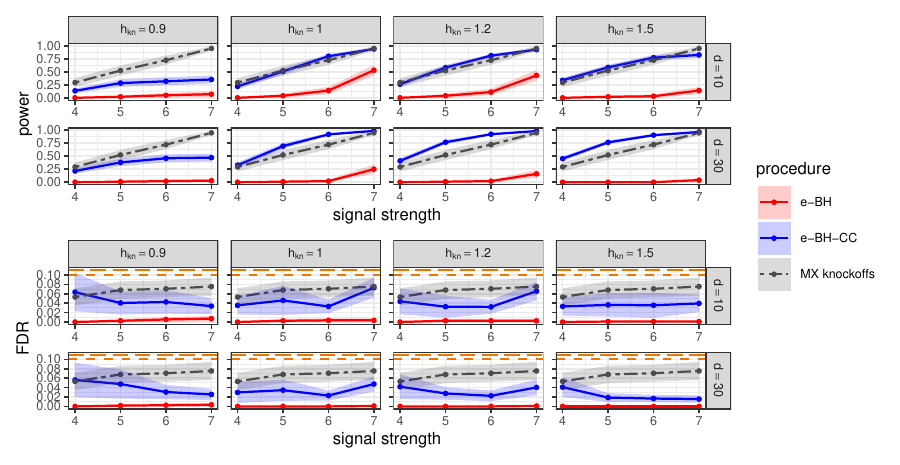}
    \caption{Realized power and FDP for the knockoffs simulations when $\beta$ is sparse (specifically, $|\cH_1|=10$ and $\alpha=0.1$, so this is \emph{at} the threshold). 
    The setup details are otherwise the same as for the power results in Figure \ref{fig:mxkn-tp25}.
    }
    \label{fig:mxkn-tp10}
    \centering
  \end{figure}
\subsection{Derandomized knockoffs under the logistic model}

Model-X knockoffs allows for FDR-controlling variable selection under arbitrary and unspecified models for $Y\given (X_1, \dots, X_m)$ as long as the distribution of $X = (X_1, \dots, X_m)$ is known. In the simulations found in Section \ref{sec:sim-knockoffs}, the model was conveniently chosen to be a Gaussian linear model. Here, we present results from the companion experiments where the response now follows a logistic model.

Specifically, let $m=200$ and $n=500$. Choosing the coefficient vector to be 
$$
  \beta = \bigg(\underbrace{0, \ldots, 0}_z, \frac{\Delta}{\sqrt{n}}, \underbrace{0, \ldots, 0}_z,  -\frac{\Delta}{\sqrt{n}}, \ldots \bigg) \in \RR^m
$$
for a pre-specified range of $\Delta$, we simulate $n$ covariate-response pairs $(X_{i1}, \dots, X_{im}, Y_i)_{i\in[n]}$ from the distribution $(X,Y)$ where
$$
    Y\given X \sim \Bern\bigg(  \frac{\exp({X^\top \beta})}{1 + \exp({X^\top \beta})} \bigg)
$$
and $X$ follows the same distribution as in Section \ref{sec:sim-knockoffs}.

Similar to the linear model simulations, we will consider a dense setting  ($z= 7, |\cH_1| = 25$) and a sparse setting ($z = 19, |\cH_1| = 10$) in order to capture the benefit of conditional calibration across different settings. We reiterate that the latter setting will experience the threshold phenomenon for regular, non-derandomized model-X knockoffs. In both settings, the FDR target is chosen to be $\alpha = 0.1$. For the dense setting, we choose to vary $\Delta \in \{ 8, 10, 12, 14 \}$ and $\alpha_\kn = \alpha h_\kn, h_\kn \in \{0.5, 0.75, 0.9, 1.0\}$. 
For the sparse setting, $\Delta \in \{ 8, 10, 12, 14, 16, 18\}$ and $\alpha_\kn = \alpha h_\kn, h_\kn \in \{1.0, 1.2\}$. Regarding the latter, we chose to restrict our attention to fewer choices of $\alpha_\kn$ and instead simulate over a wider range of signal strengths to highlight the struggle of default model-X knockoffs procedure relative to e-BH-CC in the sparse logistic model setting. 

\paragraph{Methods.} The choice of model, range of $\Delta$, range of $h_\kn$, and target FDR level $\alpha$ has already been communicated above. We choose $d=10$ derandomization runs for derandomized knockoffs and its conditionally calibrated version (for both settings). 

\paragraph{Knockoffs details.} See Section \ref{sec:sim-knockoffs} and Appendix \ref{appd:deferred_details_sims_knockoffs}.

\paragraph{Computational details.} Similar to the simulations in Section \ref{sec:sim-knockoffs}, we make use of the hybrid AVCS. The first 600 samples (in batches of 100) are used to construct the exact
AVCS, while the next 400 samples are used to construct the asymptotic AVCS. After 1000 total samples,
if the resulting AVCS still contains zero we stop early and fail to boost.

We also filter out potential hypotheses to decrease the number of e-values we boost using conditional calibration: see Appendix \ref{appd:deferred_details_sims_knockoffs} for the construction of the filter. Lastly, we choose $\alpha_0 = 0.1 \cdot \alpha = 0.01$ and use $\alpha_\avcs$ as described in Algorithm \ref{alg:ebh_cc_acvs} replacing $m$ with the size of the filter.

\paragraph{Experiments and results.} See Section \ref{sec:sim-knockoffs} and Appendix \ref{appd:deferred_details_sims_knockoffs} for the experiment details; we will only cover the results. Figure \ref{fig:mxkn-tp25-logistic} shows the power and FDP curves for the dense experiments, while Figure \ref{fig:mxkn-tp10-logistic} shows the same for the sparse experiments. For both settings, the improvement from e-BH-CC over both derandomized knockoffs and default model-X knockoffs is unmistakable (across all signal strengths and choices of $h_\kn$). Importantly, e-BH-CC is able to detect a nontrivial fraction of nonnull variables even when the model-X knockoffs procedure has near zero power. At larger signal strengths, the difference in power is even more stark. For example, in the sparse setting experiments, the average difference in power between e-BH-CC and default knockoffs at $\Delta = 18$ is 0.483 for $h_\kn = 1.0$ and 0.499 for $h_\kn = 1.2$.

In terms of the FDP, we note that it is generally controlled at the target FDR $\alpha = 0.1$. In some cases, we see slight FDR violation over the Monte-Carlo error adjusted target $0.11$, but in such cases violation is not large (and the error bars, which show the approximate 95\% confidence interval using computed standard errors, contain both target levels). We posit that the reason for such is not the choice of model, but rather the choice of Monte-Carlo hyperparameters. Note that for the corresponding linear-model experiments described in Appendix \ref{appd:deferred_details_sims_knockoffs}, we use 1200 resamples until we switch to the asymptotic AVCS, while here we only take 600 resamples. Thus, the asymptotic AVCS at 600 resamples can exhibit looser coverage, which has evident downstream effects on FDR control. With a larger number of resamples before switching to the asymptotic AVCS, the resulting FDP will be lower and even less volatile---however, for this set of experiments the main point was to show the power improvement from e-BH-CC in a different model for $Y\given X$.

\begin{figure}[tb]
    \centering
    \includegraphics[scale=1.0]{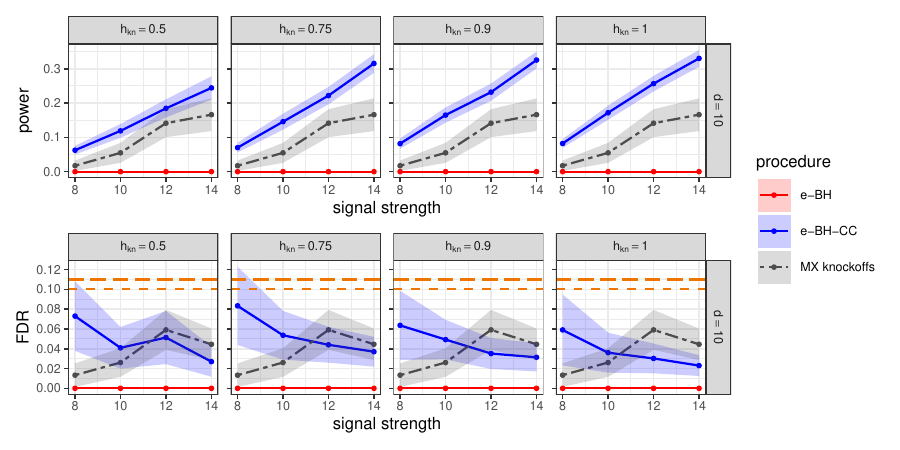}
    \caption{Realized power and FDP for the knockoffs simulations when $\beta$ is dense ($|\cH_1|=25$) and $Y\given X$ follows a logistic model. 
    Each facet contains the averaged metrics over 100 replications; shading indicates error bars. 
    The short and long orange dashed lines in the FDP plots represent the target FDR (0.1) and the FDR bound with Monte-Carlo error (0.11), respectively.
   }
    \label{fig:mxkn-tp25-logistic}
    \centering
\end{figure}

\begin{figure}[tb]
    \centering
    \includegraphics[scale=1.0]{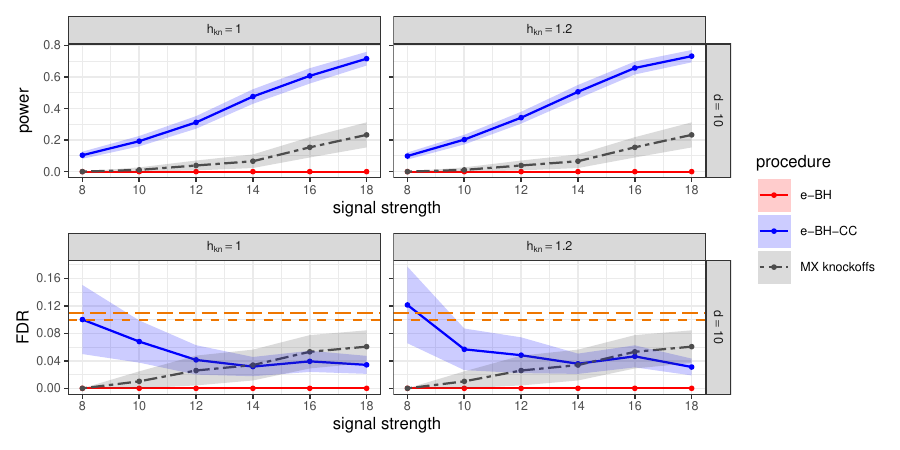}
    \caption{Realized power and FDP for the knockoffs simulations when $\beta$ is sparse (specifically, $|\cH_1|=10$ and $\alpha=0.1$, so this is \emph{at} the threshold) and $Y\given X$ follows a logistic model. 
    Each facet contains the averaged metrics over 100 replications; shading indicates error bars. 
    The short and long orange dashed lines in the FDP plots represent the target FDR (0.1) and the FDR bound with Monte-Carlo error (0.11), respectively.
   }
    \label{fig:mxkn-tp10-logistic}
    \centering
\end{figure}

\subsection[Application to HIV drug resistance data]{Application to HIV drug resistance data}
\label{appd:hiv_knockoffs_realdata}

We also revisit the HIV-1 drug resistance application studied in
\citet{ren2024derandomised}, which compares the performance of
derandomized knockoffs to their original variant.
The high-level scientific objective is to identify mutations
for HIV-1 that are associated with higher drug 
resistance~\citep{rhee2006genotypic}. We defer to the above papers
for a more detailed account of the problem setting.

Our goal is to continue the simulation study 
done by~\citet{ren2024derandomised} to showcase the power
gain of e-BH-CC over regular derandomized knockoffs. To this end,
we follow basically all of their instantiation choices
in Section 5 of their paper, and apply e-BH-CC on top of it.

After preprocessing in the same manner as~\citet{ren2024derandomised}, 
our analysis contains $219$ mutation indicators. 
To use knockoffs, they utilize the
second-order Gaussian knockoff sampler~\citep{candes2018panning}
using the empirical first and second moments of the
mutation matrix. For the feature importance statistic $W(\cdot)$,
we follow them in using the lasso coefficient difference (LCD) statistic
based on a cross-validated lasso fit.

We set the target level as $\alpha =0.1$ and use $d=20$ knockoff copies
for derandomization purposes (this is where we deviate from 
\citet{ren2024derandomised}, as they use $d=50$---for the feasibility of 
e-BH-CC, which requires resampling, we make this change). In addition, 
we follow their recommendation and set
$\alpha_\kn=\alpha/2$ (equivalently $h_\kn=0.5$). This is the 
derandomized knockoffs baseline we use. For a baseline pertaining
to the original model-X knockoffs filter, we also choose a ``majority vote''
comparison: the ordinary knockoff comparison is computed by
running the usual knockoff filter at level $\alpha=0.1$ for each of the
$d=20$ knockoff copies and reporting the mutations selected in more than half of
these runs. This is not a valid FDR-controlling selection procedure, but as
\citet{ren2024derandomised} are concerned with stability of algorithms, they
choose it as a nice baseline. For our purposes, we keep it as well.

After making selections, we cross-reference each selected mutation
against the Stanford HIV Drug Resistance Database protease inhibitor mutation
comments, available at
\url{https://hivdb.stanford.edu/dr-summary/comments/PI/}. A mutation is counted
as verified if its annotation in that reference is not \texttt{NA}. Table
\ref{tab:hiv_ebh_cc} reports the number of distinct mutation discoveries and,
in parentheses, the number of HIVDB-verified discoveries. The e-BH-CC procedure
selects $73$ raw mutation columns, but two of these columns correspond to the
same mutation label; the table therefore reports $72$ distinct mutations.
For context, the majority-over-$50$ summaries reported in
\citet{ren2024derandomised} contain $74$ mutations for original knockoffs and
$63$ mutations for derandomized knockoffs, with $39$ and $38$ HIVDB-verified
discoveries, respectively.

\begin{table}[t!]
\centering
\begin{tabular}{lc}
\toprule
Method & Discoveries (verified) \\
\midrule
Original knockoffs, majority-over-20 & $71~(39)$ \\
Derandomized knockoffs & $62~(38)$ \\
e-BH-CC & $72~(39)$ \\
\bottomrule
\end{tabular}
\caption{Discoveries on the HIV-1 protease inhibitor resistance data, using
$\alpha_\kn=\alpha/2$ for derandomized knockoffs and e-BH-CC. The numbers in
parentheses are the discoveries cross-referenced as verified by the HIVDB
protease inhibitor mutation comments.}
\label{tab:hiv_ebh_cc}
\end{table}

In our version of their real data application, we find that  
the e-BH-CC procedure recovers the power lost by the derandomized e-value aggregation. 
Derandomized knockoffs selects $62$ distinct mutations,
while e-BH-CC selects $72$ distinct mutations. 
The boosted rejection set contains all
$62$ derandomized discoveries and discovers $10$ additional mutation labels. One of
these additions, 71I, is verified by HIVDB, increasing the number of verified
discoveries from $38$ to $39$. Comparing e-BH-CC to the majority-over-20 procedure,
we find that we make one more total discovery and match the same number 
of verified discoveries: it contains all $71$ mutations in that
aggregated rejection set and adds the unverified mutation 61E.

\section{Additional details in the real data analysis}
\label{appd:real_data}
\paragraph{Computation of the likelihood ratio in task (1).}
\$
\frac{dP_{X,Y(1) \given T=0,Y(1)\le 0.3}}
{dP_{X,Y(1) \given T=1, Y(0) \le 0.3}} 
& = \frac{dP_{Y(1) \given X,T=0, Y(1) \le 0.3}}{dP_{Y(1) \given X,T=1, Y(1) \le 0.3}}  
\cdot \frac{dP_{X \given T = 0, Y(1) \le 0.3}}{dP_{X \given T = 1, Y(1) \le 0.3}}\\
& \stepa{=} \frac{dP_{X \given T = 0, Y(1) \le 0.3}}{dP_{X \given T = 1, Y(1) \le 0.3}}\\
& \stepb{=} \frac{\PP(T=1, Y(1) \le 0.3)}{\PP(T=0, Y(1)\le 0.3)} \cdot \frac{1-e(x)}{e(x)}\\ 
& \propto \frac{1-e(x)}{e(x)}.
\$
Above, steps (a) and (b) follow from the strong ignorability assumption.

Given a nonconformity score function $V(x,y)$, 
we let $\hat V_i = V(X_i,0.3)$ for any $i \in \tD \coloneqq \tD_\calib \cup \cD_\test(0)$ 
(note that the nonconformity scores here do not depend on the outcome).
Next, we define the conformal e-value for each $j \in \cD_\test(0)$:
\begin{align}
    \begin{split}
        e_j &= \Big(w(X_{j}) +\sum_{i \in \tD_\calib} w(X_i)\Big) \cdot 
        \frac{\ind \{\hat V_{j} \ge T_j\}}{w(X_{j}) + \sum_{i\in \tD_\calib} w(X_i) \ind\{\hat V_i \ge T_j\}},\\
        \text{where}\;T_j &= \inf\left\{t \in \{\hat V_i\}_{i\in \tD }
        \colon 
        \frac{ m }{w(X_{j})+\sum_{i \in \tD_\calib} w(X_i)} \cdot 
        \frac{w(X_{j})+\sum_{i \in \tD_\calib} w(X_i)  \ind\{\hat V_i \ge t \}  }
        {( \sum_{k \in \cD_\test(0)} \ind\{\hat V_{k} \ge t\} ) \vee 1} \le \alpha \right\}.
    \end{split}
\end{align}

To apply e-BH-CC, we choose the sufficient statistic $S_j = (\cE_j, \{X_\ell\}_{\ell \in \cD_\test(0) \backslash \{j\}})$, 
where now $\cE_j$ is the unordered set of $\{X_i\}_{i \in \tD_\calib\cup \{j\}}$.
Conditional on $(S_j,\cH_0)$, we sample
\$ 
\tX_{j} \sim X_{j} \given \cE_j, \{X_{\ell}\}_{\ell \in \cD_\test(0)\backslash \{j\}}, \cH_0
\sim   \sum_{X \in \cE_j} \frac{w(X)}{\sum_{X' \in \cE_j} w(X')} \delta_{X},
\$
and assign the remaining $X_{\ell}$'s arbitrarily to $\tX_{\ell}$ in $\cD_\test(0)\backslash \{j\}$.
The resampled e-values can be computed according to their definition.

\paragraph{Details for task (2).}

For task (2), our target is to discover the individuals in $\cD_\test(1)$ 
with a counterfactual less than $-0.3$: the null hypothesis for any $j \in \cD_\test(1)$ is 
\@ \label{eq:hypothesis_ite_att}
H_{j}: Y_j(0) \ge -0.3.
\@

In this case, we include the subset of $\cD_\calib(0)$ with $Y_i \ge -0.3$
as the calibration set, which we abuse the notation and also denote it as $\tD_\calib$.
For any $i\in \tD_\calib$, 
\$ 
X_i, Y_i \given T_i = 0, Y_i \ge -0.3 \sim P_{X,Y(0)\given T=0, Y\ge -0.3}.
\$
For any $j\in \cH_0$, 
\$
X_j,Y_j(0) \given T_j = 1, \cH_0 \sim P_{X,Y(0)\given T=1, Y(0)\ge -0.3}. 
\$
The likelihood ratio between the test and calibration inliers is 
\$
\frac{dP_{X,Y(0) \given T=1,Y(0)\ge -0.3}}
{dP_{X,Y(0) \given T=0, Y(0) \ge -0.3}} 
& = \frac{dP_{Y(0) \given X,T=1, Y(0) \ge -0.3}}{dP_{Y(0) \given X,T=0, Y(0) \ge -0.3}}  
\cdot \frac{dP_{X \given T = 1, Y(0) \ge -0.3}}{dP_{X \given T = 0, Y(0) \ge -0.3}}\\
& \stepa{=} \frac{dP_{X \given T = 1, Y(0) \ge -0.3}}{dP_{X \given T = 0, Y(0) \ge -0.3}}\\
& \stepb{=} \frac{\PP(T=0, Y(0) \ge -0.3)}{\PP(T=1, Y(0) \ge -0.3)} \cdot \frac{e(x)}{1-e(x)}\\ 
& \propto \frac{e(x)}{1-e(x)} \,=:\, w(x).
\$
The rest of the implementation of 
all procedures is the same as in task (1).

\subsection{Implementation details.}
Since NSLM is an observational dataset, the propensity score function is unknown. We can nevertheless 
obtain an estimate of the propensity score function, $\hat e(x)$, with the training fold $\cD_\train$ by regressing 
$T$ on $X$;  we also fit a regression function $\hat m(x)$ for $Y(1)$ given $X$ with $\cD_\train$, 
and let $V(x,y) = \hat \mu(x) - y$. %

In our implementation, $|\cD_\train| = 8,\!000$, $|\cD_\calib| = 1,\!000$, 
and $|\cD_\test| = 1,\!000$. Random forests are used to estimate the propensity 
score function $\hat e(x)$ and 
the regression function $\hat m(x)$ for $Y(1) \given X = x$ with $\cD_\train$, 
using the implementation in the \texttt{scikit-learn} 
package in Python. 
For e-BH-CC, we take $\alpha_\cc$ to be a different value than $\alpha$ (recall 
Remark~\ref{remark:ebh_cc}), where $\alpha_\cc = 1.1\cdot \alpha$.
The AVCS-approximated e-values are adopted with $\alpha_0/\alpha = 0.0001$; 
the first $600$ samples (in batches of $100$) are used for constructing the non-asymptotic 
AVCS, and the next $400$ for the asymptotic AVCS.
If no conclusions can be drawn within $1000$ Monte Carlo samples, we stop early 
and fail to boost the particular e-value.

\section{Algorithms for e-BH-CC}
\label{appd:algos}
\begin{algorithm}[ht]
    \caption{e-BH-CC with CI-approximated boosted e-values}
    \label{alg:ebh_cc_ci}
    \DontPrintSemicolon
    \KwIn{e-values $e_1,e_2,\ldots,e_m$; sufficient statistics $S_1,S_2,\ldots,S_m$; 
    target FDR level $\alpha$; number of Monte-Carlo samples $K$;  
    Monte-Carlo error budget $\alpha_0$.}
    $\alpha_\ci \leftarrow \alpha_0 |\cR^\ebh(\be)|/m$.\;
    \For{$j \in [m]$}{
     Generate $\tilde{\be}^{(1)}, \ldots, \tilde{\be}^{(K)} \stackrel{\text{i.i.d.}}{\sim} \be \given S_j$.\;
     Compute $E_{j,1},\ldots,E_{j,K}$.\;
     Construct a $(1-\alpha_{\ci})$ confidence interval $C_{j,K}$ for 
     $\phi_j(\tilde{c}_j;S_j)$.\;
     Construct the CI-approximated boosted e-values as in~\eqref{eq:ci_boosted_eval}.
    }
    Apply e-BH to $(\be_1^{\boost,\ci},\ldots,\be^{\boost,\ci}_m)$ at level $\alpha$ and 
    obtain the rejection set $\cR^\ebh(\be^{\boost,\ci})$.\;
    \KwOut{the rejection set $\cR^\ebh(\be^{\boost,\ci})$.}
\end{algorithm}
    
\begin{algorithm}[ht]
    \caption{e-BH-CC with ACVS-approximated boosted e-values}
    \label{alg:ebh_cc_acvs}
    \DontPrintSemicolon
    \KwIn{e-values $e_1,e_2,\ldots,e_m$; sufficient statistics $S_1,S_2,\ldots,S_m$; 
    target FDR level $\alpha$; Monte-Carlo error budget $\alpha_0$.}
    $\alpha_\ci \leftarrow \alpha_0 |\cR^\ebh(\be)|/m$.\;
    \For{$j \in [m]$}{
     \textbf{Initialization:} $k \leftarrow 0$ and $C_{j,0} = \RR$.\; 
     \While{$U_{j,k} >0$ \textnormal{and} $L_{j,k}\le0$}{
     $k \leftarrow k+1$.\;
     Generate $\tilde{\be}^{(k)}{\sim} \be \given S_j$.\;
     Compute $E_{j,1},\ldots, E_{j,K}$.\;
     Construct the $k$th interval $C_{j,k}$ in the $(1-\alpha_{\avcs})$ confidence sequence for $\phi_j(\tilde{c}_j;S_j)$.\;
     }
      Construct the AVCS-approximated boosted e-values $e_j^{\boost,\avcs} = m\cdot \frac{\ind\{U_{j,k}\le 0\}}{\alpha |\hcR_j(\be)|}$.
    }
    Apply e-BH to $(\be_1^{\boost,\avcs},\ldots,\be^{\boost,\avcs}_m)$ and 
    obtain the rejection set $\cR^\ebh(\be^{\boost,\avcs})$.\;
    \KwOut{the rejection set $\cR^\ebh(\be^{\boost,\avcs})$.}
\end{algorithm}

\begingroup
\putbib
\endgroup
\end{bibunit}

\end{document}